\documentclass[reqno,11pt, fleqn]{amsart}
\usepackage{amsmath,euscript,amsthm,epsfig,amssymb,latexsym,amsfonts,graphicx}
\font\scs=cmcsc10 	
    \advance\textheight by \topskip

\def\PI{P$_{\mbox{\rm\scriptsize I}}$}
\def\PsmallI{P$_{\mbox{\rm\tiny I}}$}
\def\PII{P$_{\mbox{\rm\scriptsize II}}$}

\def\PItwo{P$_{\mbox{\rm\scriptsize I}}^2$}
\def\PsmallItwo{P$_{\mbox{\rm\tiny I}}^2$}

\def\Re{\mbox{\rm Re}\,}
\def\Im{\mbox{\rm Im}\,}

\def\scriptRe{{\scriptstyle{\rm Re}}\,}

\def\varkappa{\mbox{\msbm\char'173}}      

\theoremstyle{plain}
\newtheorem{thm}{Theorem}[section]

\newtheorem{RHP}{Riemann-Hilbert problem}
\newtheorem{con}[thm]{Conjecture}

\theoremstyle{definition}
\newtheorem{defn}{Definition}[section]

\theoremstyle{remark}
\newtheorem{rem}{Remark}[section]

\numberwithin{equation}{section}

\numberwithin{figure}{section}

\begin{document}
\thispagestyle{empty}

\title[On the tritronqu\'ee solutions of \PsmallItwo]
{On the tritronqu\'ee solutions of \PItwo}
\author{Tamara Grava}
\address{SISSA, Via Bonomea 265, 34136, Trieste, Italy and Department of Mathematics, Bristol University, UK}
\email{grava@sissa.it, tamara.grava@bristol.ac.uk}
\author{Andrei Kapaev}
\address{SPSU, Phys.\ Dept., Ulyanovskaya 3, 198504, 
St.~Petersburg, Russia}
\email{kapaev55@mail.ru}
\author{Christian Klein}
\address{Institut de Math\'ematiques de Bourgogne, Universit\'e 
de Bourgogne, 9 avenue Alain Savary, 21078 Dijon Cedex, France}
\email{Christian.Klein@u-bourgogne.fr}

\begin{abstract}
For equation \PsmallItwo, the second member in the \PsmallI\ hierarchy,
we prove existence of various degenerate solutions depending on 
the complex parameter $t$ and evaluate the asymptotics
in the complex $x$ plane for $|x|\to\infty$ and $t=o(x^{2/3})$. Using this result, we 
identify the most degenerate 
solutions $u^{(m)}(x,t)$, $\hat u^{(m)}(x,t)$, $m=0,\dots,6$, called 
{\em tritronqu\'ee}, describe the quasi-linear Stokes phenomenon and 
find the large $n$ asymptotics of the  coefficients in a formal 
expansion of these solutions. 
We supplement our findings by a numerical study of the tritronqu\'ee solutions.
\end{abstract}
\maketitle

\section{Introduction}

Equation \PItwo, the second member in the hierarchy of ODEs 
associated with the classical first Painlev\'e equation \PI, $y_{xx}=6y^2+x$, 
cf.\ \cite{I}, is the 4th order ODE
\begin{equation}\label{p12}
u_{xxxx}
+10u_x^2
+20uu_{xx}
+40(u^3-6tu+6x)=0
\end{equation}
depending on $t\in\mathbb{C}$ parametrically. In the last decades, this 
ODE has attracted significant attention \cite{BMP,CG,D1,D2} 
justified by its various  applications in  mathematics and physics.

An important  class of applications of \PItwo\ concerns the 
description of some critical regimes in random matrix models, 
as well as in the asymptotics of semi-classical orthogonal polynomials
and related Fredholm determinants, see e.g.\ \cite{CV, CIK}.
Furthermore a particular solution to the   \PItwo\  equation is conjectured to describe a   certain class of critical regimes to solutions of Hamiltonian 
PDEs \cite{Sul0, D1,D2}.  This conjecture is known as the 
\emph{universality conjecture} for Hamiltonian PDEs. So far it has been proved only for the 
Korteweg-de Vries equation   (KdV)
\begin{equation}\label{KdV}
u_t+uu_x+\tfrac{1}{12}u_{xxx}=0,
\end{equation}
and for its hierarchy \cite{CG, CG2}.
It is  known  that the KdV equation is compatible with   the \PItwo\ 
equation.  
Its solutions allow one to construct a 4-parameter 
family of so-called isomonodromic solutions to the KdV equation. 

It is a remarkable fact that, from the point of view of  physical
applications, the most interesting solutions to the classical Painlev\'e 
equations and their higher order analogs are those with a quasi-stationary 
behavior. For instance, applying equation \PItwo\ to string theory, 
Brezin,  Marinari, Parisi \cite{BMP} and Moore \cite{Moore} argued that, 
for $t=0$, there exists a regular solution $U_0(x,t=0)$ to (\ref{p12}) 
real on the real line with the asymptotic behavior 
$$U_0(x,t=0)\simeq \pm |6x|^{\frac{1}{3}},\quad x\rightarrow \pm \infty.$$

Dubrovin \cite{D1} conjectured the existence of the  solution 
$U_0(x,t)$, which has to be  pole free and regular on the real axis, for any $t\in{\mathbb R}$. 
The uniqueness 
of the real and regular on the real line solution to \PItwo\ for $t=0$ 
was proved in \cite{K}. The existence of such a solution was established by 
Claeys and Vanlessen in \cite{CV} for any $t$.

Below, we call  the solution $U_0(x,t)$  {\em tritronqu\'ee}. This term originates from 
the classical paper of P.~Boutroux \cite{Boutroux} devoted  to the 
asymptotic analysis of solutions to the \PI\ equation. In particular, he has shown 
that, though the generic asymptotic solutions to \PI\ are described
using the modulated elliptic Weierstra\ss $\wp$-function, 
there exist five special directions at infinity along which the elliptic 
asymptotics degenerates to trigonometric ones. According to Boutroux, 
such trigonometric asymptotic solutions are called ``tronqu\'ee''. 
``Bitronqu\'ee'' solutions are those 1-parameter solutions whose 
leading order algebraic asymptotic term, $y_{as}\sim\pm\sqrt{-x/6}$, 
admits an analytic continuation from the special ray into one of 
the adjacent complex sectors. ``Tritronqu\'ee'' solutions are 
particular 0-parameter solutions admitting an analytic continuation 
from the ray to the interior of both the adjacent complex sectors.
Remarkably, the latter asymptotic solutions remain quasi-stationary
in {\em four} of a total of five complex sectors separated by the above 
mentioned special rays.

In the present paper, we describe a set of similar solutions to 
equation \PItwo\ using   extensively the Riemann-Hilbert (RH) problem
approach. We prove that the physically interesting solution $U_0(x,t)$  real and
regular on the real line solution  has  an extension to the complex $x$ plane with  uniform 
algebraic asymptotics in the union of two sectors of the complex plane 
\begin{multline}\label{U0_sector}
U_0(x,t)\simeq-\sqrt[3]{6}\,x^{1/3},\quad
x\to\infty,
\\
\arg x\in
\bigl[-\tfrac{3\pi}{7}-\tfrac{3}{7}\arctan\tfrac{1}{\sqrt5},
\tfrac{3\pi}{7}+\tfrac{3}{7}\arctan\tfrac{1}{\sqrt5}\bigr]\cup
\\
\cup
\bigl[3\pi-\tfrac{3}{7}\arctan\tfrac{1}{\sqrt5},
3\pi+\tfrac{3}{7}\arctan\tfrac{1}{\sqrt5}\bigr],
\end{multline}
see Figure~\ref{fig_U0}.
\begin{figure}[htb]
\begin{center}
\mbox{\epsfig{figure=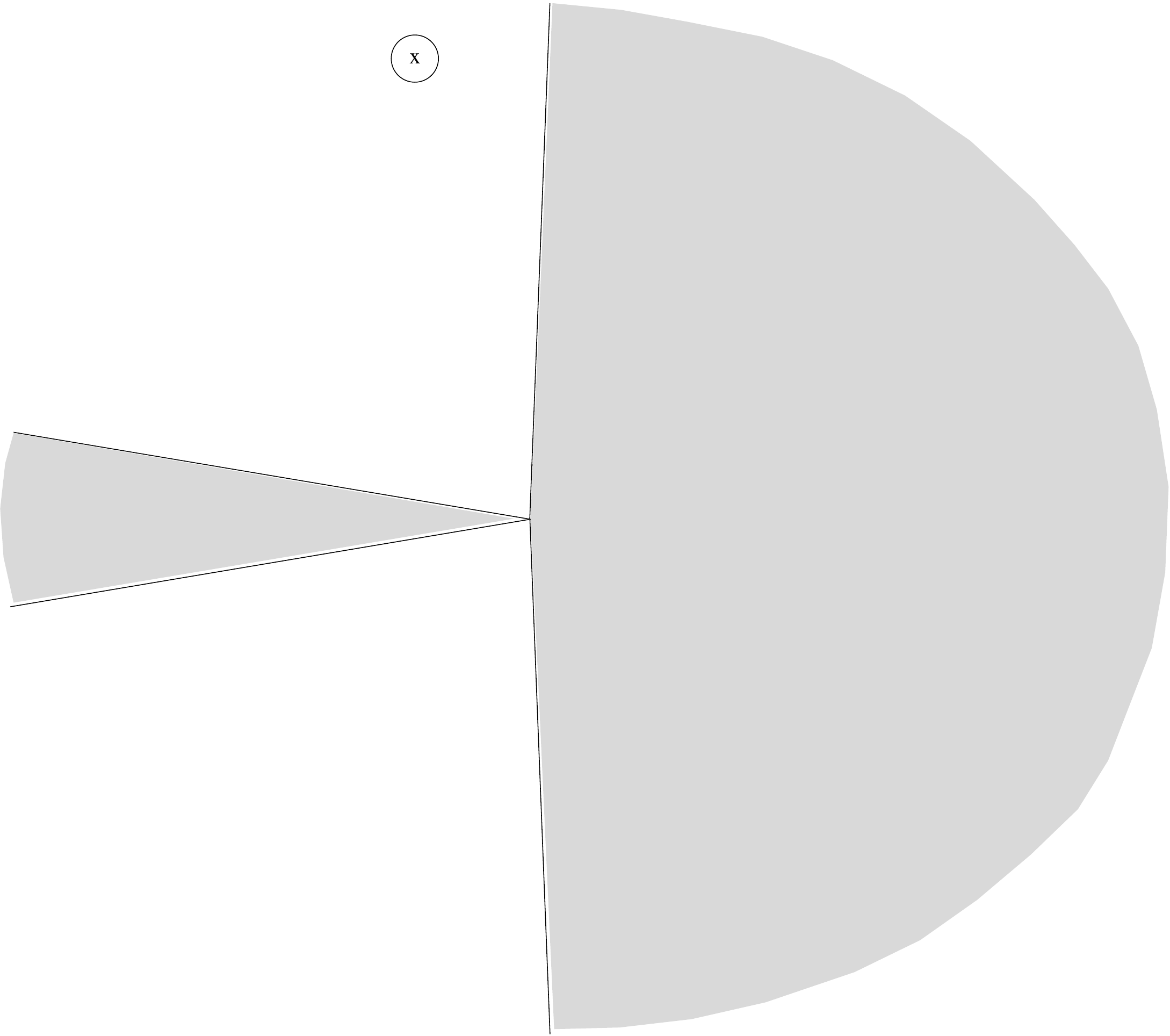,width=0.5\textwidth}}
\end{center}
\caption{Sector for the algebraic asymptotic behavior
of the solution $U_0(x)$ as $t=0$.}
\label{fig_U0}
\end{figure}

We also prove the existence and uniqueness of the 
solution $V_0(x,t)$ to \PItwo\ with the uniform algebraic 
asymptotics in the sector of the $x$ complex plane
\begin{multline}\label{V0_sector}
V_0(x,t)\simeq-\sqrt[3]{6}\,x^{1/3},\quad
x\to\infty,
\\
\shoveleft{
\arg x\in
\bigl[3\pi-\tfrac{6\pi}{7}+\tfrac{3}{7}\arctan\tfrac{1}{\sqrt5},
3\pi+\tfrac{6\pi}{7}-\tfrac{3}{7}\arctan\tfrac{1}{\sqrt5}\bigr],
}\hfill
\end{multline}
see Figure~\ref{fig_V0}.
\begin{figure}[htb]
\begin{center}
\mbox{\epsfig{figure=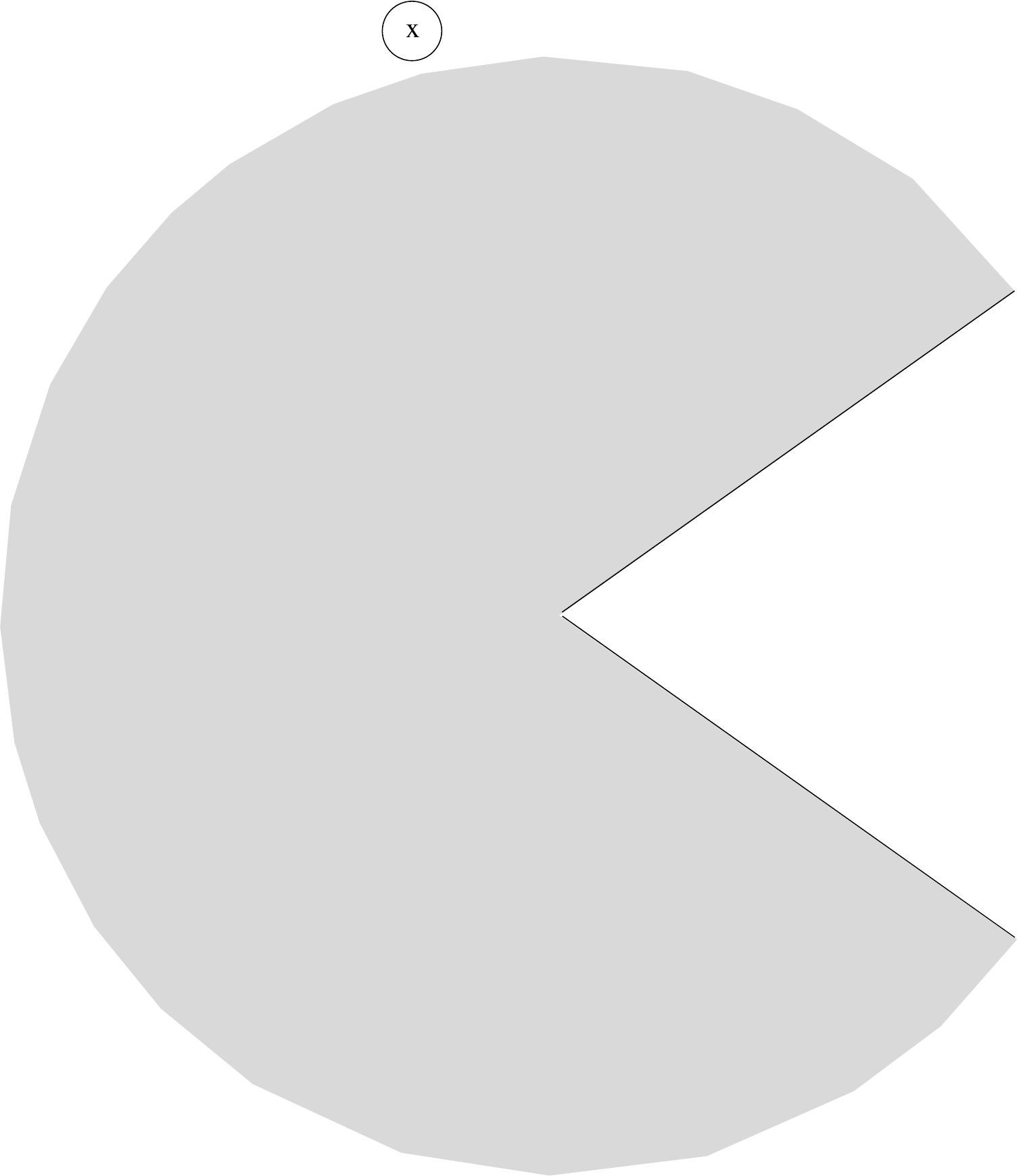,width=0.5\textwidth}}
\end{center}
\caption{Sector for the algebraic asymptotic behavior
of the solution $V_0(x)$ as $t=0$.}
\label{fig_V0}
\end{figure}

Similarly to $U_0(x,t)$, the solution $V_0(x,t)$ is real on the real line
but, in contrast to $U_0(x,t)$, it is singular on the positive part of the real line.
In the interior of the overlapping complex sector containing the negative real line,
the solutions $U_0(x,t)$ and $V_0(x,t)$ differ by exponentially small terms
while in two other overlapping complex sectors, $U_0(x,t)$ and $V_0(x,t)$
have different {\em leading} asymptotic behavior due to different 
choices of the branches of the cubic root of $x$.

With the results of Shimomura \cite{shimomura} on the Painlev\'e 
property of \PItwo, the solutions $U_0(x,t)$ and $V_0(x,t)$, whose existence
for large $x$ in the above sectors is proved below, extend to a globally defined
meromorphic function on the complex plane. We illustrate the behavior of 
$U_0(x,t)$, $V_0(x,t)$ numerically. The results can be summarized in 
the following conjecture:  
\begin{con}
    The solutions $U_{0}$ and $V_{0}$ are pole free in the whole 
    sectors of the complex plane specified in (\ref{U0_sector}) and 
    (\ref{V0_sector}) respectively. 
\end{con}

We show that $V_0(x,t)$ and $U_0(x,t)$ are the members of the families of
0-parameter solutions $\{V_m(x,t)\}_{m=0}^6$ and $\{U_m(x,t)\}_{m=0}^6$ 
which we call the tritronqu\'ee solutions of type~I and type~II, respectively. 
The solutions of each family can be constructed using a $7$-th order rotational 
symmetry applied to  $V_0(x,t)$ and $U_0(x,t)$, respectively.

In the overlapping domains where the above mentioned solutions
exhibit identical power series expansion, we compute the exponentially
small differences between them and thus establish the so-called 
quasi-linear Stokes phenomenon for equation \PItwo.

The structure of the paper is the following. In Section~2, we describe 
the RH problem implied by the linear system in the auxiliary ``spectral'' variable 
whose isomonodromy deformations are controlled by \PItwo. Section~3 is 
devoted to a short description of the ``spectral'' curve associated with 
\PItwo\ and to the model algebraic curve used to construct the asymptotic 
solution to \PItwo\ of our interest. In Section~4, assuming the triviality of 
some of the Stokes multipliers, we prove the asymptotic solvability 
of the RH problem as the pair $(x,t)$ belongs to a particular domain
of ${\mathbb C}^2$ and find 2-parameter families of the quasi-stationary 
solutions to \PItwo. In Section~5, we find 1- and 0-parameter intersections 
of the found 2-parameter families called below the bi- and tritronqu\'ee 
solutions and describe the quasi-linear analog of the Stokes phenomenon.
Finally, in Section~6, we present a numerical  study of 
the tritronqu\'ee solutions $V_0(x,t)$ and $U_0(x,t)$.

\section{Riemann-Hilbert problem associated with equation \PItwo}
\label{RHproblem}
Equation  \PItwo\ admits an isomonodromy interpretation.
Namely it can be expressed as the compatibility of 
a linear systems for a complex $2\times 2$ matrix valued function
$\Psi=\Psi(\lambda, x,t)$, cf.\ \cite{K,Cl},
\begin{align}
\label{A_p12}
&\Psi_{\lambda}=A\Psi,\\
\label{B_p12}
&\Psi_x=B\Psi,
\end{align}
where 
\begin{align}\nonumber
\begin{split}
A=& \tfrac{1}{60}\bigl[-u_x\lambda  -3uu_x-\tfrac{1}{4}u_{xxx}\bigr]\sigma_3+\tfrac{1}{30}\bigl[ \lambda^2+u\lambda+\tfrac{3}{2}u^2+\tfrac{1}{4}u_{xx}-15t\bigr]\sigma_+
\\
&\hskip-0.3in
+\tfrac{1}{30}
\bigl[
\lambda^3
-u\lambda^2
-(\tfrac{1}{2}u^2+\tfrac{1}{4}u_{xx}+15t)\lambda
+2u^3
-\tfrac{1}{4}u_x^2
+\tfrac{1}{2}uu_{xx}
+30x
\bigr]\sigma_-,
\end{split}
\\
\nonumber
B=& \begin{pmatrix}0&1\\
\lambda-2u&0\end{pmatrix}
\end{align}
and
$\sigma_3=(\begin{smallmatrix}1&0\\0&-1\end{smallmatrix})$,
$\sigma_+=(\begin{smallmatrix}0&1\\0&0\end{smallmatrix})$,
$\sigma_-=(\begin{smallmatrix}0&0\\1&0\end{smallmatrix})$.
Indeed the  compatibility condition of equations (\ref{A_p12}) and 
(\ref{B_p12}) gives  $A_x-B_{\lambda}+[A,B]=0$ which implies
the \PItwo\ equation (\ref{p12}), while the compatibility of 
(\ref{A_p12}) with the linear equation
\begin{equation}
\label{C_p12}
\begin{split}
\Psi_t&=C\Psi\\
C&=\tfrac{1}{3}\begin{pmatrix}
\tfrac{1}{2}u_x&-\lambda-u\\
-(\lambda^2-u\lambda-2u^2-\tfrac{1}{2}u_{xx})&-\tfrac{1}{2}u_x
\end{pmatrix}
\end{split}
\end{equation}
gives $A_t-C_{\lambda}+[A,C]=0$ which implies the KdV equation (\ref{KdV}).

The primary object of our study is equation (\ref{A_p12}) with 
the irregular singularity at $\lambda=\infty$.
The canonical solutions of (\ref{A_p12}) satisfying (\ref{B_p12}) and
(\ref{C_p12}) are uniquely characterized by their asymptotics
\begin{equation}
\label{Psi_k_def}
\begin{split}
\Psi_k(\lambda)&=\lambda^{-\frac{1}{4}\sigma_3}\dfrac{\sigma_3+\sigma_1}{\sqrt2}
\left(I +\dfrac{u}{2\lambda}\sigma_1
+\dfrac{i u_x}{4\lambda^{\frac{3}{2}}}\sigma_2
+{\mathcal O}(\lambda^{-2})\right)e^{(\theta(\lambda)
-\frac{H_1}{\sqrt{\lambda}}
-\frac{H_0}{3\lambda^{3/2}}
)\sigma_3},
\\
\theta(\lambda)&=
\tfrac{\lambda^{\frac{7}{2}}}{105}
-\tfrac{t}{3}\lambda^{\frac{3}{2}}
+x\lambda^{\frac{1}{2}},
\;\;
\lambda\to\infty,\quad
\arg\lambda\in
\bigl(
-\tfrac{3\pi}{7}+\tfrac{2\pi}{7}k,
\tfrac{\pi}{7}+\tfrac{2\pi}{7}k\bigr),\quad
k\in{\mathbb Z},
\end{split}
\end{equation}
where
 $\sigma_1=\sigma_++\sigma_-$ and $i\sigma_2=\sigma_+-\sigma_-$.
Here $H_0$ and $H_1$ are the functions of the coefficients $x,t,u$ and 
derivatives of $u$ related to the Hamiltonians associated with \PItwo,
\begin{equation}\label{H1}
H_1=xu
+\tfrac{1}{24}u^4
-\tfrac{1}{2}tu^2
+\tfrac{1}{24}uu_{x}^2
+\tfrac{1}{240}u_xu_{xxx}
-\tfrac{1}{480}u_{xx}^2,
\end{equation}
\begin{multline}\label{H0}
H_0=
\tfrac{1}{1920}u_{xxx}^2
+\tfrac{1}{80}uu_xu_{xxx}
+\tfrac{1}{16}u^2u_x^2
+\tfrac{1}{10}u^5
+\tfrac{1}{24}u^3u_{xx}
+\tfrac{1}{240}uu_{xx}^2
\\
-\tfrac{1}{480}u_x^2u_{xx}
-\tfrac{1}{4}u_x
+\tfrac{3}{2}xu^2
+\tfrac{1}{4}x u_{xx}
-tu^3
-\tfrac{1}{4}tuu_{xx}
+\tfrac{1}{8}tu_x^2.
\end{multline}
Observe that $(H_1)_x=u$ and $(H_0)_x=\frac{3}{2}u^2$. 
Note also that the differential system for the 
compatibility conditions (\ref{p12}), (\ref{KdV}) 
of the overdetermined system (\ref{A_p12})--(\ref{C_p12})
is equivalent to the Hamiltonian system in two time variables
(cf.\ \cite{Sul}),
\begin{multline}\nonumber
\frac{dq_j}{dt_k}=\frac{\partial{\mathcal H}_k}{\partial p_j},\quad
\frac{dp_j}{dt_k}=-\frac{\partial{\mathcal H}_k}{\partial q_j},\quad
k,j=1,2,\quad
t_1=x,\quad
t_2=t,
\\
q_1=u,\quad
p_1=\tfrac{1}{240}(u_{xxx}+8uu_x),\quad
q_2=\tfrac{1}{240}(u_{xx}+6u^2),\quad
p_2=u_x,
\\
{\mathcal H}_1=-H_1,\quad
{\mathcal H}_2=H_0.
\end{multline}

The ``ratios" of the canonical solutions
called the Stokes matrices,
\begin{multline}\nonumber
\Psi_{k+1}(\lambda)=\Psi_k(\lambda)S_k,\quad
S_{2k-1}=I+s_{2k-1}\sigma_+,\quad
S_{2k}=I+s_{2k}\sigma_-,
\end{multline}
are the first integrals of \PItwo\ (\ref{p12}) since 
they depend neither on $x$ nor on $t$. 
The Stokes multipliers satisfy a number of algebraic relations 
defining a 4-dimensional complex manifold,
\begin{equation}\label{scalar_cyclic_rel}
s_{k+7}=s_k,\quad
s_k+s_{k+2}+s_ks_{k+1}s_{k+2}=-i(1+s_{k+4}s_{k+5}),\quad
k\in{\mathbb Z}.
\end{equation}
Considering these quantities as the functions of the parameters 
$x,t,u$ and derivatives of $u$, we observe the rotational symmetry,
cf.\ \cite{K},
\begin{equation}\nonumber
x\mapsto\tilde x=e^{i\frac{2\pi}{7}n}x,\quad
t\mapsto\tilde t=e^{i\frac{6\pi}{7}n}t,\quad
u\mapsto\tilde u=e^{-i\frac{4\pi}{7}n}u,\quad
n\in{\mathbb Z},
\end{equation}
and thus
\begin{equation}\label{sk_rot_symm}
s_{k-2n}(
e^{i\frac{2\pi}{7}n}x,
e^{i\frac{6\pi}{7}n}t,
e^{-i\frac{4\pi}{7}n}u)=
s_k(x,t,u),\quad
n\in{\mathbb Z}.
\end{equation}
Another symmetry is related to the complex conjugation, 
\begin{equation}\nonumber
\overline{S_{-k}(\bar x,\bar t,\bar u)}=S_k^{-1}(x,t,u),\quad
\overline{s_{-k}(\bar x,\bar t,\bar u)}=-s_k(x,t,u).
\end{equation}

Now, we are prepared to formulate the RH
problem for the integration of the \PItwo\ equation, 
\begin{RHP}\label{initial_RHP}
Given the complex values of the parameters $x$, $t$ and $s_k$, 
$k\in{\mathbb Z}$, satisfying (\ref{scalar_cyclic_rel}), find the 
piece-wise holomorphic $2\times2$ matrix function $\Psi(\lambda)$ 
with the properties:
\begin{enumerate}
\item 
the limit
\begin{equation}\nonumber
\lim_{\lambda\to\infty}
\lambda^{1/2}
\bigl(
\tfrac{1}{\sqrt2}(\sigma_3+\sigma_1)
\lambda^{\sigma_3/4}\Psi(\lambda)
e^{-\theta\sigma_3}-I
\bigr),\quad
\theta=\tfrac{1}{105}\lambda^{\frac{7}{2}}
-\tfrac{1}{3}t\lambda^{\frac{3}{2}}
+x\lambda^{\frac{1}{2}},
\end{equation}
exists and is diagonal;
\item
at the origin, $\Psi(\lambda)$ is bounded;
\item 
on the union of the eight rays 
$\gamma=\rho\cup\bigl(\cup_{k=1}^7\gamma_{k-4}\bigr)$, where
$\gamma_k=\bigl\{\lambda\in{\mathbb C}\colon
\arg\lambda=\tfrac{2\pi}{7}k\bigr\}$, $k=-3,-2,\dots,2,3$, and 
$\rho=\bigl\{\lambda\in{\mathbb C}\colon
\arg\lambda=\pi\bigr\}$, all oriented towards infinity,
the following jump condition holds true,
\begin{equation}\nonumber
\Psi_+(\lambda)=\Psi_-(\lambda)S(\lambda),
\end{equation}
where $\Psi_+(\lambda)$ and $\Psi_-(\lambda)$ are limits of 
$\Psi(\lambda)$ on $\gamma$ from the left and from the right,
respectively, and where the piece-wise constant matrix $S(\lambda)$ 
is given by the following equations,
\begin{subequations}\label{jump_matrices}
\begin{align}\nonumber
&S(\lambda)\bigr|_{\lambda\in\gamma_k}=S_k,\quad
S_{2k}=I+s_{2k}\sigma_-,\quad
S_{2k-1}=I+s_{2k-1}\sigma_+,
\\
\nonumber
&S(\lambda)\bigr|_{\rho}=i\sigma_1.
\end{align}
\end{subequations}
\end{enumerate}
\end{RHP}

\begin{figure}[htb]
\begin{center}
\mbox{\epsfig{figure=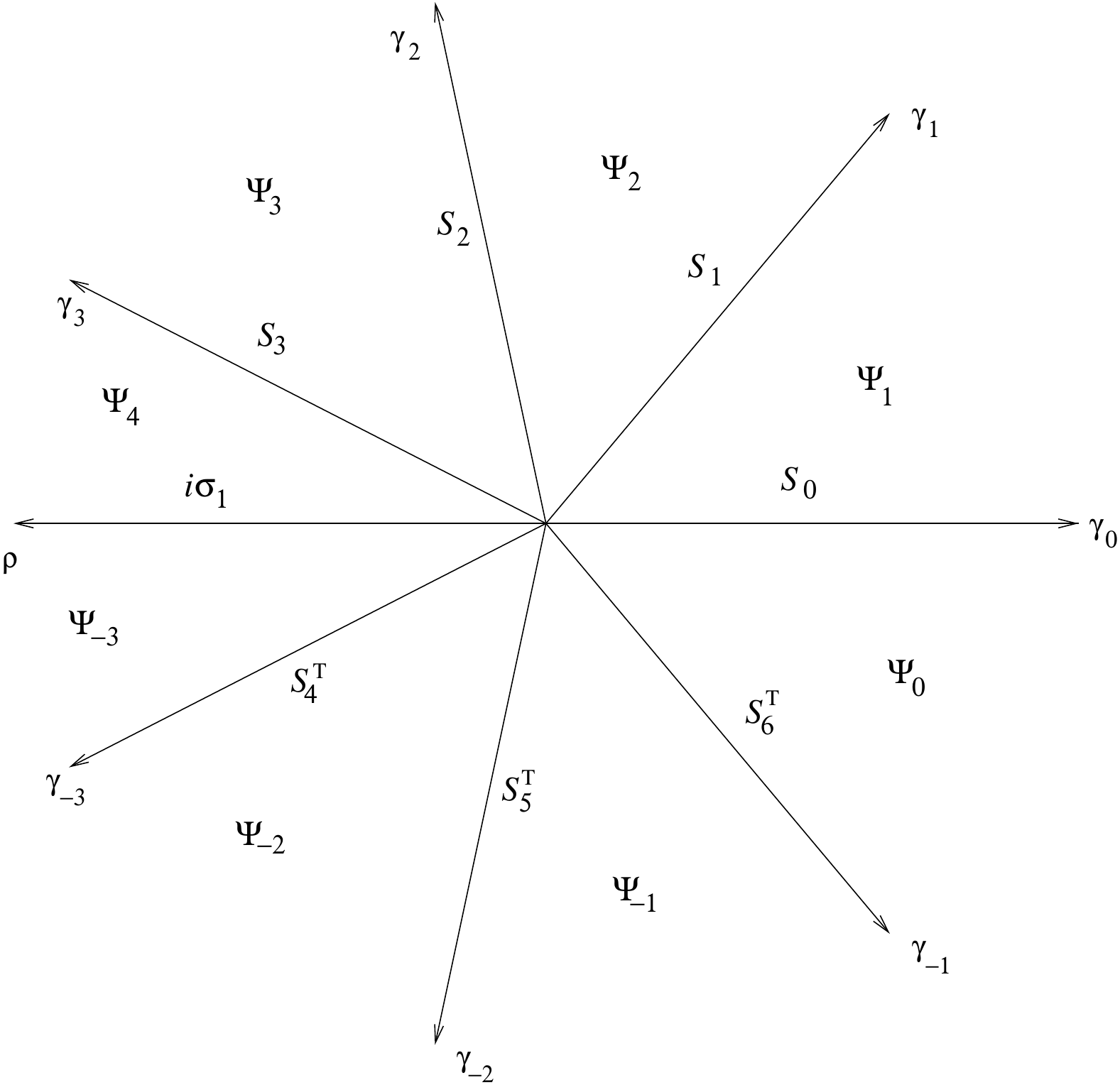,width=0.6\textwidth}}
\end{center}
\caption{The jump contour $\gamma$ for the RH problem~\ref{initial_RHP}
and canonical solutions $\Psi_j(\lambda)$, $j=-3,-2,\dots,3,4$.}
\label{fig1}
\end{figure}

The solution of the RH problem~\ref{initial_RHP}, if it  exists,  is unique. 
Having the solution $\Psi(\lambda)$ of the RH problem~\ref{initial_RHP}, 
the value of the Painlev\'e function $u$ follows from the asymptotics
(\ref{Psi_k_def}),
\begin{equation}\nonumber
u=2\lim_{\lambda\to\infty}
\lambda
\bigl(
\tfrac{1}{\sqrt2}(\sigma_3+\sigma_1)
\lambda^{\sigma_3/4}\Psi(\lambda)
e^{-\theta\sigma_3}
\bigr)_{12}.
\end{equation}

\section{``Spectral'' and model curves}
\label{discriminant_set}

The spectral curve is defined by the characteristic 
equation, $\det(\mu-A(\lambda))=0$. Explicitly, it takes the form
\begin{equation}
\label{spectral_curve_p12_as}
\begin{split}
30\mu^2&=
\tfrac{1}{30}\lambda^5
-t\lambda^3
+x\lambda^2
+(H_1+\frac{15}{2}t^2)\lambda
+(H_0-15tx+\frac{u_x}{4})
=\tfrac{1}{30}
\prod_{k=1}^5(\lambda-\lambda_k),
\end{split}
\end{equation}
where $H_1$ and $H_0$ are defined in (\ref{H1}) and (\ref{H0}) respectively.

The {\em model} algebraic curve used below is a special case
of (\ref{spectral_curve_p12_as}) with two double and 
one simple branch point, cf.\ \cite{K, Cl},
\begin{equation}\label{2+2_deg}
\lambda_1=\lambda_2\neq\lambda_3=\lambda_4\neq\lambda_5,\quad
\mbox{i.e.}\quad
\mu^2=\tfrac{1}{900}
(\lambda-\lambda_1)^2
(\lambda-\lambda_3)^2
(\lambda-\lambda_5).
\end{equation}
Identifying the leading order coefficients in 
(\ref{spectral_curve_p12_as}) and (\ref{2+2_deg}),
we find the conditions,
\begin{multline}\label{t_x_from_lambda_k}
2\lambda_1+2\lambda_3+\lambda_5=0,\quad
\lambda_1^2+4\lambda_1\lambda_3+2\lambda_1\lambda_5
+2\lambda_3\lambda_5+\lambda_3^2=-30t,
\\
2\lambda_1^2\lambda_3
+2\lambda_1\lambda_3^2
+4\lambda_1\lambda_3\lambda_5
+\lambda_1^2\lambda_5
+\lambda_3^2\lambda_5
=-30x.
\end{multline}
Thus the double branch points $\lambda_1$ and $\lambda_3$
satisfy the quadratic equation,
\begin{equation}\label{lambda13_eq_sol}
\lambda^2+\tfrac{1}{2}\lambda_5\lambda+\tfrac{3}{8}\lambda_5^2-15t=0,
\quad
\mbox{i.e.}
\quad
\lambda_{1,3}=
\tfrac{1}{4}
\bigl(
-\lambda_5
\pm i\sqrt{5}\sqrt{\lambda_5^2-48t}
\bigr),
\end{equation}
while the simple branch point $\lambda_5$ satisfies the cubic equation,
\begin{equation}\label{lambda5_eq}
\lambda_5^3
-24t\lambda_5
+48x=0.
\end{equation}
Therefore the model algebraic curve 
(\ref{spectral_curve_p12_as})--(\ref{2+2_deg}) is determined 
by the values of $t$ and $x$ up to a discrete ambiguity in 
the determination of $\lambda_5$ via (\ref{lambda5_eq}).

\section{Asymptotic solution of the reduced RH problems}\label{RHP_deg}

\subsection{0-parameter reduced RH problems}

\begin{thm}\label{phi=0_red_RHP12_solvability}
For the Stokes multipliers
\begin{equation}\label{spm2_spm1=0}
s_{\pm2}=s_{\pm1}=0,\quad
s_{\pm3}=s_0=-i,
\end{equation}
there exists a closed domain $\omega_0\subset{\mathbb C}^2$
such that the RH problem~\ref{initial_RHP} is solvable
for $\forall(x,t)\in\omega_0$. As $t=0$, the domain $\omega_0$
is the sector
\begin{equation}\nonumber
t=0\colon\quad
\omega_0=\bigl\{
x\in{\mathbb C}\colon
|x|>\rho_0,\quad
\arg x\in[-\alpha_0,\alpha_0]
\bigr\},\quad
\end{equation}
where $\alpha_0=\tfrac{3\pi}{7}-\tfrac{3}{7}\arctan\tfrac{1}{\sqrt5}$
and the positive constant $\rho_0$ is large enough. 
If $x\neq\pm2\sqrt3t^{3/2}$, and if all the roots 
of the cubic polynomial $P_3(v_0):=v_0^3-6tx^{-2/3}v_0+6$ are simple, 
then the relevant solution $u^{(0)}(x,t)$ of equation {\em\PItwo} has the 
asymptotics
\begin{equation}\label{u_s12=0_as}
u^{(0)}(x,t)=x^{1/3}v_0
+{\mathcal O}(x^{-5/6}),\quad
v_0^3-6tx^{-2/3}v_0+6=0,
\end{equation}
where the root $v_0=v_0(tx^{-2/3})$ is chosen in such a way that
\begin{equation*}
u^{(0)}(x,t)\simeq-\sqrt[3]{6}|x|^{1/3},\quad
t\to0,\quad
x\to+\infty.
\end{equation*}
\end{thm}

\begin{thm}\label{phi=3pi_red_RHP32_solvability}
For the Stokes multipliers
\begin{equation}\label{spm3_spm2=0}
s_{\pm3}=s_{\pm2}=0,\quad
s_{\pm1}=s_0=-i,
\end{equation}
there exists a closed domain $\hat\omega_0\subset{\mathbb C}^2$
such that the RH problem~\ref{initial_RHP} is solvable
for $\forall(x,t)\in\hat\omega_0$. As $t=0$, the domain 
$\hat\omega_0$ satisfies
\begin{equation}\nonumber
t=0\colon\quad
\hat\omega_0=\bigl\{
x\in{\mathbb C}\colon
|x|>\hat\rho_0,\quad
\arg x\in[3\pi-\beta_0,3\pi+\beta_0]
\bigr\},\quad
\end{equation}
where $\beta_0=\tfrac{3}{7}\arctan\tfrac{1}{\sqrt5}$
and the positive constant $\hat\rho_0$ is sufficiently large. 
If $x\neq\pm2\sqrt3t^{3/2}$, and if all the roots of the cubic 
polynomial $P_3(v_0):=v_0^3-6tx^{-2/3}v_0+6$ are simple, 
then the relevant solution $\hat u^{(0)}(x,t)$ of equation 
{\em\PItwo }\ has the asymptotics
\begin{equation}\label{u_s32=0_as}
\hat u^{(0)}(x,t)=x^{1/3}v_0
+{\mathcal O}(x^{-5/6}),\quad
v_0^3-6tx^{-2/3}v_0+6=0,
\end{equation}
where the root $v_0=v_0(tx^{-2/3})$ is chosen in such a way that
\begin{equation*}
\hat u^{(0)}(x,t)\simeq\sqrt[3]{6}|x|^{1/3},\quad
t\to0,\quad
x\to-\infty.
\end{equation*}
\end{thm}

\begin{proof}
The proofs of both theorems~\ref{phi=0_red_RHP12_solvability}
and~\ref{phi=3pi_red_RHP32_solvability} are almost identical
and follow the line explained in \cite{kap_p1_quasi, FIKN}.
Introduce the Wronski matrix
\begin{equation}\nonumber
Z_0(z)=
\begin{pmatrix}
v_2(z)&v_1(z)\\
v_2'(z)&v_1'(z)
\end{pmatrix}
,\quad
v_1(z)=\sqrt{2\pi}\,\mbox{\rm Ai}(z),\quad
v_2(z)=\sqrt{2\pi}\,e^{i\frac{\pi}{6}}\mbox{\rm Ai}(e^{i\frac{2\pi}{3}}z),\quad
\end{equation}
with $\mbox{\rm Ai}(z)$ standing for the classical Airy function satisfying
the asymptotic condition
\begin{equation}\nonumber
\mbox{\rm Ai}(z)=
\tfrac{1}{2\sqrt{\pi}}
z^{-1/4}e^{-\frac{2}{3}z^{3/2}}
(1+{\mathcal O}(z^{-3/2})),\quad
z\to\infty,\quad
\arg z\in(-\pi,\pi).
\end{equation}
Define also the functions $Z_k(z)$, $k=-1,1,2$,
\begin{multline}\nonumber
Z_{-1}(z)=Z_0(z)G_{-1}^{-1},\quad
Z_1(z)=Z_0(z)G_0,\quad
Z_2(z)=Z_1(z)G_1,\quad
\\
\shoveleft{
G_1=G_{-1}=I-i\sigma_-,\quad
G_0=I-i\sigma_+.
}\hfill
\end{multline}
Using the Stokes phenomenon of the Airy function described 
in \cite{BE}, all the introduced above matrix functions have 
the asymptotics
\begin{multline}\label{Zk_as}
Z_k(z)=z^{-\sigma_3/4}
\tfrac{1}{\sqrt2}(\sigma_3+\sigma_1)
\bigl(I+{\mathcal O}(z^{-3/2})\bigr)
e^{\frac{2}{3}z^{3/2}\sigma_3},
\\
z\to\infty,\quad
z\in\omega_k=\bigl\{
z\in{\mathbb C}\colon
\arg z\in\bigl(-\pi+\tfrac{2\pi}{3}k,
\tfrac{\pi}{3}+\tfrac{2\pi}{3}k\bigr)
\bigr\}.
\end{multline}
The piece-wise holomorphic function $Z(z)$,
\begin{equation}\label{Z_RH_def}
Z(z)=
\begin{cases}
Z_{-1}(z),\quad
\arg z\in(-\pi,-\tfrac{2\pi}{3}),
\\
Z_0(z),\quad
\arg z\in(-\tfrac{2\pi}{3},0),
\\
Z_1(z),\quad
\arg z\in(0,\tfrac{2\pi}{3}),
\\
Z_2(z),\quad
\arg z\in(\tfrac{2\pi}{3},\pi),
\end{cases}
\end{equation}
has the uniform asymptotics (\ref{Zk_as})
as $z\to\infty$, is discontinuous across the oriented rays $\arg z=\tfrac{2\pi}{3}k$, $k=0,\pm1$, and 
$\arg z=\pi$ towards 
infinity and satisfies by definition the jump conditions,
\begin{multline}\nonumber
Z_+(z)=Z_-(z)G_k,\quad
\arg z=\tfrac{2\pi}{3}k,\quad
k=0,\pm1,
\\
G_{0}=I-i\sigma_-,\quad
G_1=G_{-1}=I-i\sigma_+,
\\
\shoveleft{
Z_+(z)=Z_-(z)i\sigma_1,\quad
\arg z=\pi,
}\hfill
\end{multline}
see Figure~\ref{fig6}.
\begin{figure}[htb]
\begin{center}
\mbox{\epsfig{figure=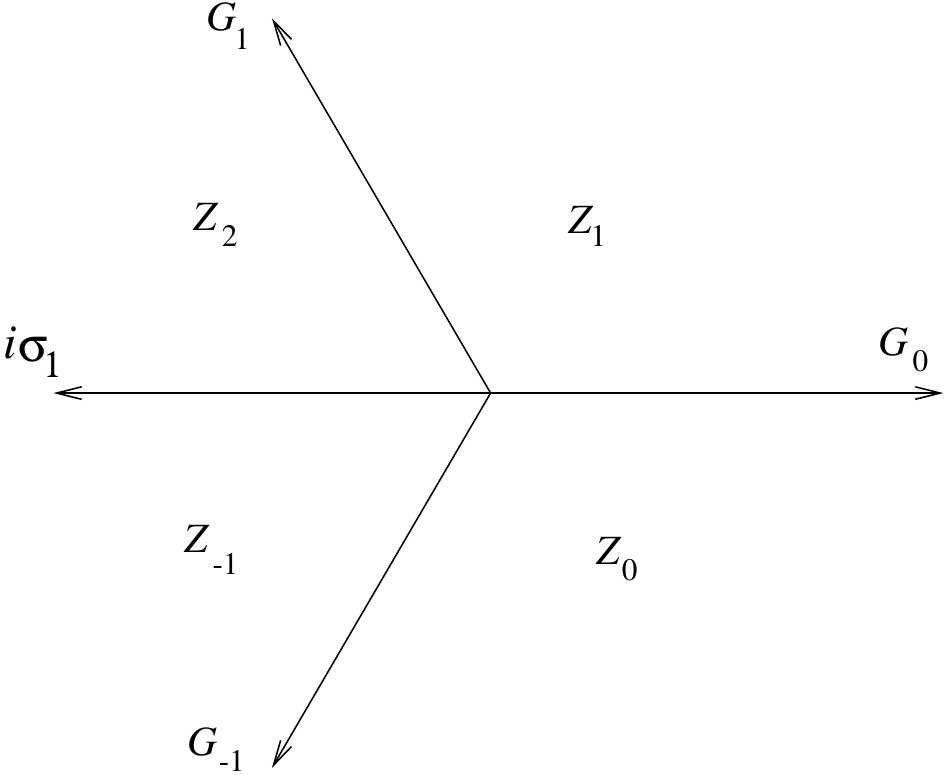,width=0.5\textwidth}}
\end{center}
\caption{The jump contour for the RH problem solved by
the collection of Airy functions.}
\label{fig6}
\end{figure}

On the complex  $\lambda$ plane cut along $(-\infty, \lambda_5)$,
define the function $F(\lambda)$ as the integral of  the model curve (\ref{2+2_deg}),
\begin{equation}\label{F_def}
F(\lambda):=\int_{\lambda_5}^{\lambda}\mu(\xi)d\xi=
\tfrac{1}{105}(\lambda-\lambda_5)^{\frac{7}{2}}
+\tfrac{1}{30}\lambda_5(\lambda-\lambda_5)^{\frac{5}{2}}
+(\tfrac{1}{24}\lambda_5^2-\tfrac{1}{3}t)
(\lambda-\lambda_5)^{\frac{3}{2}}.
\end{equation}
As $\lambda\to\infty$, this function has the asymptotics
\begin{multline}\nonumber
F(\lambda)=
\tfrac{1}{105}\lambda^{\frac{7}{2}}
-\tfrac{1}{3}t\lambda^{\frac{3}{2}}
+x\lambda^{\frac{1}{2}}
+F_{-1}\lambda^{-\frac{1}{2}}
+\tfrac{1}{3}F_{-2}\lambda^{-\frac{3}{2}}
+{\mathcal O}(\lambda^{-\frac{5}{2}}),
\\
F_{-1}=\tfrac{1}{16}\lambda_5(t\lambda_5-6x),\quad
F_{-2}=\tfrac{1}{40}\lambda_5^2(2t\lambda_5-9x).
\end{multline}

The conditions $\lambda_5^2-8t=0$ and $\lambda_5^2-48t=0$ correspond
to a triple and a quadruple branch point, respectively. For
\begin{equation}\label{non-degeneracy_conditions}
(\lambda_5^2-8t)(\lambda_5^2-48t)\neq0,
\end{equation}
define the mapping 
$\lambda\mapsto z(\lambda)$ by the equation,
\begin{multline}\label{z_zeta_mapping}
z(\lambda)=\bigl(\tfrac{3}{2}F(\lambda)\bigr)^{2/3}
=(\lambda-\lambda_5)
2^{-\frac{8}{3}}(\lambda_5^2-8t)^{\frac{2}{3}}
\times
\\
\times
\bigl(
1
+\tfrac{4\lambda_5}{5(\lambda_5^2-8t)}(\lambda-\lambda_5)
+\tfrac{8}{35(\lambda_5^2-8t)}(\lambda-\lambda_5)^2
\bigr)^{\frac{2}{3}},
\end{multline}
In the disc 
$|\lambda-\lambda_5|\leq R|x|^{1/3}<\min_{j=1,3}|\lambda_j-\lambda_5|$,
the mapping (\ref{z_zeta_mapping}) is bi-holomorphic.
Assuming that (\ref{non-degeneracy_conditions}) holds true,
define the piece-wise holomorphic function $\tilde\Psi(\lambda)$,
\begin{multline}\label{tilde_Psi_def}
\tilde\Psi(\lambda)=
\begin{cases}
(\lambda-\lambda_5)^{-\frac{1}{4}\sigma_3}
z^{\frac{1}{4}\sigma_3}Z(z(\lambda)),\quad
|\lambda-\lambda_5|<R|x|^{1/3},
\\
(\lambda-\lambda_5)^{-\frac{1}{4}\sigma_3}
\tfrac{1}{\sqrt2}(\sigma_3+\sigma_1)
e^{F(\lambda)\sigma_3},\quad
|\lambda-\lambda_5|>R|x|^{1/3},
\end{cases}
\end{multline}
where $z=z(\lambda)$ is as in (\ref{z_zeta_mapping}), and where
the root $(\lambda-\lambda_5)^{1/4}$ is defined on the plane cut along 
the level line $\Re F(\lambda)=0$ asymptotic to the ray $\arg\lambda=\pi$. 
Definition (\ref{tilde_Psi_def}) together with (\ref{Z_RH_def}),
(\ref{Zk_as}) and (\ref{z_zeta_mapping}) implies that, across the 
circle $|\lambda-\lambda_5|=R|x|^{1/3}$ with clock-wise orientation, the model function 
$\tilde\Psi(\lambda)$ has the jump
\begin{equation}\label{tilde_Psi_at_circle_jump}
\tilde\Psi_-(\lambda)\tilde\Psi_+^{-1}(\lambda)=
I+{\mathcal O}(x^{-1}\sigma_-)
+{\mathcal O}(x^{-7/6}I)
+{\mathcal O}(x^{-7/6}\sigma_{3})
+{\mathcal O}(x^{-4/3}\sigma_+).
\end{equation}

We look for the solution $\Psi^{(0)}(\lambda)$ of the RH 
problem~\ref{initial_RHP} in the form of the product
\begin{equation}\label{chi_def}
\Psi^{(0)}(\lambda)=
(I-\tilde H_1\sigma_-)
\chi(\lambda)\tilde\Psi(\lambda).
\end{equation}
Denote by $u^{(0)}$ the solution of \PItwo\ corresponding 
to $\Psi^{(0)}(\lambda)$, and  by $H_1^{(0)}$, $H_0^{(0)}$ 
the Hamiltonian functions (\ref{H1}) and (\ref{H0}) evaluated 
at $u^{(0)}$. In terms of the above introduced functions,
the asymptotics of $\chi(\lambda)$ as $\lambda\to\infty$ is 
given by
\begin{multline}\label{chi_as}
\chi(\lambda)=
I
+\lambda^{-1}
\bigl[
\tfrac{1}{2}(
u^{(0)}
-\tfrac{1}{2}\lambda_5
+\tilde H_1^2
)\sigma_3
-\tilde H_1\sigma_+
\\
+(
u^{(0)}\tilde H_1
+\tfrac{1}{4}u_x^{(0)}
-\tfrac{1}{3}\tilde H_0
+\tfrac{1}{2}\tilde H_1^3
)\sigma_-
\bigr]
+{\mathcal O}(\lambda^{-3/2}),
\\
\tilde H_1=H_1^{(0)}+F_{-1},\quad
\tilde H_0=H_0^{(0)}+F_{-2}.
\end{multline}

Combining these formulas, the function 
$\chi(\lambda)$ satisfies the RH problem,
\begin{description}
\item[(i)]
$\chi(\lambda)\to I$ as $\lambda\to\infty$;
\item[(ii)]
$\chi(\lambda)$ is discontinuous across the oriented contour $\ell$ 
in Figure~\ref{fig7}, moreover
\begin{equation}\label{chi_jumps}
\chi_+(\lambda)=\chi_-(\lambda){\mathcal S}(\lambda),\quad
\lambda\in\ell,\quad
{\mathcal S}(\lambda)=
\tilde\Psi_-(\lambda)S\tilde\Psi_+^{-1}(\lambda).
\end{equation}
\end{description}
\begin{figure}[htb]
\begin{center}
\mbox{\epsfig{figure=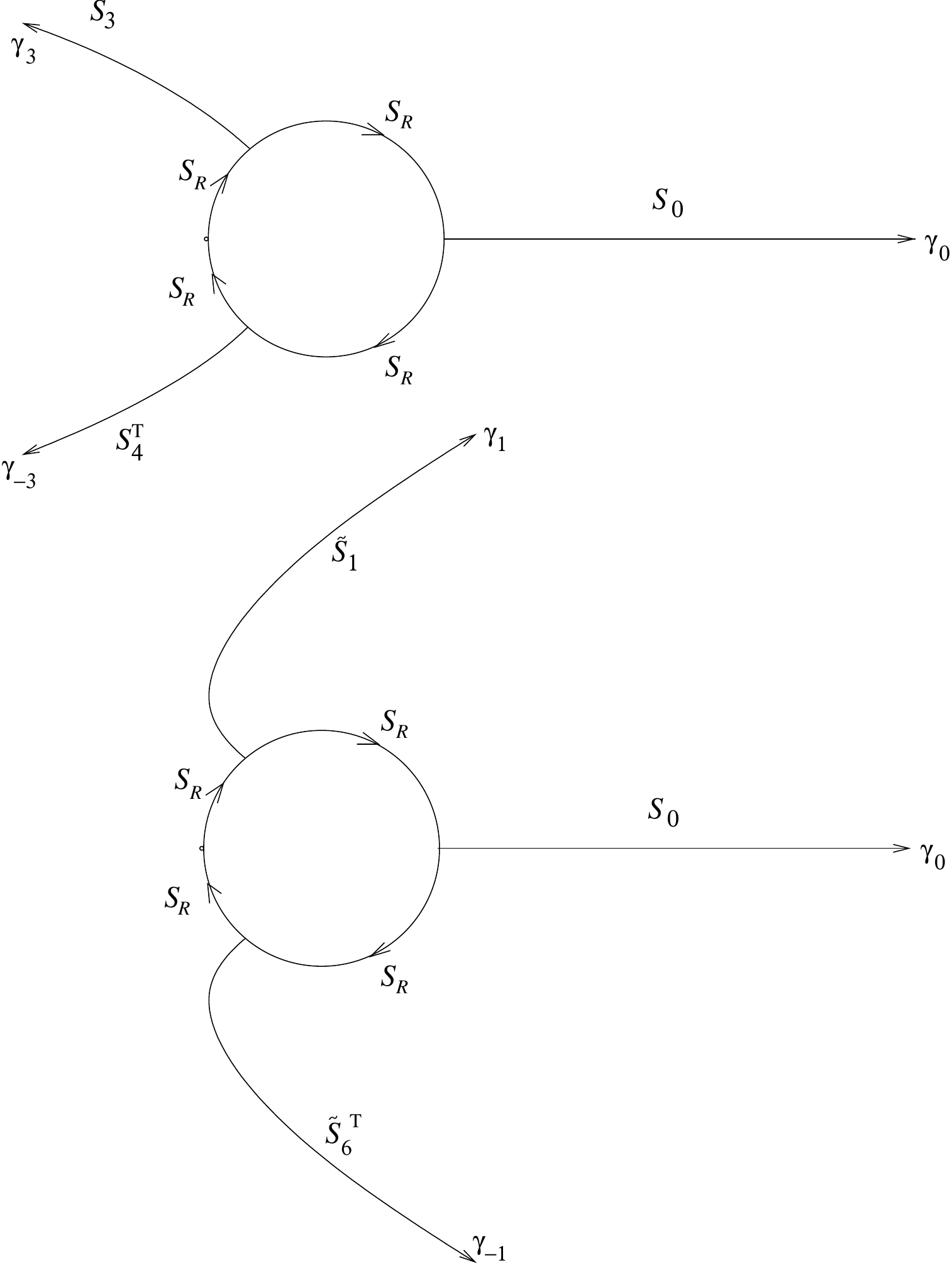,width=0.5\textwidth}}
\end{center}
\caption{Jump contour for the correction function $\chi(\lambda)$
in theorem~\ref{phi=0_red_RHP12_solvability} (above)
and in theorem~\ref{phi=3pi_red_RHP32_solvability}
(below).}
\label{fig7}
\end{figure}

Observe first of all that $\Re F(\lambda)\to+\infty$ as 
$\lambda\to\infty$, $\lambda\in\gamma_0$, 
and $\Re F(\lambda)\to-\infty$ as $\lambda\to\infty$, 
$\lambda\in\gamma_k$, $k=\pm3,\pm1$. Therefore the jump matrices
${\mathcal S}(\lambda)$ in (\ref{chi_jumps}) are exponentially 
small as $\lambda\to\infty$. This fact however does not guarantee
that the jumps of $\chi(\lambda)$ are {\em uniformly} small as
$\lambda\in\ell$ and $|x|\to\infty$. Impose the following

\smallskip
{\scs Condition~1}: the infinite branches of the jump graphs
in Figure~\ref{fig7} are such that $\Re F(\lambda)$ 
is {\em strictly} monotonic along each infinite branch.

\begin{rem}
The case of non-strictly monotonic $\Re F(\lambda)$ is more
involved and will be considered in the next section.
\end{rem}

Under Condition~1, the jump matrices across the infinite 
branches of the jump graph are estimated uniformly in $\lambda$ 
as follows,
\begin{equation}\label{chi_gamma3_jumps_estim}
\lambda\in\gamma_{0,\pm3,\pm1},
\
|\lambda-\lambda_5|>R|x|^{1/3}
\colon\
\|{\mathcal S}_{j}(\lambda)-I\|\leq
C_1e^{-c_2|x|^{2/3}\cdot|\lambda-\lambda_5|^{3/2}}
|\lambda-\lambda_5|^{1/2},
\end{equation}
with some constant $C_1>0$. The constant $c_2>0$
is determined by the least value of $|\Re F(\lambda)|$
on the infinite branch of the jump graph,
\begin{equation*}
c_2=
2R^{-3/2}
\min_{\lambda\in\gamma_k}|x^{-7/6}\Re F(\lambda)|,\quad
k=0,\pm3,\pm1.
\end{equation*}

The matrix ${\mathcal S}_R(\lambda)$ across the circle
$|\lambda-\lambda_5|=R|x|^{1/3}$ is estimated using
(\ref{tilde_Psi_at_circle_jump}),
\begin{multline}\nonumber
\lambda\in C_R\colon\quad
|({\mathcal S}_R(\lambda))_{21}|\leq
C_{21}|x|^{-1},
\\
|({\mathcal S}_R(\lambda)_{jj}-1|\leq
C_{jj}|x|^{-7/6},\quad
|({\mathcal S}_R(\lambda))_{12}|\leq
C_{12}|x|^{-4/3},
\end{multline}
with some constants $C_{ij}>0$.

Estimates (\ref{chi_gamma3_jumps_estim}) and 
(\ref{tilde_Psi_at_circle_jump}) imply the solvability of 
the RH problem for large enough $|x|$ in some domains of 
${\mathbb C}^2\ni(x,t)$, see \cite{FIKN}, consistent with the 
{\scs Condition}~1. Now, we are going to construct 
certain {\em closed simply connected domains\ } $\omega_0$ 
and $\hat\omega_0$ consistent with {\scs Condition}~1.

At $t=0$, using (\ref{lambda5_eq}) 
and (\ref{lambda13_eq_sol}),
\begin{equation}\nonumber
\lambda_5=-2\sqrt[3]{6}x^{1/3},\quad
\lambda_{1,3}=-\tfrac{1}{4}(1\mp i\sqrt5)\lambda_5.
\end{equation}
Therefore, if $t=0$ and $\arg x=0$, the configuration of 
the steepest descent lines $\Im F(\lambda)=0$ coincides with the 
configuration of the infinite tails in the jump graph shown
in Figure~\ref{fig7} (above), and {\scs Condition}~1 is satisfied.
In what follows, it is important that $\Re F(\lambda_1)>0$
and $\Re F(\lambda_3)>0$ for $t=0$ and $\arg x=0$.

If $t=0$ and 
\begin{equation}\nonumber
\arg x=\tfrac{3}{7}\arctan\tfrac{1}{\sqrt5}
\end{equation}
then $\Im F(\lambda_1)=0$, i.e.\ the steepest descent path 
$\Im F(\lambda)\equiv0$ meets the double branch point $\lambda_1$ 
and breaks down at this point.
Nevertheless, since $\Re F(\lambda)$ remains strictly monotonic along it, 
the jump curve $\gamma_0$ can be chosen as the steepest descent line
$\Im F(\lambda)\equiv0$, and {\scs Condition}~1 is still satisfied.

Further increase of $\arg x$ forces the splitting of the steepest descent 
path $\Im F(\lambda)\equiv0$ into two branches. One of these branches emanates 
from a point of the circle $C_R$ and approaches the ray 
$\arg\lambda=-\tfrac{4\pi}{7}$, while another branch comes from 
infinity at the direction $\arg\lambda=-\tfrac{4\pi}{7}$, passes 
through $\lambda_1$ and then approaches the ray $\arg\lambda=0$.

As the result, the continuous jump curve $\gamma_0$ cannot coincide with the
steepest descent path. Nevertheless, taking into account that
$\Re F(\lambda_5)=0$ and $\Re F(\lambda)\to+\infty$ as $\lambda\to+\infty$,
it is possible to chose the continuous line $\gamma_0$ satisfying {\scs Condition}~1 
as soon as the inequality $\Re F(\lambda_1)>0$ holds. Thus the direction
\begin{equation}\nonumber
\arg x=\tfrac{3\pi}{7}+\tfrac{3}{7}\arctan\tfrac{1}{\sqrt5}
\quad
\mbox{corresponding to}
\quad
\Re F(\lambda_1)=0,
\end{equation}
gives us an upper bound for the sector of directions consistent with
{\scs Condition}~1 and hence in accordance with the solvability of
the RH problem~\ref{initial_RHP} with the Stokes data 
(\ref{spm2_spm1=0}). 

Using similar considerations for the negative values of $\arg x$, 
we prove the solvability of the RH problem~\ref{initial_RHP} with 
the Stokes data (\ref{spm2_spm1=0}) in the closed sector
\begin{equation}\nonumber
\arg x\in[-\tfrac{3\pi}{7}-\tfrac{3}{7}\arctan\tfrac{1}{\sqrt5}+\epsilon,
\tfrac{3\pi}{7}+\tfrac{3}{7}\arctan\tfrac{1}{\sqrt5}-\epsilon],\quad
|x|>\rho_0(\epsilon),\quad
\epsilon>0.
\end{equation}

Observing that $\arg x=\tfrac{3\pi}{7}-\tfrac{3}{7}\arctan\tfrac{1}{\sqrt5}$
corresponds to the equality $\Re F(\lambda_3)=0$, it is more convenient 
to restrict ourselves to the less wide sector $\omega_0$ of the complex $x$ plane,
\begin{equation}\label{Re_int>=0_around_phi=0}
\omega_0|_{t=0}=\{x\in\mathbb{C}\colon\quad
|x|>\rho_0,\quad
\arg x\in[-\alpha_0,\alpha_0]\},\quad
\alpha_0=\tfrac{3\pi}{7}-\tfrac{3}{7}\arctan\tfrac{1}{\sqrt5}.
\end{equation}
(The reason for the above made choice (\ref{Re_int>=0_around_phi=0}) of 
$\omega_0$ will become clear in Section~\ref{um_proliferation}.)
The above definition of $\omega_0$ admits a straightforward extension 
to $t\neq0$,
\begin{equation}\label{phi=0_Re_int_inequality}
\omega_0=\{(x,t)\in\mathbb{C}^2\colon\quad
|x|>\rho_0(t),\quad 
\Re F(\lambda_1)\geq 0
\quad\mbox{and}\quad
\Re F(\lambda_3)\geq 0\}.
\end{equation}
 
Similarly, it is possible to prove the solvability of the RH 
problem~\ref{initial_RHP} with the Stokes data (\ref{spm3_spm2=0})
at $t=0$ in the sector
\begin{equation}\nonumber
\arg x\in[
3\pi-\tfrac{6\pi}{7}+\tfrac{3}{7}\arctan\tfrac{1}{\sqrt5}+\epsilon,
3\pi+\tfrac{6\pi}{7}-\tfrac{3}{7}\arctan\tfrac{1}{\sqrt5}-\epsilon],\quad
|x|>\hat\rho_0(\epsilon),\quad
\epsilon>0.
\end{equation}
However, similarly to the choice of $\omega_0$ 
(\ref{Re_int>=0_around_phi=0}), it is more convenient 
to restrict ourselves to the less wide 
domain $\hat{\omega}_0$,
\begin{equation}\label{Re_int<=0_around_phi=3pi}
\hat{\omega}_0|_{t=0}=
\{x\in\mathbb{C}\colon\
|x|>\rho_0,\
\arg x\in\bigl[
3\pi-\beta_0,
3\pi+\beta_0
\bigr]\},\quad
\beta_0=\tfrac{3}{7}\arctan\tfrac{1}{\sqrt5}.
\end{equation}
This definition extends to any $t$ in an arbitrary simply connected
domain of the complex $t$ plane containing $t=0$,
\begin{equation}\label{phi=pi_Re_int_inequality}
\hat{\omega}_0=
\{(x,t)\in\mathbb{C}^2\colon\quad
|x|>\rho_0(t),\quad
\Re F(\lambda_1)\leq 0
\quad\mbox{and}\quad
\Re F(\lambda_3)\leq 0\}.
\end{equation}

In both the definitions (\ref{phi=0_Re_int_inequality}) and 
(\ref{phi=pi_Re_int_inequality}), the dependence of $\rho_0(t)$
on $t$ can be given using the appropriate rescaling, see e.g.\
\cite{K1, K2},
\begin{equation}\label{rho0(t)}
\rho_0(t)=\rho_0(1+|t|^{3/2}),\quad
\rho_0=const.
\end{equation}

\begin{rem}\label{sclim}
In the case of unbounded $t$, it is possible to recast
the whole problem taking $|t|$ as a large parameter. 
However, it is more convenient to introduce
a small parameter $\delta>0$ and to scale variables as 
$x=\delta^{-3/2}(x_0+\delta \xi)$,
$t=\delta^{-1}(t_0+\delta\tau)$, $x_0,t_0=const$,
$|x_0|+|t_0|\neq0$, cf.\ \cite{K1, K2}. 
This approach enables us to keep
$x$ and $t$ on equal footing and e.g.\ to study
the case of bounded $x$ and large $t$.
The scaling limit analysis described in \cite{K1, K2} 
yields the domains of solvability (\ref{phi=0_Re_int_inequality}) 
and (\ref{phi=pi_Re_int_inequality}) with the condition (\ref{rho0(t)}) on 
$|x|$ replaced by a condition of the 
form $\delta<\delta_0$.
The relevant boundary conditions (\ref{Re_int>=0_around_phi=0}) 
and (\ref{Re_int<=0_around_phi=3pi}) are replaced, respectively,
by similar conditions taken at $t_0=0$ or, equivalently, at 
$x_0\to\infty$. In the present paper, however, we do not 
develop such a generalized approach in detail.
\end{rem}

Finally, we find the asymptotic behavior of the  solution of \PItwo. 
Consider the singular integral equation equivalent to 
the RH problem for $\chi(\lambda)$,
\begin{equation*}
\chi(\lambda)=I
+\frac{1}{2\pi i}\int_{\ell}
\frac{\chi_{_-}(z)({\mathcal S}(z)-I)}{z-\lambda}\,dz,
\end{equation*}
where
\begin{equation*}
\chi_{_-}(\lambda)=I-\frac{1}{2}\chi_{_-}(\lambda)({\mathcal S}(\lambda)-I)
+\frac{P.V.}{2\pi i}\int_{\ell}
\frac{\chi_{_-}(z)({\mathcal S}(z)-I)}{z-\lambda}\,dz,
\end{equation*}
where $P.V.$ stands for principal value.

As $|x|$ is sufficiently large and $(x,t)\in\omega_0$ (respectively,
$(x,t)\in\hat\omega_0$),
the norm of the integral operator is less than one, and the integral 
equation can be solved iteratively. Thus, using the asymptotics 
(\ref{tilde_Psi_at_circle_jump}) and (\ref{chi_as}),
\begin{multline}\label{chi_as_computed}
(
u^{(0)}\tilde H_1
+\tfrac{1}{4}u_x^{(0)}
-\tfrac{1}{3}\tilde H_0
+\tfrac{1}{2}\tilde H_1^3
)\sigma_-
+\tfrac{1}{2}(
u^{(0)}
-\tfrac{1}{2}\lambda_5
+\tilde H_1^2
)\sigma_3
-\tilde H_1\sigma_+
=
\\
={\mathcal O}(x^{-2/3}\sigma_-)
+{\mathcal O}(x^{-5/6}\sigma_3)
+{\mathcal O}(x^{-1}\sigma_+),
\end{multline}
we find
\begin{equation}\label{H1(0)_u(0)_as0}
\tilde{H}_1={\mathcal O}(x^{-1}),\quad
u^{(0)}=\tfrac{1}{2}\lambda_5
+{\mathcal O}(x^{-5/6}).
\end{equation}
This yields the asymptotics (\ref{u_s12=0_as}) for $u^{(0)}(x,t)$
and (\ref{u_s32=0_as}) for $\hat u^{(0)}(x,t)$.
\end{proof}

\begin{rem}
For bounded $t$ as $|x|\to\infty$, the asymptotics of $u^{(0)}(x,t)$
and $\hat u^{(0)}(x,t)$ in (\ref{H1(0)_u(0)_as0}) yields 
the asymptotic formula from \cite{BMP, CV},
\begin{equation}\label{u0_CV}
u^{(0)}(x,t)\simeq\tfrac{1}{2}\lambda_5=
-\sqrt[3]{6}\,x^{1/3}
-\tfrac{2}{\sqrt[3]{6}}tx^{-1/3}
+{\mathcal O}(t^2x^{-1}).
\end{equation}
The expression (\ref{H1(0)_u(0)_as0})
extends the formula (\ref{u0_CV}) to the unbounded 
values of $t$, see Remark~\ref{sclim}.  For instance,
as $t\to-\infty$, the domain $\hat\omega_0$ contains
a neighborhood of $x=0$, see Figures~\ref{fig14} 
and~\ref{fig14a}, and it is possible to find the asymptotics
of $\hat u^{(0)}(x,t)$ for the values of $x$ satisfying
$|x|\ll(-t)^{3/2}$,
\begin{equation*}
\hat u^{(0)}(x,t)=xt^{-1}
+{\mathcal O}(x^3t^{-4}).
\end{equation*}
\end{rem}

\begin{rem}\label{multiple_branch}
The violation of the conditions $\lambda_5^2-8t\neq0$ 
and $\lambda_5^2-48t\neq0$, see (\ref{non-degeneracy_conditions}),
can take place at the boundary of the domains $\omega_0$ and
$\hat\omega_0$, see Figures~\ref{fig14} and~\ref{fig14a}. 
At such points, the construction of the model $\Psi$
function involves the use of the $\Psi$ functions associated with 
the first and the second Painlev\'e transcendents, respectively.
This does not affect the leading order term in the asymptotics 
(\ref{H1(0)_u(0)_as0}), while the lower order terms should 
be modified.
\end{rem}

\break
\subsection{2-parameter perturbations of the 0-parameter 
Riemann-Hilbert problem}
\subsubsection{Perturbation $s_2,s_{-2}\neq0$ of the RH problem~\ref{initial_RHP}
satisfying (\ref{spm2_spm1=0})}

\begin{thm}\label{s2_perturb_solvability}
For the Stokes multipliers
\begin{equation}\label{spm1=0}
s_{-1}=s_1=0,\quad
s_{-3}=s_3=s_{-2}+s_0+s_2=-i,
\end{equation}
and $(x,t)\in\omega_0$, the RH problem~\ref{initial_RHP}
is solvable. Assuming that $(x,t)\in\omega_0$ are such that
$\lambda_{1,3}\neq\lambda_5$ and $\lambda_1\neq\lambda_3$, 
the $x$ large asymptotics of the corresponding solution $u(x,t)$ 
of equation {\em\PItwo } is as follows,
\begin{multline}\label{u_as_eval}
u=u^{(0)}
-is_2\tfrac{1}{2\sqrt{\pi}}
(F''(\lambda_3))^{-1/2}
e^{-2F(\lambda_3)}
(1+{\mathcal O}(x^{-1/6}))
\\
-is_{-2}\tfrac{1}{2\sqrt{\pi}}
(F''(\lambda_1))^{-1/2}
e^{-2F(\lambda_1)}
(1+{\mathcal O}(x^{-1/6})),
\end{multline}
where $u^{(0)}\simeq-\sqrt[3]{6}\,x^{1/3}$ is the solution 
of {\em\PItwo } (\ref{u_s12=0_as}) corresponding 
to the Stokes multipliers
$s_{-2}=s_{-1}=s_1=s_2=0$, $s_{-3}=s_0=s_3=-i$.
The function $F(\lambda)$ is defined in (\ref{F_def}).
The branches of $F(\lambda_{1,3})$ and
$(F''(\lambda_{1,3}))^{-1/2}$ are fixed
by their values at $t=0$,
\begin{multline*}
F(\lambda_{1,3})\bigr|_{t=0}=
\tfrac{1}{2}
h_{\mp}
\tfrac{6}{7}
x^{7/6},
\\
\shoveleft{
(F''(\lambda_{1,3}))^{-1/2}\bigr|_{t=0}=
\pm i
2\sqrt{\pi}
A_{\mp}
x^{-1/4},
}\hfill
\end{multline*}
where the upper (resp., lower) subscript on the right hand side
corresponds to $\lambda_1$ (resp., to $\lambda_3$) on the left hand side,
and
\begin{equation}\label{h_pm_def}
h_{\sigma}
=
5^{1/4}
3^{5/12}
2^{11/12}
e^{\sigma\frac{i}{2}\arctan\frac{1}{\sqrt5}},\
A_{\sigma}=
\tfrac{\sqrt[8]{15}}{\sqrt[8]{32}\sqrt{\pi}}
e^{\sigma(i\frac{\pi}{4}-i\frac{1}{4}\arctan\frac{1}{\sqrt5})},\
\sigma\in\{+,-\}.
\end{equation}

\end{thm}

\begin{rem}
According to Remark~\ref{multiple_branch}, the asymptotics
(\ref{u_as_eval}) has to be modified at the points where condition
(\ref{non-degeneracy_conditions}) is violated.
\end{rem}

\begin{rem}
Two terms in the asymptotic expansions of $F(\lambda_{1,3})$ and
$(F''(\lambda_{1,3}))^{-1/2}$ with respect to $tx^{-2/3}\to0$ are
given below in (\ref{F13_ddF13_as}).
\end{rem}

\begin{proof}
First observe that the jump graph of the RH problem
under condition (\ref{spm1=0}) can be transformed to the
graph shown in Figure~\ref{fig8}.
\begin{figure}[htb]
\begin{center}
\mbox{\epsfig{figure=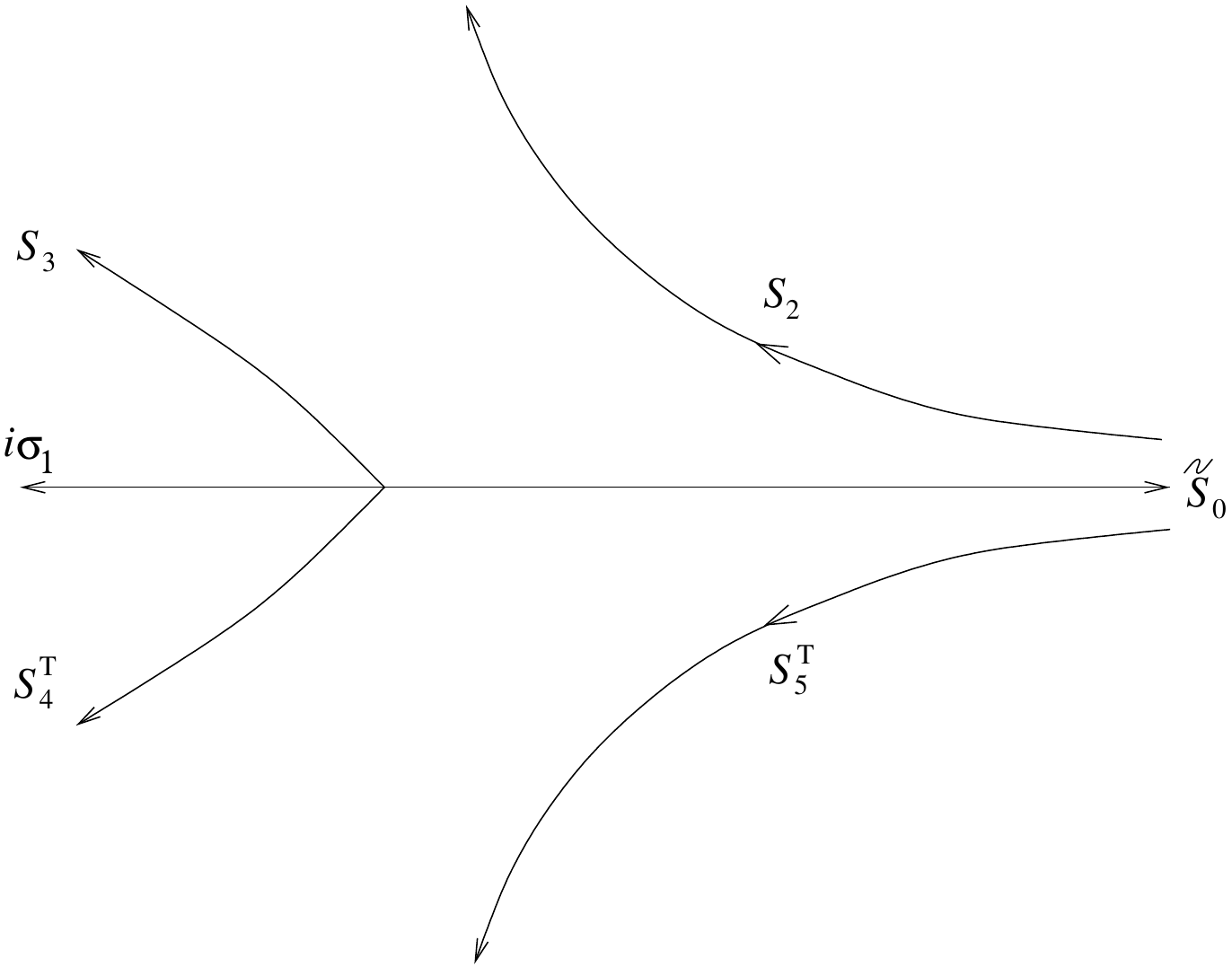,width=0.5\textwidth}}
\end{center}
\caption{The jump contour for the RH problem corresponding to
the degeneration $s_1=s_{-1}=0$. Here, $S_{\pm3}=I-i\sigma_+$,
$\tilde S_0=I-i\sigma_-$, $S_{\pm2}=I+s_{\pm2}\sigma_-$.}
\label{fig8}
\end{figure}
We look for the solution of the above problem in the form 
of the product,
\begin{equation}\label{Phi_spm1_or_spm2=0}
\Psi(\lambda)=
\bigl(I-(H_1-H^{(0)}_1)\sigma_-\bigr)
X(\lambda)\Psi^{(0)}(\lambda).
\end{equation}
Here $\Psi^{(0)}(\lambda)$ is the solution  constructed above
of the RH problem~\ref{initial_RHP} with the Stokes data
(\ref{spm2_spm1=0}), and $H^{(0)}_1$ denotes 
the Hamiltonian function (\ref{H1(0)_u(0)_as0}) computed on
$u^{(0)}(x,t)$.

The asymptotics of $X(\zeta)$ at infinity is as follows,
\begin{multline}\label{X_as_def}
X(\lambda)=
I
+\lambda^{-1}
\bigl[
-(H_1-H^{(0)}_1)\sigma_+
+\tfrac{1}{2}(u-u^{(0)}+(H_1-H^{(0)}_1)^2)\sigma_3
\bigr]
\\
+{\mathcal O}(\lambda^{-1}\sigma_-)
+{\mathcal O}(\lambda^{-3/2}),
\end{multline}
where $u^{(0)}$ is the solution to \PItwo\ (\ref{u_s12=0_as}).
$X(\lambda)$ is piece-wise holomorphic as being discontinuous
across the lines $\gamma_{2}$ and $\gamma_{-2}$ shown in 
Figure~\ref{fig10}.
\begin{figure}[htb]
\begin{center}
\mbox{\epsfig{figure=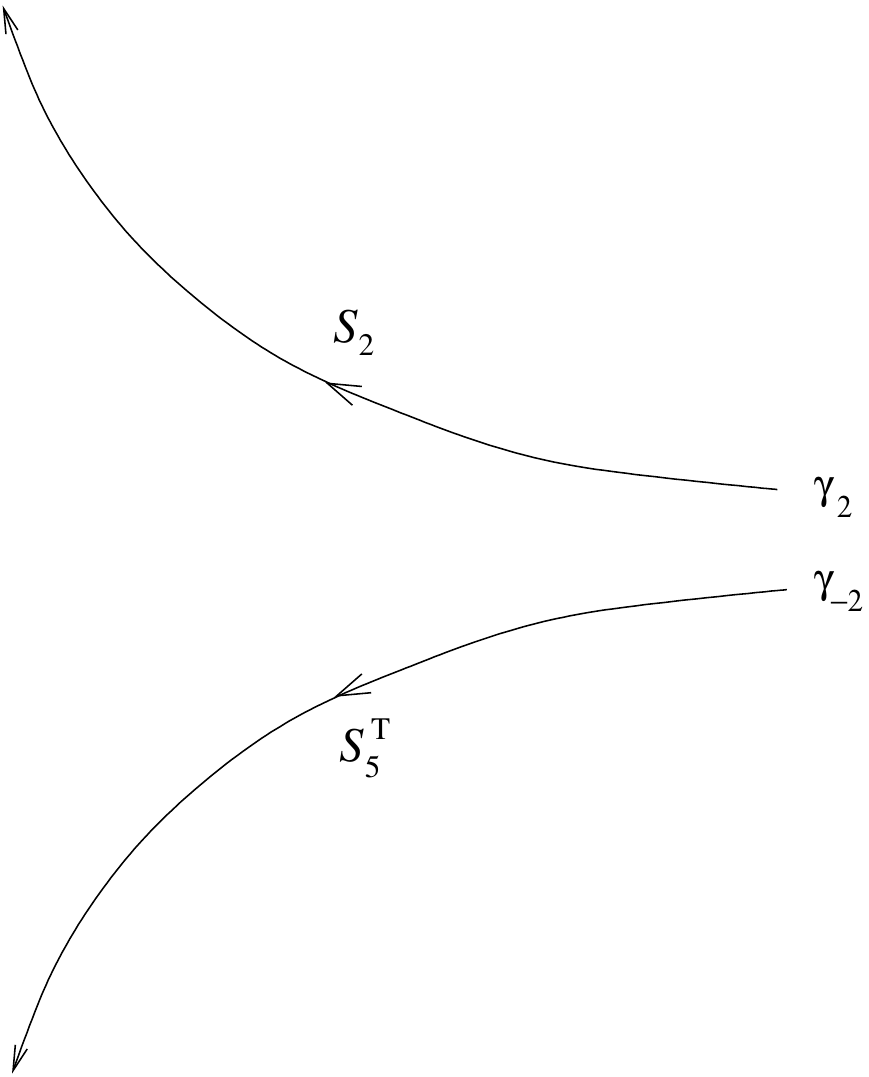,width=0.3\textwidth}}
\end{center}
\caption{The jump contour for $X(\lambda)$ as $s_{-1}=s_1=0$.}
\label{fig10}
\end{figure}
Its jumps are described by the equations
\begin{multline}\label{X_jump}
X_+(\lambda)=X_-(\lambda){\mathcal S}_k(\lambda),\quad
\lambda\in\gamma_k,
\\
\shoveleft{
{\mathcal S}_k(\lambda):=
\Psi^{(0)}(\lambda)S_k(\Psi^{(0)}(\lambda))^{-1},\quad
k=\pm2.
}\hfill
\end{multline}

Let us find the large $x$ asymptotics of the jump matrices 
${\mathcal S}_k(\lambda)$, $k=\pm2$. Using (\ref{chi_def}), 
(\ref{tilde_Psi_def}) for $|\lambda-\lambda_5|>R|x|^{1/3}$ 
and (\ref{chi_as_computed}),
\begin{multline}\label{tilde_S_pm2_def}
{\mathcal S}_{k}(\lambda):=
\Psi^{(0)}(\lambda)S_{k}(\Psi^{(0)}(\lambda))^{-1}=
I+s_{k}\Psi^{(0)}(\lambda)\sigma_-(\Psi^{(0)}(\lambda))^{-1}=
\\
=I+s_{k}
e^{-2F(\lambda)}
B_{k}(\lambda),\quad
\lambda\in\gamma_{k},\quad
k=\pm2,
\\
B_k(\lambda)=
\bigl(I-\tilde H_1\sigma_-\bigr)
\chi(\lambda)
(\lambda-\lambda_5)^{-\frac{1}{4}\sigma_3}
\tfrac{1}{2}(\sigma_3+\sigma_+-\sigma_-)
\times
\\
\times
(\lambda-\lambda_5)^{\frac{1}{4}\sigma_3}
\chi^{-1}(\lambda)
\bigl(I+\tilde H_1\sigma_-\bigr).
\end{multline}
Formula (\ref{tilde_S_pm2_def}) immediately yields 
the exponentially fast decay of the jump multiplier
as $x\to\infty$ in the interior of $\omega_0$ and therefore 
existence of $X(\lambda)$. The proof of the existence of 
$X(\lambda)$ for all $(x,t)\in\omega_0$ including its boundary
$\partial\omega_0$ requires more efforts.

Let $\gamma_2$, $\gamma_{-2}$ be the level lines
$\Im F(\lambda)=const$ passing through
$\lambda=\lambda_3$ and $\lambda=\lambda_1$, respectively.
Introduce the auxiliary functions $\hat X_k(\lambda)$,  $k=\pm2$ 
\begin{equation}
\label{hat_Xpm2_def}
\begin{split}
\hat X_k(\lambda)&=I
+\frac{s_kB_k}{2\pi i}
\int_{\gamma_k}
\frac{e^{-2F(z)}}{z-\lambda}\,dz\,,\quad
k={\pm2},
\\
B_2&:=B_2(\lambda_3),\quad
B_{-2}:=
\hat X_{2}(\lambda_1)
B_{-2}(\lambda_1)
\hat X_{2}^{-1}(\lambda_1)
\end{split}
\end{equation}
where $B_k(\lambda)$, $k=\pm 2$ have been defined in (\ref{tilde_S_pm2_def}).

We are looking for $X(\lambda)$ defined by (\ref{X_as_def}) and (\ref{X_jump})  in the form of the product
\begin{equation}\label{Y_def}
X(\lambda)=Y(\lambda)\hat X_{-2}(\lambda)\hat X_2(\lambda)
\end{equation}
with $\hat X_k(\lambda)$, $k=\pm2$ as in (\ref{hat_Xpm2_def}).
The correction function $Y(\lambda)$ is piece-wise holomorphic,
normalized at infinity to unity and, across $\gamma_{\pm2}$, 
satisfies the jump condition
\begin{equation}\label{Y_jump}
\begin{split}
Y_+(\lambda)&=
Y_-(\lambda)
G(\lambda),\quad
\lambda\in\gamma_{\pm2},
\\
G(\lambda)&=
(\hat X_{-2}(\lambda))_-
(\hat X_2(\lambda))_-
{\mathcal S}_k(\lambda)
(\hat X_{2}(\lambda))_+^{-1}
(\hat X_{-2}(\lambda))_+^{-1}, \;\;k=\pm2,
\end{split}
\end{equation}
where $(\hat X_{k}(\lambda))_{\pm}$ are the boundary values on the left and right of the contour $\gamma_k$.
Since the $B_k(\lambda)$, $k={\pm2}$, are conjugate to 
the nilpotent constant matrix $\tfrac{1}{2}(\sigma_3+i\sigma_2)$,
where $\sigma_2=-i\sigma_++i\sigma_-$,
the jump matrix $G(\lambda)$ satisfies the estimate,
\begin{equation}\label{Y_jumps_corr_estim}
\|G(\lambda)-I\|_{\lambda\in\gamma_{\pm2}}
\leq C_3
|s_{\pm2}|
\frac{|\lambda-\lambda_5|^{1/2}
|\lambda-\lambda_{3,1}|}{1+c_3|\lambda-\lambda_{3,1}|}
e^{-2|\scriptRe F(\lambda)|}
\end{equation}
for some positive constants $C_3$ and $c_3$.
This implies the $L^2$ estimate
\begin{equation}\label{Y_jumps_corr_L2_estim}
\|G-I\|_{L^2(\gamma_2\cup\gamma_{-2})}
\leq C_4 |x|^{-5/24}
(
|s_2|
e^{-2|\scriptRe F(\lambda_3)|}
+|s_{-2}|
e^{-2|\scriptRe F(\lambda_1)|}),
\end{equation}
for some $C_4>0$ 
and therefore the existence of $Y(\lambda)$ (and hence of $X(\lambda)$) 
as $(x,t)\in\omega_0$, cf.\ e.g.\ \cite{FIKN}. 
Furthermore, since the Cauchy operator is bounded
in $L^2(\gamma_2\cup\gamma_{-2})$,
the difference $Y-I$ admits the $L^2$-estimate similar to 
(\ref{Y_jumps_corr_L2_estim}). We also have the matrix norm estimate
\begin{equation}\label{Y_jumps_corr_matrix_norm_estim}
\Bigl\|
\int_{\gamma_2\cup\gamma_{-2}}(G(z)-I)\,dz
\Bigr\|\leq
C_5|x|^{-5/12}(
|s_2|
e^{-2|\scriptRe F(\lambda_3)|}
+|s_{-2}|
e^{-2|\scriptRe F(\lambda_1)|}),
\end{equation}
with some $C_5>0$, useful to estimate the contribution of $Y(\lambda)$
into the asymptotics of $X(\lambda)$ as $\lambda\to\infty$.

In leading order of $x$,  
the asymptotics (\ref{X_as_def})
of $X(\lambda)$ at infinity is computed with 
(\ref{Y_def}) and (\ref{Y_jumps_corr_matrix_norm_estim}) to be
\begin{multline}\label{X_as}
X(\lambda)=
I
+\lambda^{-1}
\bigl[
-(H_1-H^{(0)}_1)\sigma_+
+\tfrac{1}{2}(u-u^{(0)}+(H_1-H^{(0)}_1)^2)\sigma_3
\bigr]
\\
+{\mathcal O}(\lambda^{-1}\sigma_-)
+{\mathcal O}(\lambda^{-3/2})
=I
+\lambda^{-1}
\Bigl\{
-\frac{s_{-2}}{2\pi i}
\int_{\gamma_{-2}}
e^{-2F(z)}
\,dz\,\bigl(B_{-2}+{\mathcal O}(x^{-1/6})\bigr)
\\
-\frac{s_2}{2\pi i}
\int_{\gamma_2}
e^{-2F(z)}
\,dz\,\bigl(B_2+{\mathcal O}(x^{-1/6})\bigr)
\Bigr\}
+{\mathcal O}(\lambda^{-2}).
\end{multline}
By the definitions in (\ref{hat_Xpm2_def}) and (\ref{tilde_S_pm2_def}),
\begin{multline}\label{Bpm2_as}
B_{\pm2}=
\tfrac{1}{2}
(\lambda_{3,1}-\lambda_5)^{-1/2}\sigma_+
+\tfrac{1}{2}\sigma_3
-\tfrac{1}{2}(\lambda_{3,1}-\lambda_5)^{1/2}\sigma_-
\\
+{\mathcal O}(x^{-7/6}\sigma_{3})
+{\mathcal O}(x^{-1}\sigma_{-})
+{\mathcal O}(x^{-4/3}\sigma_+).
\end{multline}
From the $\sigma_+$-component of (\ref{X_as}), we find
\begin{equation}\label{H-H0_estim}
|H_1-H_1^{(0)}|\leq
C_6|x|^{-5/12}(
|s_2|
e^{-2|\scriptRe F(\lambda_3)|}
+|s_{-2}|
e^{-2|\scriptRe F(\lambda_1)|}).
\end{equation}
Using (\ref{Bpm2_as}) and (\ref{H-H0_estim}) in the $\sigma_3$-component 
of (\ref{X_as}), we obtain the leading order of $u(x,t)$ in terms of quadratures,
\begin{multline}\label{H_u_as}
u-u^{(0)}=
-\tfrac{1}{2\pi i}\Bigl(
s_2
\int_{\gamma_2}
+s_{-2}
\int_{\gamma_{-2}}
\Bigr)
e^{-2F(z)}
\,dz\,
(1+{\mathcal O}(x^{-1/6})).
\end{multline}
Evaluating the asymptotics of the integrals via the classical
steepest descent method,
we get (\ref{u_as_eval}).
\end{proof}

\subsubsection{Perturbation $s_3,s_{-3}\neq0$ of the RH 
problem~\ref{initial_RHP} satisfying (\ref{spm3_spm2=0})}

\begin{thm}
For the Stokes multipliers
\begin{equation}\label{spm2=0}
s_{-2}=s_2=0,\quad
s_0=s_1+s_3=s_{-1}+s_{-3}=-i,
\end{equation}
and $(x,t)\in\hat\omega_0$, the RH problem~\ref{initial_RHP}
is solvable. Assuming that $(x,t)\in\hat\omega_0$ are such that
$\lambda_{1,3}\neq\lambda_5$ and $\lambda_1\neq\lambda_3$, 
the large $x$ asymptotics of the corresponding solution 
$\hat u(x,t)$ of equation {\em\PItwo} is as follows,
\begin{multline}\label{u3_as_eval}
\hat u=
\hat u^{(0)}
-is_3\tfrac{1}{2\sqrt{\pi}}
(-F''(\lambda_1))^{-1/2}
e^{2F(\lambda_1)}
(1+{\mathcal O}(x^{-1/6}))
\\
-is_{-3}\tfrac{1}{2\sqrt{\pi}}
(-F''(\lambda_3))^{-1/2}
e^{2F(\lambda_3)}
(1+{\mathcal O}(x^{-1/6})),
\end{multline}
where $\hat u^{(0)}\simeq\sqrt[3]{6}(e^{-i3\pi}x)^{1/3}$
is the solution of {\em\PItwo\ } (\ref{u_s32=0_as})
corresponding to the Stokes multipliers
$s_{-3}=s_{-2}=s_2=s_3=0$,
$s_{-1}=s_0=s_1=-i$.
The function $F(\lambda)$ is defined in (\ref{F_def}).
The branches of $F(\lambda_{1,3})$ and
$(-F''(\lambda_{1,3}))^{-1/2}$ are fixed by their
values at $t=0$,
\begin{multline*}
F(\lambda_{1,3})\bigr|_{t=0}=
\tfrac{1}{2}
h_{\mp}
\tfrac{6}{7}
x^{7/6},
\\
\shoveleft{
(-F''(\lambda_{1,3}))^{-1/2}\bigr|_{t=0}=
-2\sqrt{\pi}
A_{\mp}
x^{-1/4},
}\hfill
\end{multline*}
where the upper (resp., lower) subscript on the right hand side
corresponds to $\lambda_1$ (resp., to $\lambda_3$) on the left
hand side, and the constants $h_{\pm}$ and $A_{\pm}$ are 
defined in (\ref{h_pm_def}).
\end{thm}

\begin{proof}
Observe first of all that the jump graph for the 
RH problem~\ref{initial_RHP} under the condition
(\ref{spm2=0}) can be transformed to the one shown in 
Figure~\ref{fig9}.
\begin{figure}[htb]
\begin{center}
\mbox{\epsfig{figure=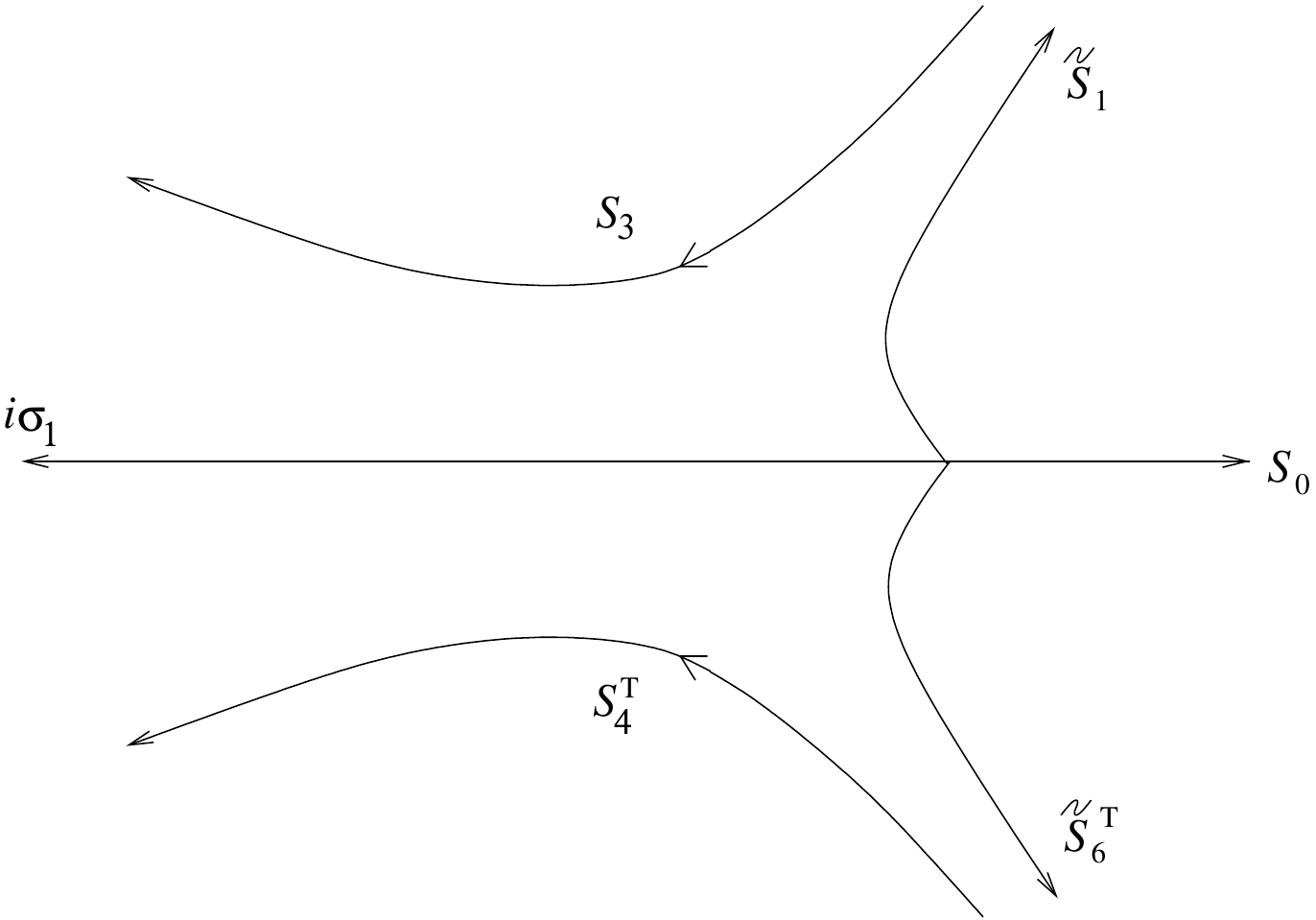,width=0.6\textwidth}}
\end{center}
\caption{The jump contour for the RH problem corresponding to
the degeneration $s_2=s_{-2}=0$. Here, 
$\tilde S_1=\tilde S_{-1}=I-i\sigma_+$,
$S_0=I-i\sigma_-$, $S_3=I+s_3\sigma_+$, $S_{-3}=I+s_{-3}\sigma_+$.}
\label{fig9}
\end{figure}
Look for the solution of the RH problem in the form 
of the product (\ref{Phi_spm1_or_spm2=0}). The correction 
function $X(\lambda)$ has the asymptotics at infinity as 
in (\ref{X_as_def}). However in contrast to (\ref{X_jump}),
it is discontinuous across the jump graph 
$\gamma_3\cup\gamma_{-3}$
shown in Figure~\ref{fig11}, and the jumps are described by 
the formulas
\begin{multline}\label{X3_jump}
X_+(\lambda)=X_-(\lambda){\mathcal S}_k(\lambda),\quad
\lambda\in\gamma_k,
\\
\shoveleft{
{\mathcal S}_k(\lambda):=
\Psi^{(0)}_-(\lambda)S_k(\Psi_-^{(0)}(\lambda))^{-1},\quad
k=\pm3.
}\hfill
\end{multline}
\begin{figure}[htb]
\begin{center}
\mbox{\epsfig{figure=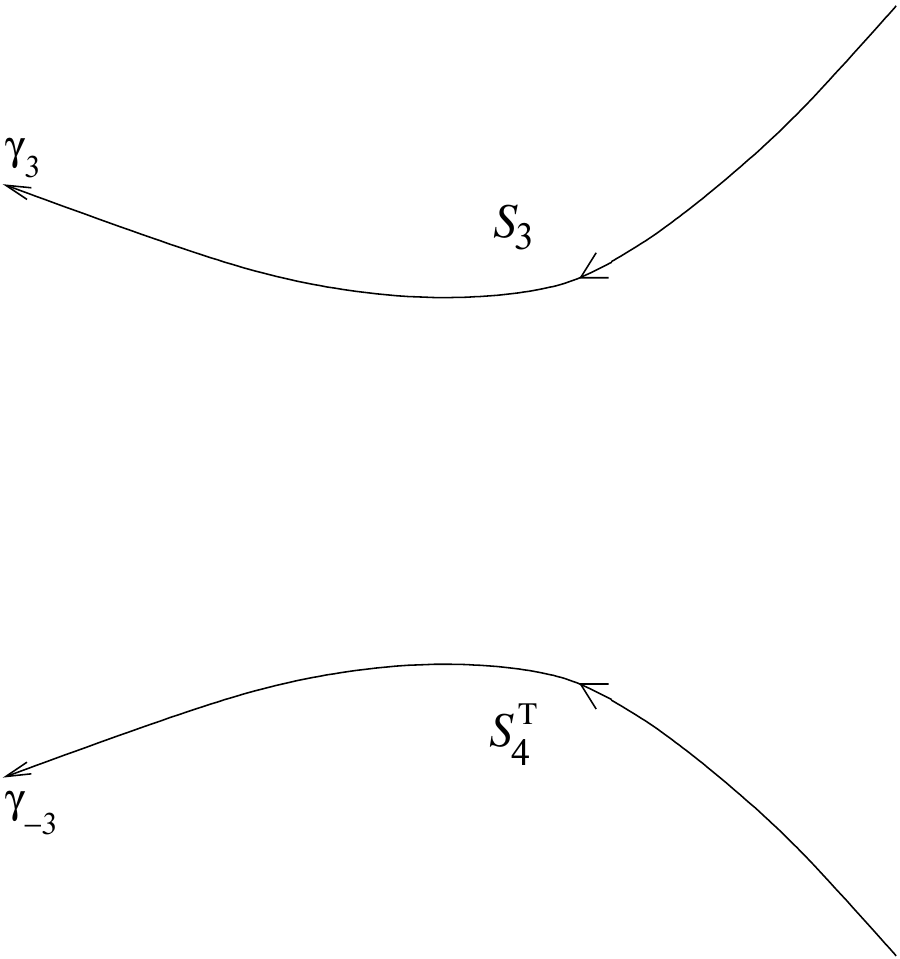,width=0.3\textwidth}}
\end{center}
\caption{The jump contour for the correction RH problem as 
$s_{-2}=s_2=0$.}
\label{fig11}
\end{figure}
To construct $X(\lambda)$, find the asymptotics of the jump 
matrix in (\ref{X3_jump}).
Using (\ref{chi_def}), (\ref{tilde_Psi_def}) for 
$|\lambda-\lambda_5|>R|x|^{1/3}$ and (\ref{chi_as_computed}), we find
\begin{multline}\nonumber
{\mathcal S}_k(\lambda):=
\Psi^{(0)}_-(\lambda)S_k(\Psi_-^{(0)}(\lambda))^{-1}=
I+s_k\Psi^{(0)}_-(\lambda)
\sigma_+
(\Psi_-^{(0)}(\lambda))^{-1}=
\\
=I
+s_k
e^{2F(\lambda)}
B_k(\lambda),
\\
B_k(\lambda)
=(I-\tilde H_1\sigma_-)
\chi(\lambda)
(\lambda-\lambda_5)^{-\frac{1}{4}\sigma_3}
\tfrac{1}{2}
(\sigma_3-\sigma_++\sigma_-)
\times
\\
\times
(\lambda-\lambda_5)^{\frac{1}{4}\sigma_3}
\chi^{-1}(\lambda)
(I+\tilde H_1\sigma_-),\quad
k={\pm3}.
\end{multline}
Let $\gamma_3$ and $\gamma_{-3}$ be
the level lines $\Im F(\lambda)=const$ passing through 
$\lambda=\lambda_1$ and $\lambda=\lambda_3$, respectively. 
Introduce the auxiliary function $\hat X_3(\lambda)$
and $\hat X_{-3}(\lambda)$,
\begin{multline}\nonumber
\hat X_k(\lambda)=I
+\frac{s_k}{2\pi i}
\int_{\gamma_k}
\frac{e^{2F(z)}}{z-\lambda}\,dz\,B_k,\quad
k={\pm3},
\\
B_3:=B_3(\lambda_3),\quad
B_{-3}:=
\hat X_3(\lambda_1)
B_{-3}(\lambda_1)\hat X_3^{-1}(\lambda_1).
\end{multline}
We are looking for $X(\lambda)$ in the form (\ref{Y_def}),
\begin{equation}\nonumber
X(\lambda)=Y(\lambda)\hat X_{-3}(\lambda)\hat X_3(\lambda).
\end{equation}
$Y(\lambda)$ is piece-wise holomorphic,
normalized at infinity to unity and, across $\gamma_{\pm3}$,
satisfying the jump condition as (\ref{Y_jump}),
\begin{multline}\nonumber
Y_+(\lambda)=
Y_-(\lambda)
G(\lambda),\quad
\lambda\in\gamma_k,\quad
k={\pm3},
\\
G(\lambda)=
(\hat X_{-3}(\lambda))_-
(\hat X_3(\lambda))_-
{\mathcal S}_k(\lambda)
(\hat X_3(\lambda))_+^{-1}
(\hat X_{-3}(\lambda))_+^{-1}.
\end{multline}
Similarly to (\ref{Y_jumps_corr_estim}),
the jump matrix $G(\zeta)$ satisfies
\begin{multline}\nonumber
\|G(\lambda)-I\|_{\gamma_k}
\leq C_7
|s_k|
\frac{|\lambda-\lambda_5|^{1/2}
|\lambda-\lambda_{j_k}|}{1+c_7|\lambda-\lambda_{j_k}|}
e^{-2|\scriptRe F(\lambda)|},\quad
\\
k=\pm3,\quad
j_3=1,\quad
j_{-3}=3.
\end{multline}
This implies the $L^2$-estimate
\begin{equation}\label{Y3_jumps_corr_L2_estim}
\|G-I\|_{L^2(\gamma_3\cup\gamma_{-3})}
\leq C_8
|x|^{-5/24}(
|s_3|e^{-2|\scriptRe F(\lambda_1)|}
+|s_{-3}|e^{-2|\scriptRe F(\lambda_3)|}),
\end{equation}
and thus the existence of $Y(\lambda)$ and 
of $X(\lambda)$ as $(x,t)\in\hat\omega_0$.
We also have the matrix norm estimate
\begin{equation}\nonumber
\Bigl\|
\int_{\gamma_3\cup\gamma_{-3}}
(G(z)-I)\,dz
\Bigr\|\leq
C_9|x|^{-5/12}(
|s_3|e^{-2|\scriptRe F(\lambda_1)|}
+|s_{-3}|e^{-2|\scriptRe F(\lambda_3)|}).
\end{equation}
In the leading order in $x$,  the asymptotics of $X(\lambda)$ 
as $\lambda\to\infty$ is computed with (\ref{Y_def}) and
(\ref{Y3_jumps_corr_L2_estim}) to be (cf.\ (\ref{X_as}))
\begin{multline}\nonumber
X(\zeta)=
I
+\lambda^{-1}\Bigl\{
-\frac{s_3}{2\pi i}
\int_{\gamma_3}
e^{2F(z)}\,dz\,(B_3+{\mathcal O}(x^{-1/6}))
\\
-\frac{s_{-3}}{2\pi i}
\int_{\gamma_{-3}}
e^{2F(z)}\,dz\,(B_{-3}+{\mathcal O}(x^{-1/6}))
\Bigr\}
+{\mathcal O}(\lambda^{-2}),
\\
B_{\pm3}=
-\tfrac{1}{2}(\lambda_{1,3}-\lambda_5)^{-1/2}\sigma_+
+\tfrac{1}{2}\sigma_3
+\tfrac{1}{2}(\lambda_{1,3}-\lambda_5)^{1/2}\sigma_-
\\
+{\mathcal O}(x^{-7/6}\sigma_{3})
+{\mathcal O}(x^{-1}\sigma_{-})
+{\mathcal O}(x^{-4/3}\sigma_+).
\end{multline}
This implies the quadrature formula for the leading order term
(cf.\ (\ref{H_u_as})),
\begin{multline}\nonumber
u-u^{(0)}=
-\tfrac{1}{2\pi i}
\Bigl(
s_3
\int_{\gamma_3}
+s_{-3}
\int_{\gamma_{-3}}
\Bigr)
e^{2F(z)}\,dz\,(1+{\mathcal O}(x^{-1/6}))
.
\end{multline}
Then formula (\ref{u3_as_eval}) follows using 
 classical steepest descent analysis.
\end{proof}

\section{Symmetries, tronqu\'ee solutions and
the quasi-linear Stokes phenomenon}
\label{um_proliferation}

\subsection{Rotational symmetry and families of 
degenerated solutions}

Using (\ref{sk_rot_symm}), we find that if $u=f(x,t,\{s_k\})$ 
is a solution of \PItwo\ corresponding to the Stokes multipliers
$\{s_k\}$, then 
\begin{equation}\label{u_rot_symm}
\tilde u(\tilde x,\tilde t)=
e^{-i\frac{4\pi}{7}n}
f(e^{-i\frac{2\pi}{7}n}\tilde x,
e^{-i\frac{6\pi}{7}n}\tilde t,
\{s_{k-2n}\}),\quad
\tilde x=e^{i\frac{2\pi}{7}n}x,\quad
\tilde t=e^{i\frac{6\pi}{7}n}t,\quad
\end{equation}
is another solution of \PItwo\ corresponding
to the cyclically permuted multipliers $\{s_{k-2n}\}$.

First consider the case $t=0$ when the domains 
$\omega_0$ and $\hat\omega_0$ reduce to
the sectors (\ref{Re_int>=0_around_phi=0}) and 
(\ref{Re_int<=0_around_phi=3pi}), respectively.

Denoting the solution with the asymptotics (\ref{u_as_eval}) 
and (\ref{u3_as_eval}) at $t=0$ as $u_0(x)$ and $\hat u_0(x)$,
respectively, and applying to them the symmetry transformation (\ref{u_rot_symm})
with $n=3m$, $m\in{\mathbb Z}$, we find solutions $u_m(x)$ and 
$\hat u_m(x)$.
Solutions $u_m(x)$ correspond to the Stokes multipliers
\begin{equation}\label{um_Stokes}
s_{m-1}=s_{m+1}=0,\quad
s_{m-3}=s_{m+3}=-i,\quad
s_{m-2}+s_{m}+s_{m+2}=-i,
\end{equation}
and have the large $x$ asymptotics
\begin{multline}\label{um_rot_symm}
u_m(x)=
u^{(m)}(x)
-s_{m+2}
A_+
i^m
x^{-1/4}
e^{(-1)^{m+1}\frac{6}{7}x^{7/6}h_+}
(1+{\mathcal O}(x^{-1/6}))
\\
+s_{m-2}
A_-
i^m
x^{-1/4}
e^{(-1)^{m+1}\frac{6}{7}x^{7/6}h_-}
(1+{\mathcal O}(x^{-1/6})),
\\
x\to\infty,\quad
\arg x\in[
-\alpha_0+\tfrac{6\pi}{7}m,\alpha_0+\tfrac{6\pi}{7}m],\quad
\alpha_0=
\tfrac{3\pi}{7}-\tfrac{3}{7}\arctan\tfrac{1}{\sqrt5},
\end{multline}
where $u^{(m)}(x)=e^{-i\frac{12\pi}{7}m}
u^{(0)}(e^{-i\frac{6\pi}{7}m}x)
\simeq
-\sqrt[3]{6}\,x^{1/3}$
is the solution of \PItwo\ corresponding to 
$s_{m+1}=s_{m-1}=s_{m+2}=s_{m-2}=0$,
$s_m=s_{m+3}=s_{m-3}=-i$, and where the constants $h_{\sigma}$
and $A_{\sigma}$, $\sigma\in\{+,-\}$,
are defined in (\ref{h_pm_def}).

Respectively, the solutions $\hat u_m(x)$ correspond to 
the multipliers
\begin{equation}\label{hat_um_Stokes}
s_{m-2}=s_{m+2}=0,\quad
s_{m}=s_{m-1}+s_{m-3}=
s_{m+1}+s_{m+3}=-i,
\end{equation}
and have the asymptotics
\begin{multline}\label{hat_um_rot_symm}
\hat u_m(x)=
\hat u^{(m)}(x)
+s_{m+3}
A_-
i^{m+1}
x^{-1/4}
e^{(-1)^m\frac{6}{7}x^{7/6}h_-}
(1+{\mathcal O}(x^{-1/6}))
\\
+s_{m-3}
A_+
i^{m}
x^{-1/4}
e^{(-1)^{m+1}\frac{6}{7}x^{7/6}h_+}
(1+{\mathcal O}(x^{-1/6})),
\\
x\to\infty,\quad
\arg x\in\bigl[
3\pi-\beta_0
+\tfrac{6\pi}{7}m,
3\pi+\beta_0
+\tfrac{6\pi}{7}m\bigr],\quad
\beta_0=\tfrac{3}{7}\arctan\tfrac{1}{\sqrt5},
\end{multline}
where $\hat u^{(m)}(x)=
e^{-i\frac{12\pi}{7}m}
\hat u^{(0)}(e^{-i\frac{6\pi}{7}m}x)
\simeq-\sqrt[3]{6}\,x^{1/3}$
is the solution of \PItwo\ for the Stokes
multipliers $s_{m+2}=s_{m-2}=s_{m+3}=s_{m-3}=0$,
$s_{m}=s_{m+1}=s_{m-1}=-i$.

\subsection{2-parameter degenerated solutions as $t\neq0$}

The extension of the asymptotics (\ref{um_rot_symm}) and
(\ref{hat_um_rot_symm}) to $t\neq0$ and $x\to\infty$ is 
straightforward. Observing that the variables 
$t$ and $x$ can be expressed using (\ref{t_x_from_lambda_k}) 
in terms of the branch points of the model algebraic curve,
we can write $u(x,t)=g(\{\lambda_j\},\{s_k\})$ and
recast (\ref{u_rot_symm}) into the form
\begin{equation}\label{u_lambda_k_rot_symm}
\tilde u(\tilde x,\tilde t)=
e^{-i\frac{4\pi}{7}n}
g(\{e^{i\frac{4\pi}{7}n}\tilde\lambda_j\},\{s_{k-2n}\}),\quad
\tilde\lambda_j=e^{-i\frac{4\pi}{7}n}\lambda_j.
\end{equation}

Denoting the solution with the asymptotics (\ref{u_as_eval}) 
and (\ref{u3_as_eval}) as $u_0(x,t)$ where $(x,t)\in\omega_0$,
and $\hat u_0(x,t)$ with $(x,t)\in\hat\omega_0$,
respectively, and applying to them the symmetry transformation
(\ref{u_lambda_k_rot_symm}) with $n=3m$, $m\in{\mathbb Z}$,
we find solutions $u_m(x,t)$ and $\hat u_m(x,t)$ with
$(x,t)\in\omega_m$ and $(x,t)\in\hat\omega_m$, respectively.
The latter sectors are defined as the images of $\omega_0$ and 
$\hat\omega_0$ under rotation, see (\ref{u_rot_symm}): 
\begin{defn}\label{omega_m_def}
$(x,t)\in\omega_m$ iff
$(e^{-i\frac{6\pi}{7}m}x,
e^{-i\frac{18\pi}{7}m}t)\in\omega_0$.
The sectors $\hat\omega_m$ are defined similarly.
\end{defn}

Solutions $u_m(x,t)$ corresponding to the Stokes multipliers
(\ref{um_Stokes}) have the asymptotics
\begin{multline}\label{um_as_eval}
u_m(x,t)=
u^{(m)}(x,t)
-(-1)^m
\frac{is_{m+2}}{2\sqrt{\pi}}
(F''(\lambda_3))^{-1/2}
e^{-2F(\lambda_3)}
(1+{\mathcal O}(x^{-1/6}))
\\
-(-1)^m
\frac{is_{m-2}}{2\sqrt{\pi}}
(F''(\lambda_1))^{-1/2}
e^{-2F(\lambda_1)}
(1+{\mathcal O}(x^{-1/6})),
\\
x\to\infty,\quad
(x,t)\in\omega_m,
\end{multline}
where
$u^{(m)}(x,t)=
e^{-i\frac{12\pi}{7}m}
u^{(0)}(e^{-i\frac{6\pi}{7}m}x,e^{-i\frac{18\pi}{7}m}t)$.

The solutions $\hat u_m(x,t)$ for the
multipliers (\ref{hat_um_Stokes}) have the asymptotics
\begin{multline}\label{hat_um_as_eval}
\hat u_m(x,t)=
\hat u^{(m)}(x,t)
-(-1)^m
\frac{is_{m+3}}{2\sqrt{\pi}}
(-F''(\lambda_1))^{-1/2}
e^{2F(\lambda_1)}
(1+{\mathcal O}(x^{-1/6}))
\\
-(-1)^m
\frac{is_{m-3}}{2\sqrt{\pi}}
(-F''(\lambda_3))^{-1/2}
e^{2F(\lambda_3)}
(1+{\mathcal O}(x^{-1/6})),
\\
x\to\infty,\quad
(x,t)\in\hat\omega_m,
\end{multline}
where 
$\hat u^{(m)}(x,t)=
e^{-i\frac{12\pi}{7}m}
\hat u^{(0)}(e^{-i\frac{6\pi}{7}m}x,e^{-i\frac{18\pi}{7}m}t)$.

The domains $\omega_m$ and $\hat\omega_m$
in (\ref{um_as_eval}) and (\ref{hat_um_as_eval})
at $t=0$ become
\begin{multline}\nonumber
\omega_m=\bigl\{
x\in{\mathbb C}\colon\quad
\arg x\in[
-\alpha_0+\tfrac{6\pi}{7}m,\alpha_0+\tfrac{6\pi}{7}m],\quad
|x|>\rho_0
\bigr\},
\\
\shoveleft{
\hat\omega_m=\bigl\{
x\in{\mathbb C}\colon\quad
\arg x\in[
3\pi-\beta_0+\tfrac{6\pi}{7}m,3\pi+\beta_0+\tfrac{6\pi}{7}m],\quad
|x|>\rho_0
\bigr\},
}\hfill
\\
\alpha_0=
\tfrac{3\pi}{7}-\tfrac{3}{7}\arctan\tfrac{1}{\sqrt5},\quad
\beta_0=
\tfrac{3}{7}\arctan\tfrac{1}{\sqrt5}.
\end{multline}
The points $\lambda_j$, $j=1,3,5$, satisfy the conditions
(\ref{lambda13_eq_sol}), (\ref{lambda5_eq}). 
$F(\lambda)$ is defined in (\ref{F_def}). In (\ref{um_as_eval}) 
and (\ref{hat_um_as_eval}), the branches of $F(\lambda_{1,3})$,
$(F''(\lambda_{1,3}))^{-1/2}$ and $(-F''(\lambda_{1,3}))^{-1/2}$ 
are chosen as the branches with the asymptotics (\ref{um_rot_symm}) 
and (\ref{hat_um_rot_symm}) at $t=0$.

\subsection{``Bitronqu\'ee'' solutions}

The domains $\omega_m$ and $\hat\omega_{m-3}$ as well as
$\hat\omega_{m-3}$ and $\omega_{m+1}$ are adjacent 
at $t=0$, 
\begin{multline*}
\omega_m\cap\hat\omega_{m-3}=\bigl\{
x\in{\mathbb C}\colon\quad
\arg x=\tfrac{3\pi}{7}-\beta_0+\tfrac{6\pi}{7}m,\quad
|x|>\rho_0
\bigr\},
\\
\shoveleft{
\hat\omega_{m-3}\cap\omega_{m+1}=\bigl\{
x\in{\mathbb C}\colon\quad
\arg x=\tfrac{3\pi}{7}+\beta_0+\tfrac{6\pi}{7}m,\quad
|x|>\rho_0
\bigr\},
}\hfill
\end{multline*}
and remain adjacent for arbitrary $t$.

It is convenient to interpret solutions $u_m(x)$ and 
$\hat u_{m-3}(x)$ as the solution families parameterized
by $s_{m\pm2}$ and $s_{m\pm3}$, respectively. Intersections 
of these families yield 1-parameter families corresponding
to the Stokes multipliers
\begin{equation}\label{um_hat_u_m-3_Stokes}
s_{m+2}=s_{m-1}=s_{m+1}=0,\quad
s_{m-3}=s_{m+3}=-i,\quad
s_{m-2}+s_{m}=-i.
\end{equation}
The relevant large $x$ asymptotics is as follows,
\begin{multline}\label{um_hat_u_m-3_as}
u_m(x)=\hat u_{m-3}(x)=
\\
=\begin{cases}
u^{(m)}(x,t)
-(-1)^m
\frac{is_{m-2}}{2\sqrt{\pi}}
(F''(\lambda_1))^{-\frac{1}{2}}
e^{-2F(\lambda_1)}
(1+{\mathcal O}(x^{-1/6})),
\\
\hat u^{(m-3)}(x,t)
+(-1)^m
\frac{is_{m}}{2\sqrt{\pi}}
(-F''(\lambda_1))^{-\frac{1}{2}}
e^{2F(\lambda_1)}
(1+{\mathcal O}(x^{-1/6})),
\end{cases}
\\
\hfill
(x,t)\in\omega_m\cup\hat\omega_{m-3},
\\
\mathop{\mapsto}_{t=0}
\begin{cases}
u^{(m)}(x)
+s_{m-2}
A_-
i^m
x^{-\frac{1}{4}}
e^{(-1)^{m+1}\frac{6}{7}x^{7/6}h_-}
(1+{\mathcal O}(x^{-1/6})),
\\
\hat u^{(m-3)}(x)
-s_{m}
A_-
i^{m}
x^{-\frac{1}{4}}
e^{(-1)^{m+1}\frac{6}{7}x^{7/6}h_-}
(1+{\mathcal O}(x^{-1/6})),
\end{cases}
\\
\hfill
x\in\omega_m\cup\hat\omega_{m-3}.
\end{multline}

Similarly, the intersection of the families $\hat u_{m-3}(x)$
and $u_{m+1}(x)$ is the 1-parameter family corresponding to the 
Stokes multipliers
\begin{equation}\label{u_m+1_hat_u_m-3_Stokes}
s_{m}=s_{m+2}=s_{m-1}=0,\quad
s_{m-2}=s_{m-3}=s_{m+1}+s_{m+3}=-i,
\end{equation}
with the asymptotics as $x\to\infty$,
$x\in\hat\omega_{m-3}\cup\omega_{m+1}$,
\begin{multline}\label{hat_u_m-3_u_m+1_as}
\hat u_{m-3}(x)=u_{m+1}(x)=
\\
=
\begin{cases}
\hat u^{(m-3)}(x,t)
+(-1)^m
\frac{is_{m+1}}{2\sqrt{\pi}}
(-F''(\lambda_3))^{-1/2}
e^{2F(\lambda_3)}
(1+{\mathcal O}(x^{-1/6})),
\\
u^{(m+1)}(x,t)
+(-1)^m
\frac{is_{m+3}}{2\sqrt{\pi}}
(F''(\lambda_3))^{-1/2}
e^{-2F(\lambda_3)}
(1+{\mathcal O}(x^{-1/6})),
\end{cases}
\\
\hfill
(x,t)\in\hat\omega_{m-3}\cup\omega_{m+1},
\\
\mathop{\mapsto}_{t=0}
\begin{cases}
\hat u^{(m-3)}(x)
+s_{m+1}
A_+
i^{m+1}
x^{-\frac{1}{4}}
e^{(-1)^{m}\frac{6}{7}x^{7/6}h_+}
(1+{\mathcal O}(x^{-1/6})),
\\
u^{(m+1)}(x)
-s_{m+3}
A_+
i^{m+1}
x^{-\frac{1}{4}}
e^{(-1)^{m}\frac{6}{7}x^{7/6}h_+}
(1+{\mathcal O}(x^{-1/6})),
\end{cases}
\\
\hfill
x\in\hat\omega_{m-3}\cup\omega_{m+1}.
\end{multline}
\begin{rem}
Observe the different choice of the branches of $F(\lambda_{1,3})$
and $(F''(\lambda_{1,3}))^{-1/2}$ in the adjacent domains.
\end{rem}

\subsection{``Tritronqu\'ee" solutions}

A simple investigation of all possible intersections of three of the 
families $u_m(x,t)$ and $\hat u_n(x,t)$ shows that there exist 
two types of the 0-parameter solutions.

\subsubsection{Type I tritronqu\'ee solutions}

There are seven 0-pa\-ra\-me\-ter solutions $V_{m-3}(x,t)$
corresponding to the intersections of the 2-parameter families 
$u_m(x,t)$, $\hat u_{m-3}(x,t)$, $u_{m+1}(x,t)$. These intersections
are characterized by the Stokes multipliers
\begin{equation}\nonumber
s_{m-1}=s_{m}=s_{m+1}=s_{m+2}=0,\quad
s_{m-3}=s_{m-2}=s_{m+3}=-i,
\end{equation}
and their asymptotics are described as follows:
\begin{multline}\label{V_m-3_as}
V_{m-3}(x,t)=
\\
=
\begin{cases}
u^{(m)}(x,t)
-(-1)^m
\frac{1}{2\sqrt{\pi}}
(F''(\lambda_1))^{-1/2}
e^{-2F(\lambda_1)}
(1+{\mathcal O}(x^{-1/6})),
\\
\hat u^{(m-3)}(x,t),
\\
u^{(m+1)}(x,t)
+(-1)^m
\frac{1}{2\sqrt{\pi}}
(F''(\lambda_3))^{-1/2}
e^{-2F(\lambda_3)}
(1+{\mathcal O}(x^{-1/6})),
\end{cases}
\\
\hfill
(x,t)\in\omega_m\cup\hat\omega_{m-3}\cup\omega_{m+1},
\\
\mathop{\mapsto}_{t=0}
\begin{cases}
u^{(m)}(x)
-A_-
i^{m+1}
x^{-1/4}
e^{(-1)^{m+1}\frac{6}{7}x^{7/6}h_-}
(1+{\mathcal O}(x^{-1/6})),
\\
\hat u^{(m-3)}(x),
\\
u^{(m+1)}(x)
-
A_+
i^{m}
x^{-1/4}
e^{(-1)^{m}\frac{6}{7}x^{7/6}h_+}
(1+{\mathcal O}(x^{-1/6})),
\end{cases}
\\
\hfill
x\in\omega_m\cup\hat\omega_{m-3}\cup\omega_{m+1}.
\end{multline}
In Figure~\ref{fig12}, we present the sectors of the algebraic
asymptotic behavior of the tritronqu\'ee solutions of Type~I
at $t=0$. 
\begin{figure}[htb]
\begin{center}
\mbox{\epsfig{figure=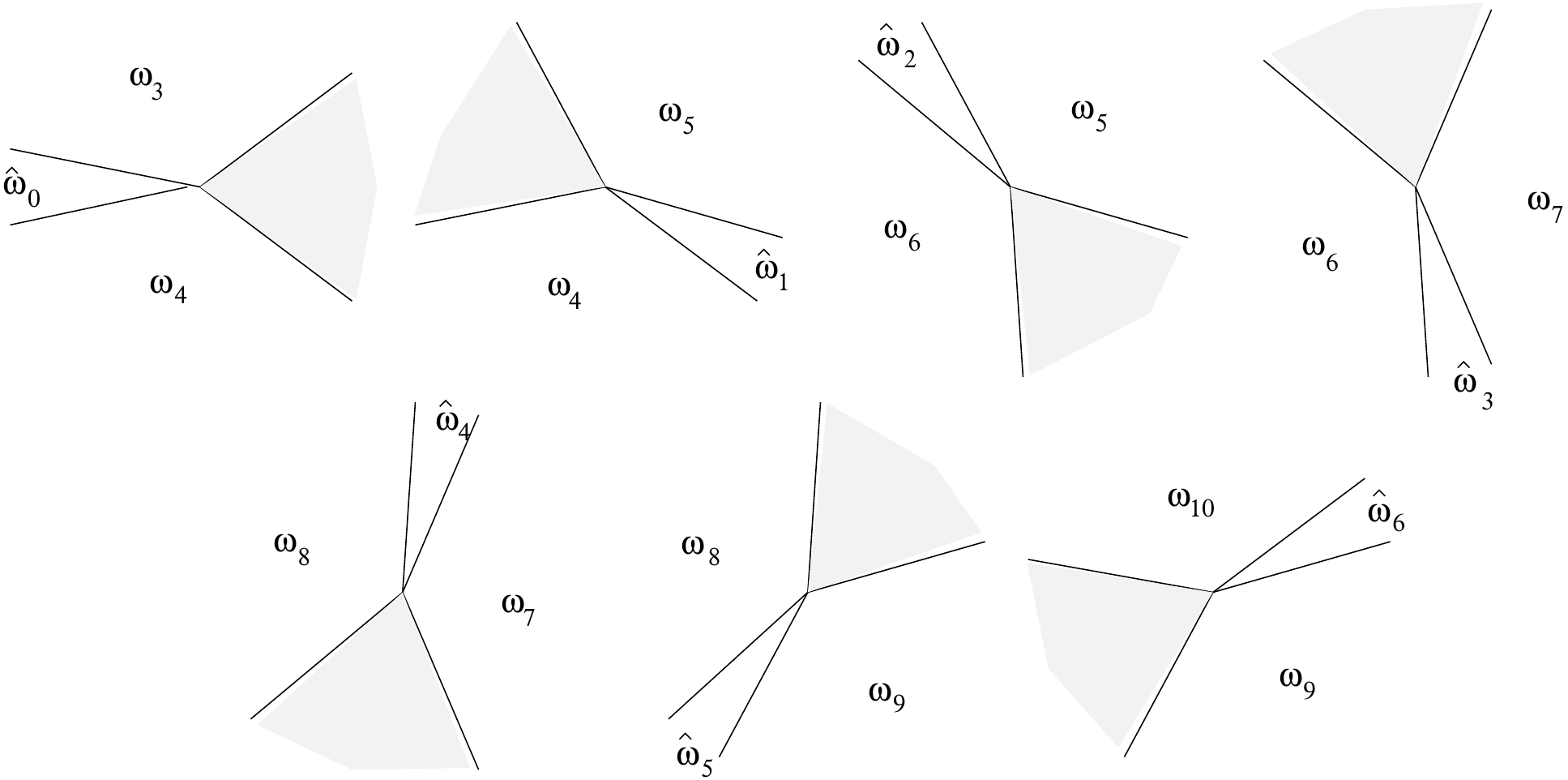,width=0.8\textwidth}}
\end{center}
\caption{Domains $\omega_m$, $\hat\omega_{m-3}$ and $\omega_{m+1}$
for the algebraic asymptotic behavior of tritronqu\'ee Type~I 
solutions for $t=0$. 
The sectors of the elliptic asymptotic behavior are shown in grey.}
\label{fig12}
\end{figure}

\subsubsection{Type II tritronqu\'ee solutions}
There exist seven intersections $U_{m}(x,t)$ of the families
$\hat u_{m-4}(x,t)$, 
$u_{m}(x,t)$ and $\hat u_{m-3}(x,t)$. These 0-parameter solutions 
correspond to the Stokes multipliers
\begin{equation}\nonumber
s_{m-2}=s_{m-1}=s_{m+1}=s_{m+2}=0,\quad
s_{m+3}=s_{m-3}=s_{m}=-i,
\end{equation}
and have the asymptotics
\begin{multline}\label{U_m_as}
U_{m}(x,t)=
\\
=
\begin{cases}
\hat u^{(m-4)}(x,t)
-(-1)^m
\frac{1}{2\sqrt{\pi}}
(-F''(\lambda_3))^{-1/2}
e^{2F(\lambda_3)}
(1+{\mathcal O}(x^{-1/6})),
\\
u^{(m)}(x,t),
\\
\hat u^{(m-3)}(x,t)
+(-1)^m
\frac{1}{2\sqrt{\pi}}
(-F''(\lambda_1))^{-1/2}
e^{2F(\lambda_1)}
(1+{\mathcal O}(x^{-1/6})),
\\
\hat u^{(m)}(x,t)
-(-1)^m
\frac{1}{2\sqrt{\pi}}
(-F''(\lambda_1))^{-1/2}
e^{2F(\lambda_1)}
(1+{\mathcal O}(x^{-1/6}))
\\
-(-1)^m
\frac{1}{2\sqrt{\pi}}
(-F''(\lambda_3))^{-1/2}
e^{2F(\lambda_3)}
(1+{\mathcal O}(x^{-1/6})),
\end{cases}
\\
\hfill
(x,t)\in\hat\omega_{m-4}\cup\omega_m
\cup\hat\omega_{m-3}\cup\hat\omega_m,
\\
\mathop{\mapsto}_{t=0}
\begin{cases}
\hat u^{(m-4)}(x)
-A_+
i^{m+1}
x^{-1/4}
e^{(-1)^{m+1}\frac{6}{7}x^{7/6}h_+}
(1+{\mathcal O}(x^{-1/6})),
\\
u^{(m)}(x),
\\
\hat u^{(m-3)}(x)
+A_-
i^{m+1}
x^{-1/4}
e^{(-1)^{m+1}\frac{6}{7}x^{7/6}h_-}
(1+{\mathcal O}(x^{-1/6})),
\\
\hat u^{(m)}(x)
+A_-
i^{m}
x^{-1/4}
e^{(-1)^m\frac{6}{7}x^{7/6}h_-}
(1+{\mathcal O}(x^{-1/6}))
\\
-A_+
i^{m+1}
x^{-1/4}
e^{(-1)^{m+1}\frac{6}{7}x^{7/6}h_+}
(1+{\mathcal O}(x^{-1/6})),
\end{cases}
\\
\hfill
x\in\hat\omega_{m-4}\cup\omega_m
\cup\hat\omega_{m-3}\cup\hat\omega_m.
\end{multline}
In Figure~\ref{fig13}, we present the domains for the algebraic
asymptotic behavior of the tritronqu\'ee solutions $U_m(x)$ of Type~II
at $t=0$.
\begin{figure}[htb]
\begin{center}
\mbox{\epsfig{figure=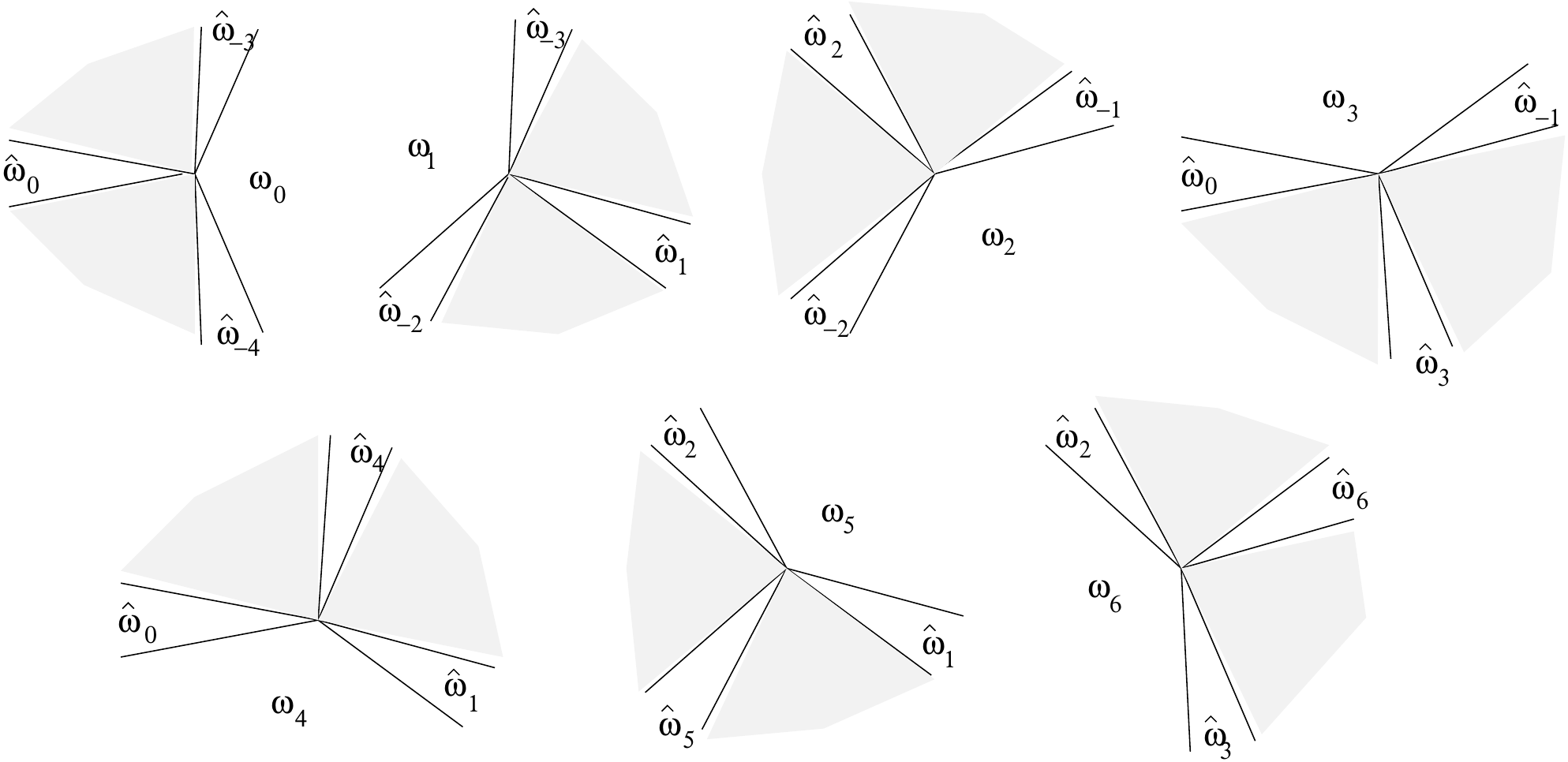,width=0.8\textwidth}}
\end{center}
\caption{Domains $\hat\omega_m$, $\hat\omega_{m-4}$, $\omega_m$
and $\hat\omega_{m-3}$ for the algebraic asymptotic behavior
of tritronqu\'ee Type~II solutions at $t=0$. The sectors of the elliptic
asymptotic behavior are shown in grey.}
\label{fig13}
\end{figure}

The  solution $U_0(x,t)$  real on the real line, with the algebraic asymptotics 
as $x\to\pm\infty$ was found in \cite{K} for  $t=0$ using the isomonodromic
deformation approach. The fact that this solution is regular on the real line  for any $t$
was proved in \cite{CV}.

In Figure~\ref{fig14}, we present the domains
$\hat\omega_0$, $\hat\omega_{-4}$, $\omega_0$
and $\hat\omega_{-3}$ for the algebraic asymptotic behavior
of the real and regular on the real line solution $U_0(x,t)$
as $t\to\mp\infty$ in the complex plane of 
the scaled variable $x|t|^{-3/2}$.
\begin{figure}[htb]
\begin{center}
\mbox{\epsfig{figure=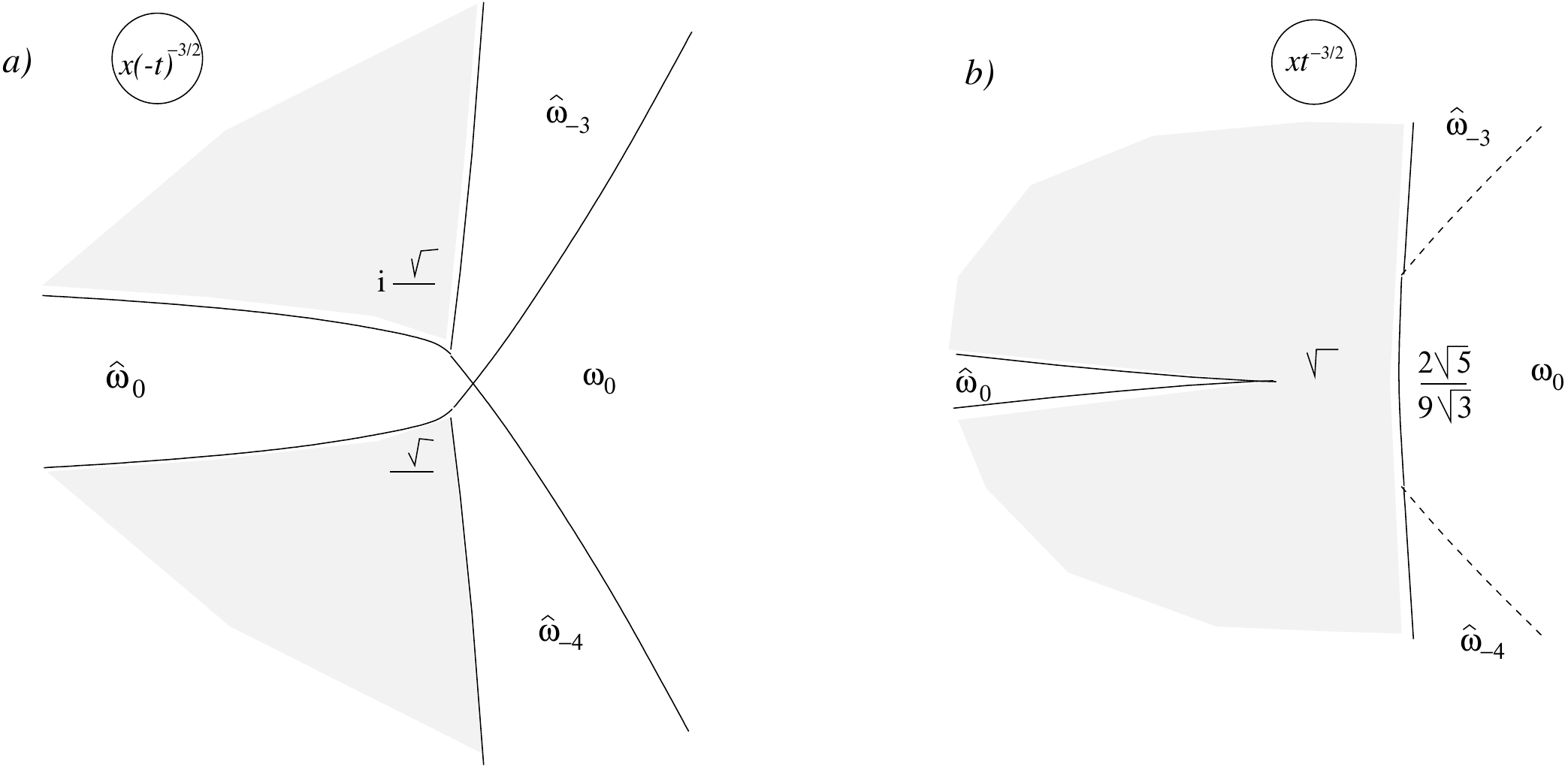,width=0.8\textwidth}}
\end{center}
\caption{The pole-free domains $\hat\omega_0$, $\hat\omega_{-4}$, 
$\omega_0$ and $\hat\omega_{-3}$ of the solution $U_0(x,t)$
in the plane of the scaled variable $x|t|^{-3/2}$: 
a) as $t\to-\infty$; b) as $t\to+\infty$. The sectors of 
the elliptic asymptotic behavior are shown in grey.}
\label{fig14}
\end{figure}
Observe some remarkable properties of the domains. 

As $t\to-\infty$, two domains of the elliptic behavior are separated by
a domain for the algebraic asymptotic behavior. 
In some neighborhood of the vertices $x=\pm i\tfrac{2\sqrt2}{3}(-t)^{3/2}$,
the asymptotics of $U_0(x,t)$ is given in terms of the tritronqu\'ee
solutions of the first Painlev\'e equation \PI.

As $t\to+\infty$, the sectors of the algebraic asymptotic behavior of 
$U_0(x,t)$ are separated by a connected domain of the elliptic behavior,
and in a neighborhood of the vertex $x=-2\sqrt3\,t^{3/2}$, the solution
$U_0(x,t)$ is approximated by the Hastings-McLeod solution of 
the second Painlev\'e equation \PII\ \cite{CG2}. More details on 
the pole distribution in the elliptic sector as $t\to+\infty$ can be 
found in \cite{DK}.

In Figure~\ref{fig14a}, we present the domains
$\omega_3$, $\hat\omega_0$ and $\omega_4$
for the algebraic asymptotic behavior of the real on 
the real line solution $V_0(x,t)$ as $t\to\mp\infty$  
in the complex plane of the scaled variable 
$x|t|^{-3/2}$.
\begin{figure}[htb]
\begin{center}
\mbox{\epsfig{figure=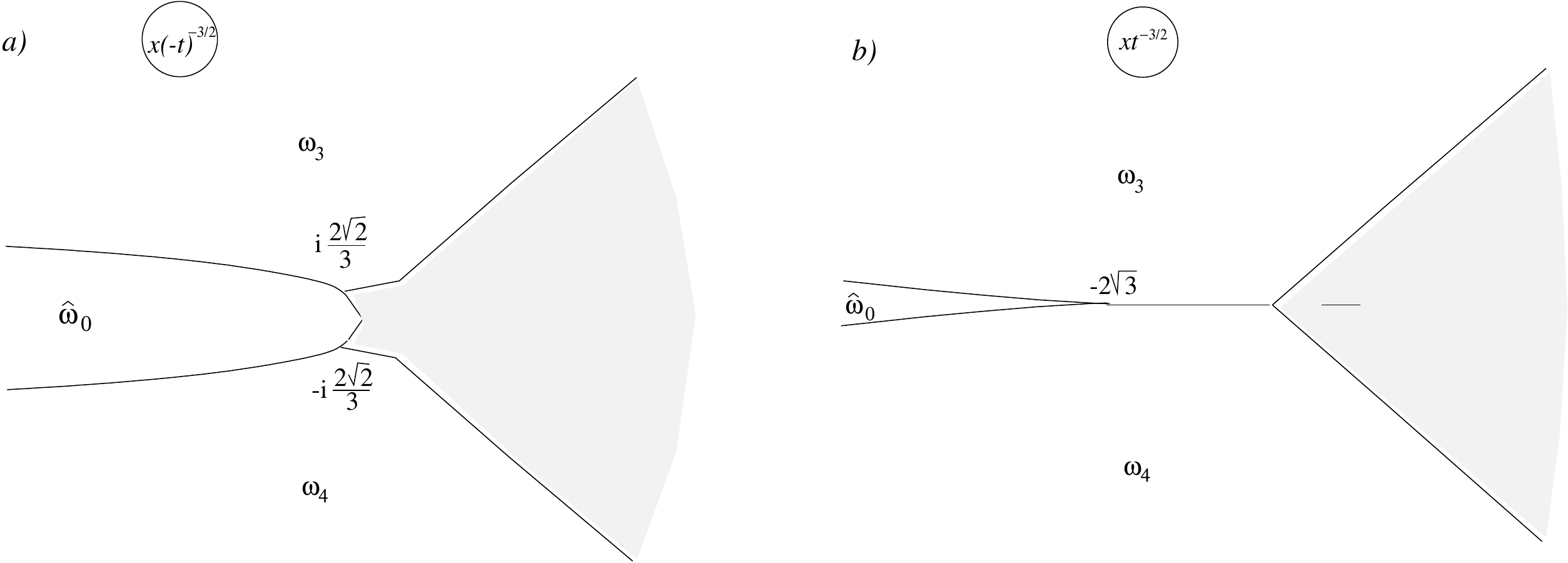,width=0.8\textwidth}}
\end{center}
\caption{The pole-free domains $\omega_3$, $\hat\omega_0$ 
and $\omega_4$ of the solution $V_0(x,t)$
in the plane of the scaled variable $x|t|^{-3/2}$: 
a) as $t\to-\infty$; b) as $t\to+\infty$. The sectors of 
the elliptic asymptotic behavior are shown in grey.}
\label{fig14a}
\end{figure}

In the neighborhood of the vertex $x=\tfrac{2\sqrt2}{3}\,t^{3/2}$
as $t\to+\infty$ and in the neighborhoods of the vertices 
$x=\pm i\tfrac{2\sqrt2}{3}(-t)^{3/2}$ as $t\to-\infty$, 
the asymptotics of $V_0(x,t)$ is given in terms of the tritronqu\'ee
solutions of the first Painlev\'e equation \PI.

\subsection{Quasi-linear Stokes phenomenon}

The notion of the quasi-linear Stokes phenomenon introduced in \cite{IK} 
refers to an exponentially small difference between two analytic functions 
with identical formal power expansions in a common sector of the complex 
plane. In the case of equation \PItwo\ at bounded values of $t$, 
all solutions $u^{(m)}(x,t)$ and $\hat u^{(n)}(x,t)$ have 
the same leading order asymptotics 
$u_f(x,t)\sim-\sqrt[3]{6}x^{1/3}$
as $x\to\infty$ in certain sectors of the complex $x$ plane
and the very same complete formal asymptotic power series 
expansion, $u_f(x,t)=\sum_{n=0}^{\infty}a_n x^{-\frac{1}{3}(n-1)}$,
uniquely determined by the leading order coefficient
$a_0=-\sqrt[3]{6}$ using the recurrence relation below
(the sum is assumed to be empty if the upper bound is less than the lower one),
\begin{multline}\label{um_formal}
u_f(x,t)=\sum_{n=0}^{\infty}a_nx^{-\frac{1}{3}(n-1)},\quad
a_k=
\begin{cases}
0,\quad
k<0,
\\
-\sqrt[3]{6},\quad
k=0,
\\
P_k(a_{n<k}),\quad
k>0,
\end{cases}
\\
P_k(a_{n<k})=
a_0^{-2}
\Bigl(
2ta_{k-2}
-\sum_{m=2}^{k-2}
\tfrac{1}{3}a_0a_m a_{k-m}
-\sum_{n=2}^{k-1}
\sum_{m=0}^{k-n}
\tfrac{1}{3}a_n a_m a_{k-n-m}
\\
-\sum_{n=0}^{k-7}
\tfrac{(n-1)(k+n-4)}{108}
a_n a_{k-n-7}
-\tfrac{(k-15)(k-12)(k-9)(k-6)}{9720}
a_{k-14}
\Bigr),\quad
k>0.\quad
\end{multline}
If $t=0$,  the formal series simplifies,
\begin{multline}\label{um_formal_t=0}
u_f(x,t)=\sum_{n=0}^{\infty}b_nx^{-\frac{1}{3}(7n-1)},\quad
b_0=-\sqrt[3]{6},\quad
b_1=\tfrac{1}{36},\quad
\\
b_{n+1}
=
-\sum_{m=0}^{n-1}
\tfrac{1}{3}
b_0^{-1}
b_{n-m}
b_{m+1}
-\sum_{l=0}^{n-1}
\sum_{m=0}^{n-l}
\tfrac{1}{3}
b_0^{-2}
b_{n-m-l}
b_{l+1}
b_m
\\
-\sum_{m=0}^{n}
\tfrac{(7m-1)(7n+7m+3)}{3^3\cdot2^2}
b_0^{-2}
b_m
b_{n-m}
-\tfrac{(7n-8)(7n-5)(7n-2)(7n+1)}{5\cdot3^5\cdot2^3}
b_0^{-2}
b_{n-1}
,\quad
n\geq1.
\end{multline}

The exponentially small differences between the solutions
$u^{(m)}(x)$ and $\hat u^{(n)}(x)$ with the identical expansions 
(\ref{um_formal}) can be extracted from the asymptotics of 
the bitronqu\'ee (\ref{um_hat_u_m-3_as}), (\ref{hat_u_m-3_u_m+1_as})
or tritronqu\'ee (\ref{V_m-3_as}), (\ref{U_m_as}) solutions above. 
Indeed, these formulas can be understood as the mutual analytic 
continuations of $u_m(x)$ and $\hat u_n(x)$ into adjacent sectors
of the complex $x$ plane. For instance, taking into account the 
relation $s_{m-2}+s_m=-i$ in (\ref{um_hat_u_m-3_Stokes}), the 
asymptotics (\ref{um_hat_u_m-3_as}) is rewritten as
\begin{multline}\label{m-3_m_quasi_Stokes}
\hat u^{(m-3)}(x,t)-u^{(m)}(x,t)=
\tfrac{(-1)^{m+1}}{2\sqrt{\pi}}
(F''(\lambda_1))^{-\frac{1}{2}}
e^{-2F(\lambda_1)}
(1+{\mathcal O}(x^{-1/6}))
\\
\mathop{\mapsto}_{t=0}
A_-i^{m-1}x^{-\frac{1}{4}}e^{(-1)^{m+1}\frac{6}{7}x^{7/6}h_-}
(1+{\mathcal O}(x^{-1/6})),\quad
x\in\omega_m\cup\hat\omega_{m-3}.
\end{multline}
Similarly, (\ref{hat_u_m-3_u_m+1_as}) with the use of
the relation $s_{m+1}+s_{m+3}=-i$ from
(\ref{u_m+1_hat_u_m-3_Stokes}) yield
\begin{multline}\label{m_m+1_quasi_Stokes}
u^{(m+1)}(x,t)-\hat u^{(m-3)}(x,t)
=\tfrac{(-1)^{m+1}}{2\sqrt{\pi}}
(F''(\lambda_3))^{-\frac{1}{2}}
e^{-2F(\lambda_3)}
(1+{\mathcal O}(x^{-1/6}))
\\
\mathop{\mapsto}_{t=0}
A_+
i^{m}
x^{-\frac{1}{4}}
e^{(-1)^{m}\frac{6}{7}x^{7/6}h_+}
(1+{\mathcal O}(x^{-1/6})),\quad
x\in\omega_{m+1}\cup\hat\omega_{m-3}.
\end{multline}

\section{Coefficient asymptotics for $t=o(x^{2/3})$}

The intimate relation between the Stokes phenomenon and the coefficient
asymptotics in the formal solutions to linear and nonlinear ODEs is well
explained in \cite{FIKN}, see also \cite{IK, kap_p1_quasi}. The most recent
developments of the coefficient asymptotics evaluation method 
including numerical tests can be found in \cite{GIKM}. 
Here we develop similar ideas for the formal expansion 
(\ref{um_formal}) which admits an asymptotic interpretation for 
$tx^{-2/3}\to0$ as $x\to\infty$. 
Our approach follows the above mentioned papers and
significantly differs from the method developed in the framework of the resurgent analysis, see e.g.\ \cite{Mar, SchV} and references mentioned
therein.

Introduce the piece-wise meromorphic function
as the collection of the 14 tritronqu\'ee solutions
defined above,
\begin{multline}\label{W_coeff_def}
W(\tau,t)=
\begin{cases}
u^{(m)}(\tau^3,t),\quad
\arg\tau\in[-\frac{\pi}{7}+\frac{1}{3}\beta_0+\frac{2\pi}{7}m,
\frac{\pi}{7}-\frac{1}{3}\beta_0+\frac{2\pi}{7}m],
\\
\hat u^{(m-3)}(\tau^3,t),\quad
\arg\tau\in[\frac{\pi}{7}-\frac{1}{3}\beta_0+\frac{2\pi}{7}m,
\frac{\pi}{7}+\frac{1}{3}\beta_0+\frac{2\pi}{7}m],
\end{cases}
\\
\beta_0=\tfrac{3}{7}\arctan\tfrac{1}{\sqrt5},\quad
m=0,1,\dots,6.
\end{multline}
This function has the uniform asymptotic expansion as 
$\tau\to\infty$,
\begin{equation}\nonumber
W(\tau,t)=\sum_{n=0}^{M}a_n(t)\tau^{-n+1}
+{\mathcal O}(\tau^{-M}),
\end{equation}
with the coefficients $a_n$ defined in (\ref{um_formal}).
Note that $W(\tau)$ can have a {\em finite} number of poles 
all of which are located in the interior of the disc $|\tau|<r_0$.

Integrating the product
\begin{equation}\nonumber
\tau^{N-2}W(\tau,t)=
P_{N-1}(\tau,t)
+a_{N}(t)\tau^{-1}
+{\mathcal O}(\tau^{-2}),\quad
P_{N-1}(\tau,t)=\sum_{n=0}^{N-1}a_n\tau^{N-1-n},
\end{equation}
over a counter-clock-wise oriented circle of large radius 
$|\tau|=r\gg r_0$, we find
\begin{equation}\nonumber
a_{N}(t)=
\tfrac{1}{2\pi i}
\oint_{C_r}\tau^{N-2}W(\tau,t)\,d\tau
+{\mathcal O}(r^{-1}).
\end{equation}
Contracting the arcs of the circle $C_r$ in the interior of the sectors
in (\ref{W_coeff_def}) to the circle of radius $|\tau|=r_0$, 
and using (\ref{m_m+1_quasi_Stokes}), (\ref{m-3_m_quasi_Stokes}), we compute
\begin{multline*}
a_{N}(t)
=\tfrac{1}{2\pi i}
\sum_{m=0}^6
\int_{e^{-i\frac{1}{3}\alpha_0+i\frac{2\pi}{7}m}(r_0,r)}
\tau^{N-2}(\hat u^{(m-4)}(\tau^3,t)-u^{(m)}(\tau^3,t))\,
d\tau
\\
+\tfrac{1}{2\pi i}
\sum_{m=0}^6
\int_{e^{i\frac{1}{3}\alpha_0+i\frac{2\pi}{7}m}(r_0,r)}
\tau^{N-2}(u^{(m)}(\tau^3,t)-\hat u^{(m-3)}(\tau^3,t))\,
d\tau
\\
+\tfrac{1}{2\pi i}
\oint_{C_{r_0}}\tau^{N-2}W(\tau,t)\,d\tau
+{\mathcal O}(r^{-1})=
\\
=-\tfrac{1}{4\pi^{3/2} i}
\sum_{m=0}^6
e^{i\frac{2\pi}{7}Nm}
\int_{e^{-i\frac{1}{3}\alpha_0}(r_0,r)}
\tau^{N-2}
(F''(\lambda_3))^{-\frac{1}{2}}
e^{-2F(\lambda_3)}
(1+{\mathcal O}(\tau^{-\frac{1}{2}}))
\bigr|_{\genfrac{}{}{0pt}{}{x\mapsto\tau^3}{t\mapsto t\exp(-i 18\pi m/7)}}
\,
d\tau
\\
+\tfrac{1}{4\pi^{3/2}i}
\sum_{m=0}^6
e^{i\frac{2\pi}{7}Nm}
\int_{e^{i\frac{1}{3}\alpha_0}(r_0,r)}
\tau^{N-2}
(F''(\lambda_1))^{-\frac{1}{2}}
e^{-2F(\lambda_1)}
(1+{\mathcal O}(\tau^{-\frac{1}{2}}))
\bigr|_{\genfrac{}{}{0pt}{}{x\mapsto\tau^3}{t\mapsto t\exp(-i18\pi m/7)}}
\,
d\tau
\\
+{\mathcal O}(r_0^{N-1})
+{\mathcal O}(r^{-1})
.
\end{multline*}
As $xt^{-3/2}\to\infty$, we find the asymptotics
\begin{multline}\label{F13_ddF13_as}
F(\lambda_{1,3})=
\tfrac{1}{2}
\tfrac{6}{7}x^{7/6}
h_{\mp}
\pm i(\tfrac{15}{2})^{1/4}
e^{\mp i\frac{3}{2}\arctan\frac{1}{\sqrt5}}
t
x^{1/2}
+{\mathcal O}(x^{-1/6}),
\\
\shoveleft{
(F''(\lambda_{1,3}))^{-1/2}=
\pm i2\sqrt{\pi}
A_{\mp}
x^{-1/4}
(1+{\mathcal O}(tx^{-2/3})),
}\hfill
\end{multline}
where $h_{\pm}$ and $A_{\pm}$ are defined in (\ref{h_pm_def}),
letting $r=\infty$ and using conventional steepest descent 
analysis of the integrals,
\begin{multline}\label{aN_as}
a_N(t)
=\tfrac{1}{2\pi}
\sum_{\sigma\in\{+,-\}}
A_{-\sigma}
\sum_{m=0}^6
e^{i\frac{2\pi}{7}Nm}
\int_{e^{\sigma i\frac{1}{3}\alpha_0}(r_0,\infty)}
\exp\bigl\{
-\tfrac{6}{7}\tau^{7/2}
h_{-\sigma}
\\
-\sigma i2^{3/4}15^{1/4}
e^{-\sigma i\frac{3}{2}\arctan\frac{1}{\sqrt5}}
te^{-i\frac{18\pi}{7}m}
\tau^{3/2}
+(N-\tfrac{11}{4})\ln\tau
\bigr\}
(1+{\mathcal O}(\tau^{-\frac{1}{2}}))
\,
d\tau
\\
+{\mathcal O}(r_0^{N-1})=
\\
=\tfrac{1}{2\sqrt{7\pi}}
e^{(\frac{2}{7}N-1)\ln(N-\frac{11}{4})}
e^{-(\frac{2}{7}N-\frac{11}{14})}
\bigl(
5^{\frac{1}{2}}
3^{\frac{17}{6}}
2^{\frac{11}{6}}
\bigr)^{-\frac{1}{7}N+\frac{1}{4}}
\times
\\
\times
\sum_{\sigma\in\{+,-\}}
A_{-\sigma}
e^{\sigma i(\frac{1}{7}N-\frac{1}{4})\arctan\frac{1}{\sqrt5}}
\sum_{m=0}^6
e^{i\frac{2\pi}{7}Nm
-it
N^{3/7}
b_{\sigma,m}
}
(1+{\mathcal O}(N^{-1/7}))
+{\mathcal O}(r_0^{N-1}),
\\
b_{\sigma,m}=
\sigma 
5^{\frac{1}{7}}
3^{-\frac{5}{14}}
2^{\frac{5}{14}}
e^{-\sigma i\frac{9}{7}\arctan\frac{1}{\sqrt5}}
e^{-i\frac{18\pi}{7}m},\quad
N\to\infty.
\end{multline}
Observe the agreement of (\ref{aN_as}) with the triviality
of the coefficients $a_n$ unless $n\equiv0\mod(7)$ at $t=0$,
see (\ref{um_formal_t=0}).

\section{Numerical evaluation of regular solutions to the \PItwo\
equation}

In this section we present a numerical approach to pole-free 
solutions to \PItwo\ to illustrate some of the results of 
the previous sections.

\subsection{Numerical Methods}
We will study here the special solutions to the equation \PItwo\ 
called tritronqu\'ee in the previous sections. The type~I solution 
denoted by $V_{0}(x,t)$ is similar to the tritronqu\'ee solution 
of the \PI\ equation (see \cite{DubrovinGravaKlein} for figures). 
The type~II solution denoted by $U_{0}(x,t)$ is real and pole-free 
on the real axis. Here we are interested in the sectors of the
complex $x$ plane where these solutions are regular and exhibit
the algebraic asymptotic behavior,
\begin{equation}
u\sim-\sqrt[3]{6}\,x^{1/3}
\quad
\mbox{ as }
\quad
|x|\to\infty,
\label{asym}
\end{equation}
in different sectors of the complex plane,
see Figures~\ref{fig13}, \ref{fig14} and
Figures~\ref{fig12}, \ref{fig14a}.

In the literature, it is possible to find various numerical approaches 
to solutions of Painlev\'e-type equations. For instance,
if a Painlev\'e transcendent can be represented in terms of
a Fredholm determinant, it is possible to apply the methods of 
\cite{fredholm}. Unfortunately such an expression does not yet 
exist for the solution $U_0(x,t)$. 

It is known \cite{shimomura} that equation \PItwo\ possesses
the Painlev\'e property, thus all its solutions are meromorphic
functions of $x$ and $t$. A convenient approach to study numerically 
meromorphic functions are Pad\'e approximants, see \cite{FW} for the 
tritronqu\'ee solution of \PI. A disadvantage of this approach is a 
lack of error control for the Pad\'e approximants. 
It might be more promising to solve numerically 
the  Riemann-Hilbert problem as in \cite{olver}.

Here we concentrate on the pole free sectors for solutions to \PItwo\ as 
in \cite{GK1}. The idea of our numeric method is the formulation
of a boundary value problem for \PItwo\ on a {\em finite} interval
consistent with the asymptotic condition (\ref{asym}).
First we construct the series (\ref{um_formal}) in the form
\begin{equation}\label{pain3}
u=Y+\sum_{n=1}^{\infty}c_{n}Y^{-n},\quad
Y=-\sqrt[3]{6}\,x^{1/3}.
\end{equation}
We find the non-zero coefficients
$c_1=2t$, $c_2=c_3=c_4=0$, $c_5=-\tfrac{8}{3}t^3$, $c_6=1$,
$c_7=\tfrac{16}{3}t^4$, $c_8=-\tfrac{10}{3}t$, $c_9=0$,
$c_{10}=-\tfrac{28}{3}t^2$, $c_{11}=-\tfrac{256}{9}t^6$,  
$c_{12}=96t^3$, $c_{13}=\tfrac{640}{9}t^7-21$, 
$c_{14}=-\tfrac{1936}{9}t^4$, \dots 
We truncate this formal series at the $n$-th term for which
$|c_nY^{-n}|<10^{-6}$ at the boundary of the computational 
domain $x\in[x_{l},x_{r}]$. At the values $x_{l}$ and $x_{r}$, 
the truncated series (\ref{pain3}) provides us with the necessary 
boundary data, and we obtain a  {\em boundary} value problem 
which replaces the original {\em asymptotic} value problem.

The standard approach to boundary value problems for an ODE 
is to choose a suitable discretization of the dependent variable, see 
for instance \cite{trefethen,chebop}. 
This leads to an approximation of the derivatives in terms of 
so called differentiation matrices. In \cite{GK1}, we used 
a collocation method with cubic splines distributed as \emph{bvp4} 
with Matlab. In \cite{GK12}, we applied a Chebyshev collocation 
method on Chebyshev collocation points $x_{j}$, $j=0,\ldots,N_{c}$. 
This is related to an expansion of the solution in terms of Chebyshev 
polynomials. Since the derivative of a Chebyshev polynomial can be 
again expressed in terms of a linear combination of Chebyshev 
polynomials, this leads to the well known Chebyshev differentiation 
matrices, see for instance \cite{trefethen}. The \PItwo\ equation 
(\ref{p12}) is thus replaced by $N_{c}+1$ algebraic equations. 
The boundary data are included via the so-called $\tau$-method: 
the equations for $j=0,1,N_{c}-1,N_{c}$ are replaced by the 
boundary conditions following from (\ref{pain3}). 

The resulting system of algebraic equations is solved using 
Newton's method with the initial iterate 
$u=-6^{1/3}x/(1+x^{2})^{1/3}$, a smooth function which satisfies 
for large $x$ the asymptotic conditions, or a linear interpolation 
between the boundary data. For highly oscillatory 
solutions, the iteration in general fails to converge. 
Thus we use a Newton-Armijo method, see \cite{kelley, armijo} 
and references therein for details. 

The normal Newton iteration for the solution of an equation 
$F(v_{n})=0$ takes the form
\begin{equation*}
v_{n+1}=v_{n}-(\mbox{Jac}F(v_{n}))^{-1}F(v_{n}),
\end{equation*} 
where $\mbox{Jac} F$ is the Jacobian of $F$. The basic idea is 
to check at each step of the iteration whether the new iterate 
$v_{n+1}$ satisfies the equation better than the previous one, 
i.e.\ whether $\|F(v_{n+1})\|<\|F(v_{n})\|$. If this is not 
the case, a so called line search is performed, i.e.\ the new 
iterate is taken as 
$v_{n+1}=v_{n}-\lambda(\mbox{Jac}F(v_{n}))^{-1}F(v_{n})$ where 
$0<\lambda<1$. In practice we choose a quadratic model for 
$F(v_{n-1})$, $F(v)$ and $F(v_{n+1})$ to optimize the choice of 
$\lambda$ as discussed in \cite{kelley}. For highly oscillatory 
solutions it can happen that there is no $\lambda$ satisfying 
the condition. In this case we take a $\lambda$ of the order of 
$10^{-5}$ and continue the 
iteration. In the shown examples below, the solution will converge 
after some iterations even if the line search failed at one point. 
The precision is mainly limited by the conditioning of the Chebyshev 
differentiation matrix which is of the order of $N_{c}^{2}$, see the 
discussion in \cite{trefethen} for the second order differentiation 
matrices. We use here $N_{c}=512$ or $N_{c}=1024$ and reach an 
accuracy of the order of $10^{-6}$.

To study solutions in the complex plane, i.e.\ as a holomorphic function 
of the complex coordinate $x$ in a given sector, we proceed as in 
\cite{DubrovinGravaKlein}: we determine the solution as discussed above 
on a line given by $x=\xi e^{i\phi}$ with $\xi,\phi\in\mathbb{R}$. 
This solution is then used as the boundary data for the two-dimensional 
Laplace equation which is solved as discussed in 
\cite{DubrovinGravaKlein,trefethen}. Due to the coordinate 
singularity of the solution for $|x|=0$, the precision is much lower 
than for the solution on a line in the complex plane and serves
mainly for visualization. 

\subsection{Type~I tritronqu\'ee solutions}
We first study the type~I tritronqu\'ee solution $V_{0}$ which is a 
straightforward generalization of the well known real tritronqu\'ee
solution of \PI\ studied numerically in \cite{DubrovinGravaKlein,FW}.
This similarity allows us to concentrate on the most interesting details. 

We first determine the solution on the imaginary axis where we put 
$x=\exp(i\phi)\xi$, $\xi\in\mathbb{R}$ and $\phi=5\pi/2$. We get the
solution shown in Fig.~\ref{pI2triphi5pi2b0t0}
\begin{figure}[htb!]
\begin{center}
\includegraphics[width=0.46\textwidth]{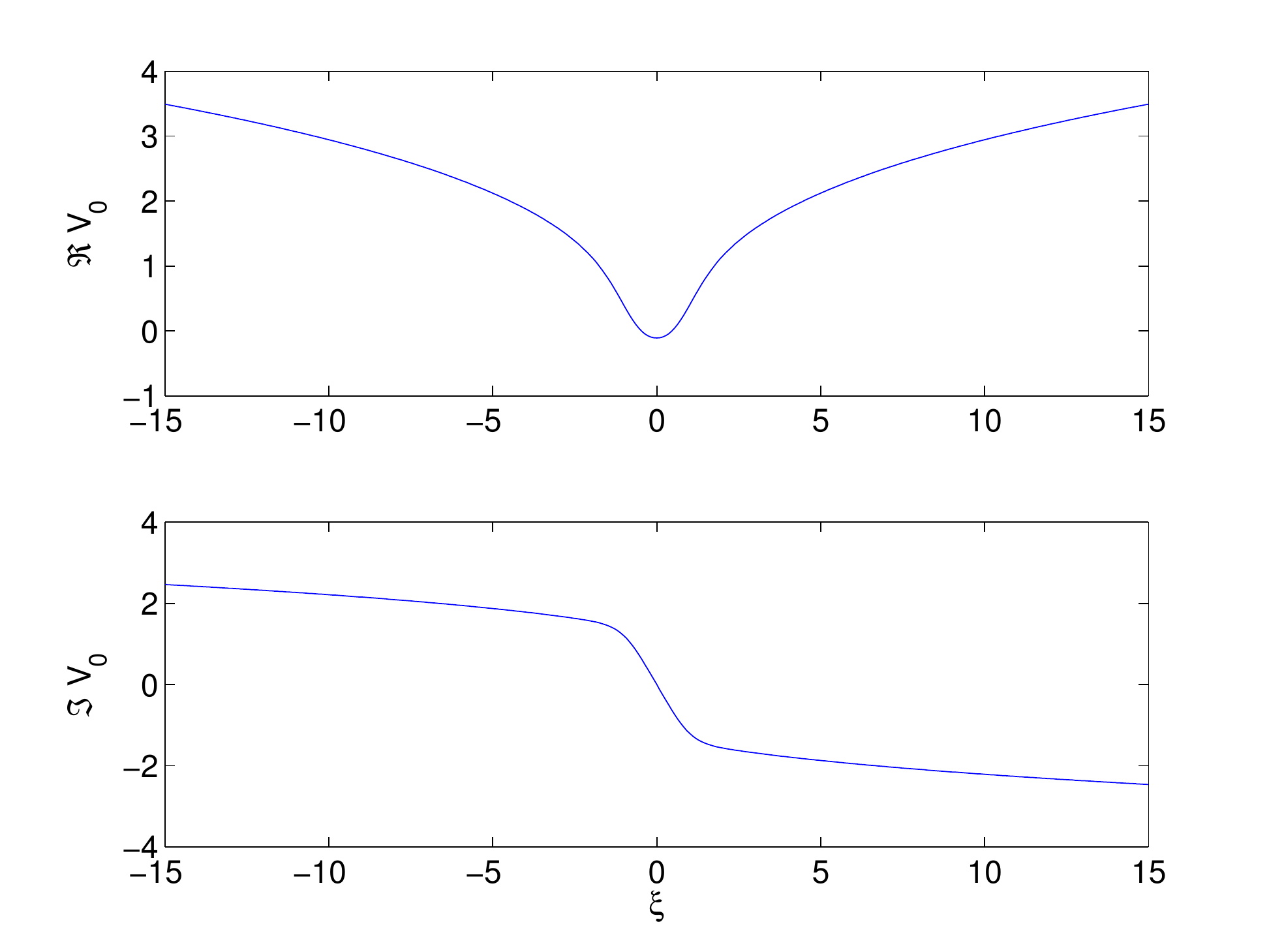}
\includegraphics[width=0.46\textwidth]{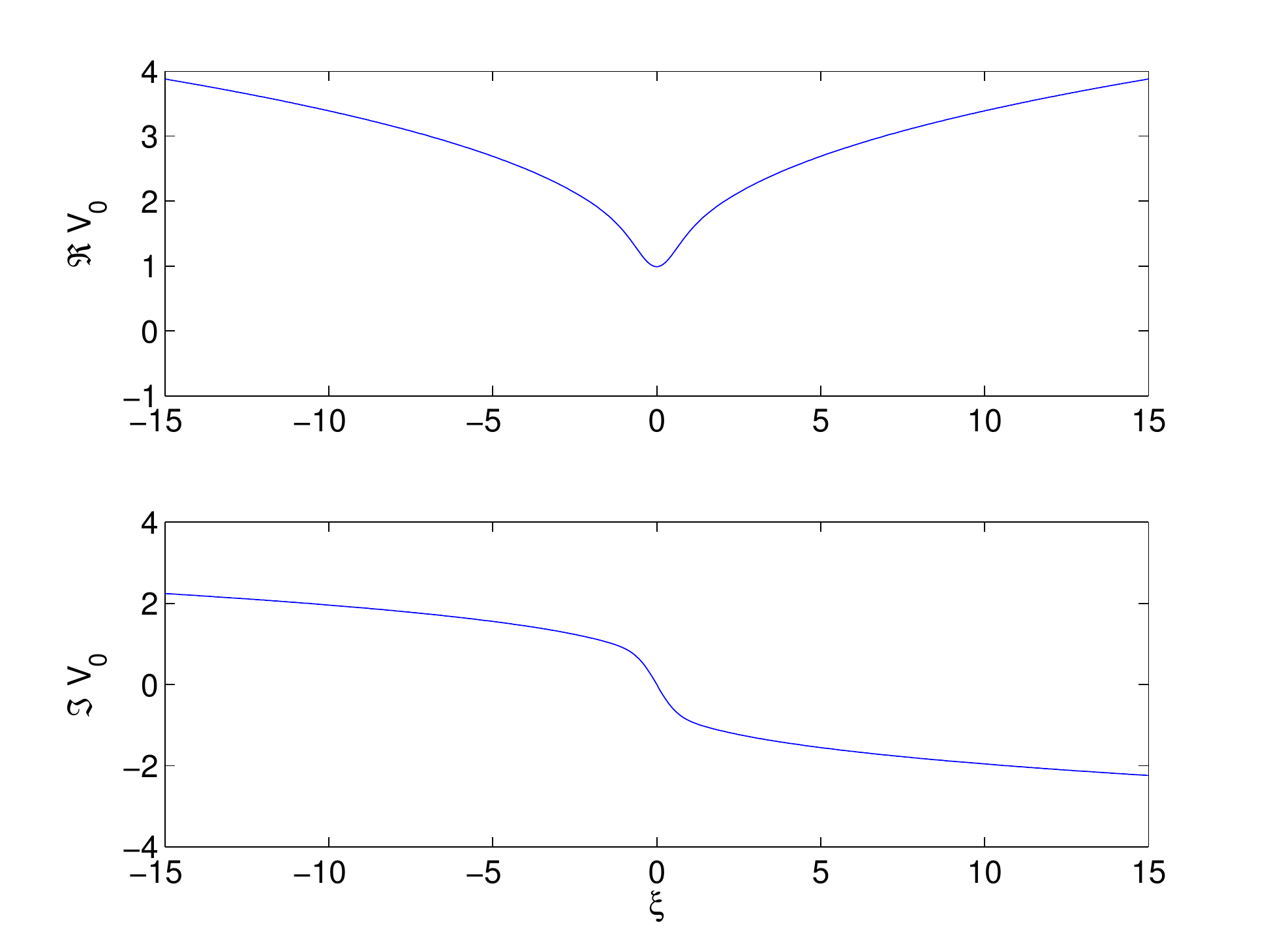}
\end{center}
 \caption{Solution $V_0(x,t)$ to \PItwo\ satisfying the asymptotic 
condition (\ref{asym}) for  $x=\exp(i\phi)\xi$, 
$\xi\in\mathbb{R}$ and $\phi=5\pi/2$; on the left for $t=-1$, on the right for $t=0$.}
\label{pI2triphi5pi2b0t0}
\end{figure}

If we consider the same solution for $t=0$ on a parallel to the imaginary axis, 
$x=\exp(i\phi)\xi+b$, $\xi,b\in\mathbb{R}$ and $\phi=5\pi/2$, the 
solution becomes more peaked for positive $b$ as can be seen in 
Fig.~\ref{pI2triphi0b55t0}. It appears as if this peak will become 
a pole for larger $b$, but this cannot be determined with the used 
numerical methods. 
\begin{figure}[htb!]
\begin{center}
\includegraphics[width=0.7\textwidth]{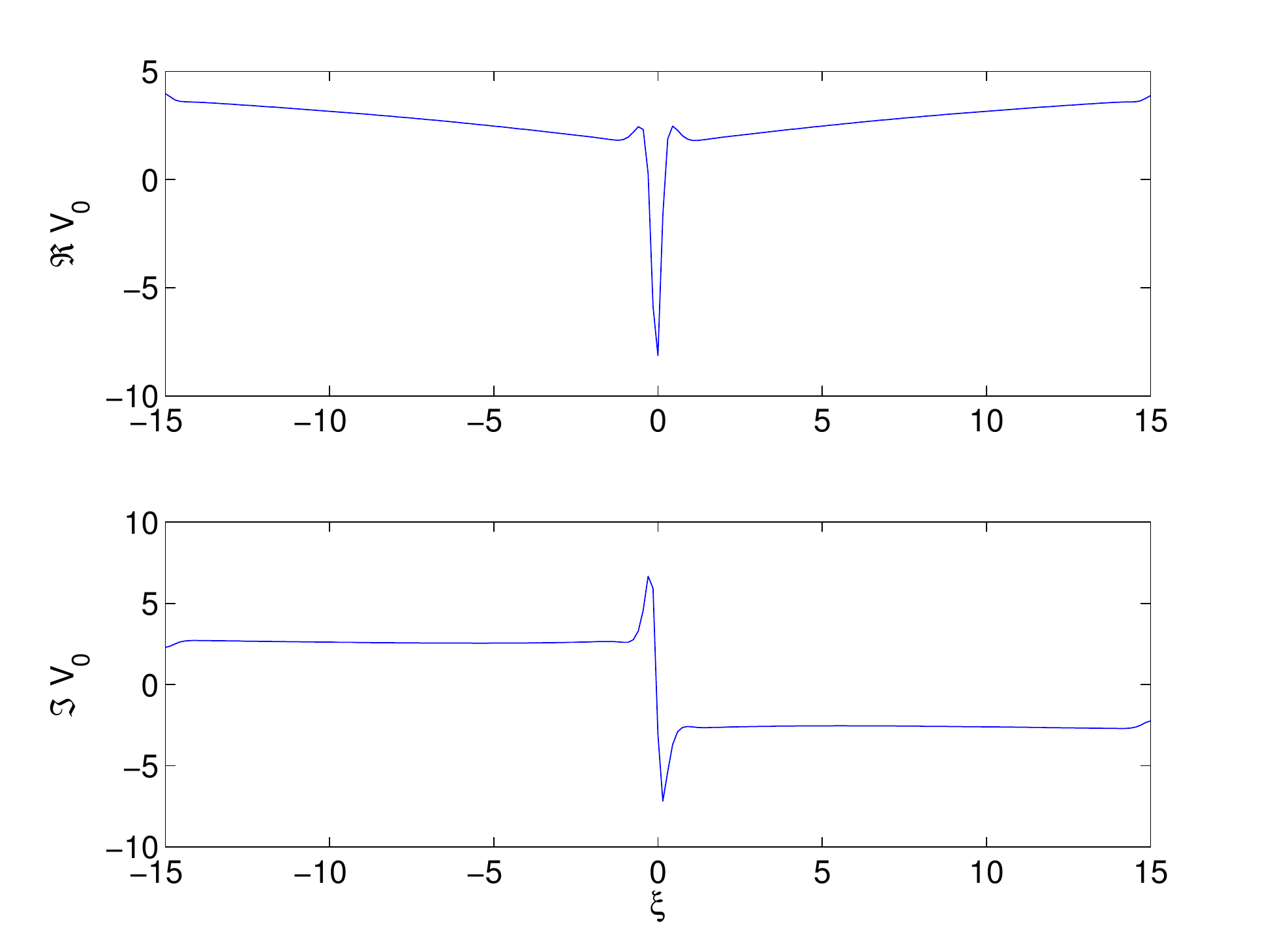}
\end{center}
 \caption{Solution $V_0(x,t)$ to \PItwo\ satisfying the asymptotic 
condition (\ref{asym}) for $t=0$ and $x=\exp(i\phi)\xi+b$, 
$\xi\in\mathbb{R}$ and $\phi=5\pi/2$ for $b=5.5$.}
\label{pI2triphi0b55t0}
\end{figure}

If one changes the angle $\phi$ for $b=0$, the solution will become 
trigonometric on the lines with 
$\phi=15\pi/7+3/7\arctan(1/\sqrt{5})\sim0.6290+2\pi$. This behavior is 
illustrated in Fig.~\ref{pI2triphi6932b0t0} for a slightly larger 
$\phi$. 
\begin{figure}[htb!]
\begin{center}
\includegraphics[width=0.7\textwidth]{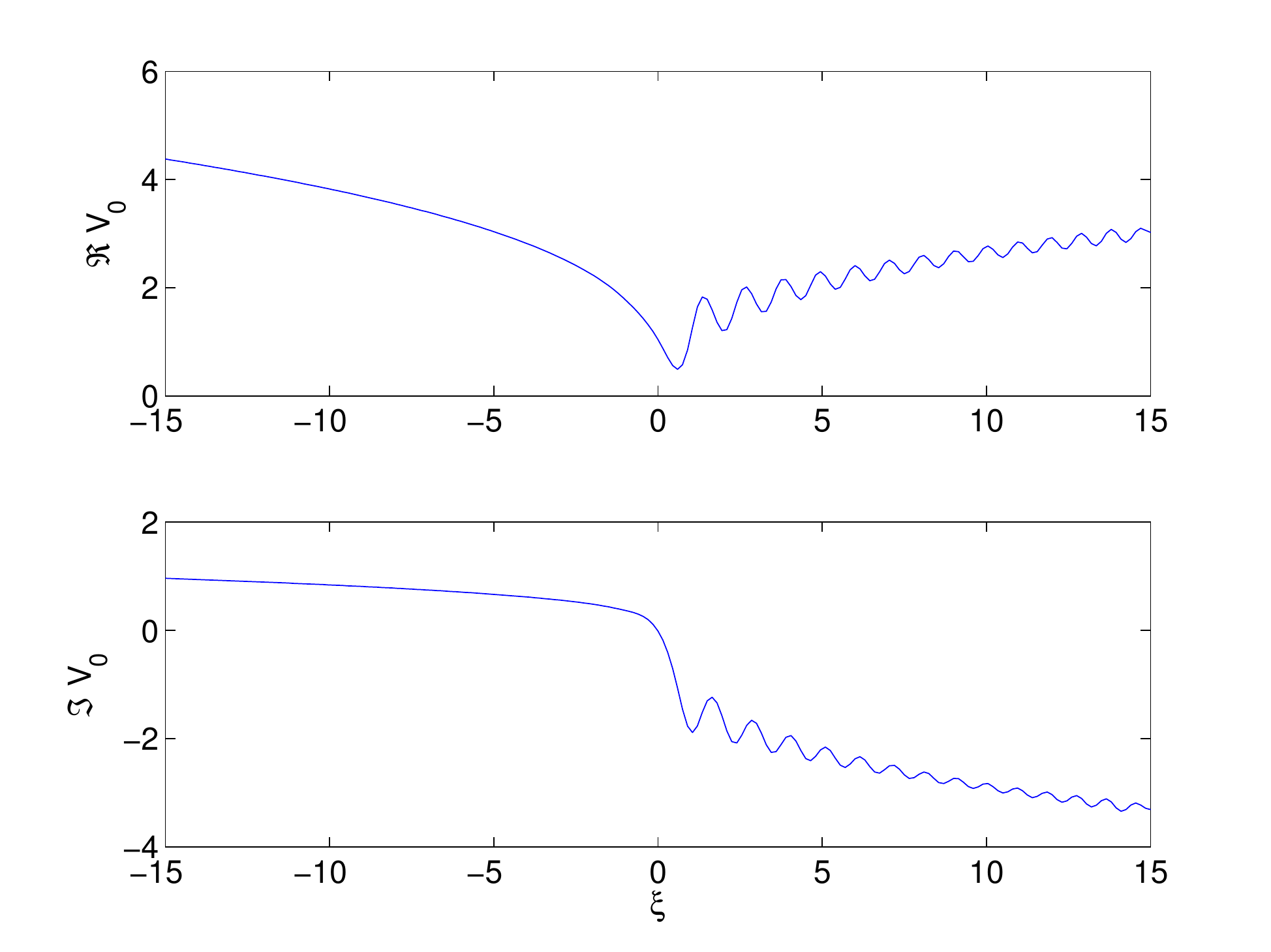}
\end{center}
 \caption{Solution $V_0(x,t)$ to \PItwo\ satisfying the asymptotic 
condition (\ref{asym}) for $t=0$ and $x=\exp(i\phi)\xi$, 
$\xi\in\mathbb{R}$ and $\phi=2\pi+0.6932$.}
\label{pI2triphi6932b0t0}
\end{figure}

As in the case of the tritronqu\'ee solutions of \PI\ in 
\cite{DubrovinGravaKlein}, it is 
possible to study the \PItwo\ solution $V_{0}(x,t)$ with the condition (\ref{asym}) in 
the sector of the complex plane with $\mbox{arg}x\in[-\phi, \phi]$  in the cases 
shown in Figure~\ref{figPI2phit}. We analytically continue  the cubic 
root not to be branched in the shown sector close to the negative real axis. 
For $t=0$, one obtains Figure~\ref{PI2trireal_typeI} for the real part of the 
solution and  Figure~\ref{PI2triimag_typeI} for its imaginary part. It 
appears that there are no poles in the shown sectors even for 
finite $|x|$. 
\begin{figure}[htb!]
\begin{center}
    \includegraphics[width=0.7\textwidth]{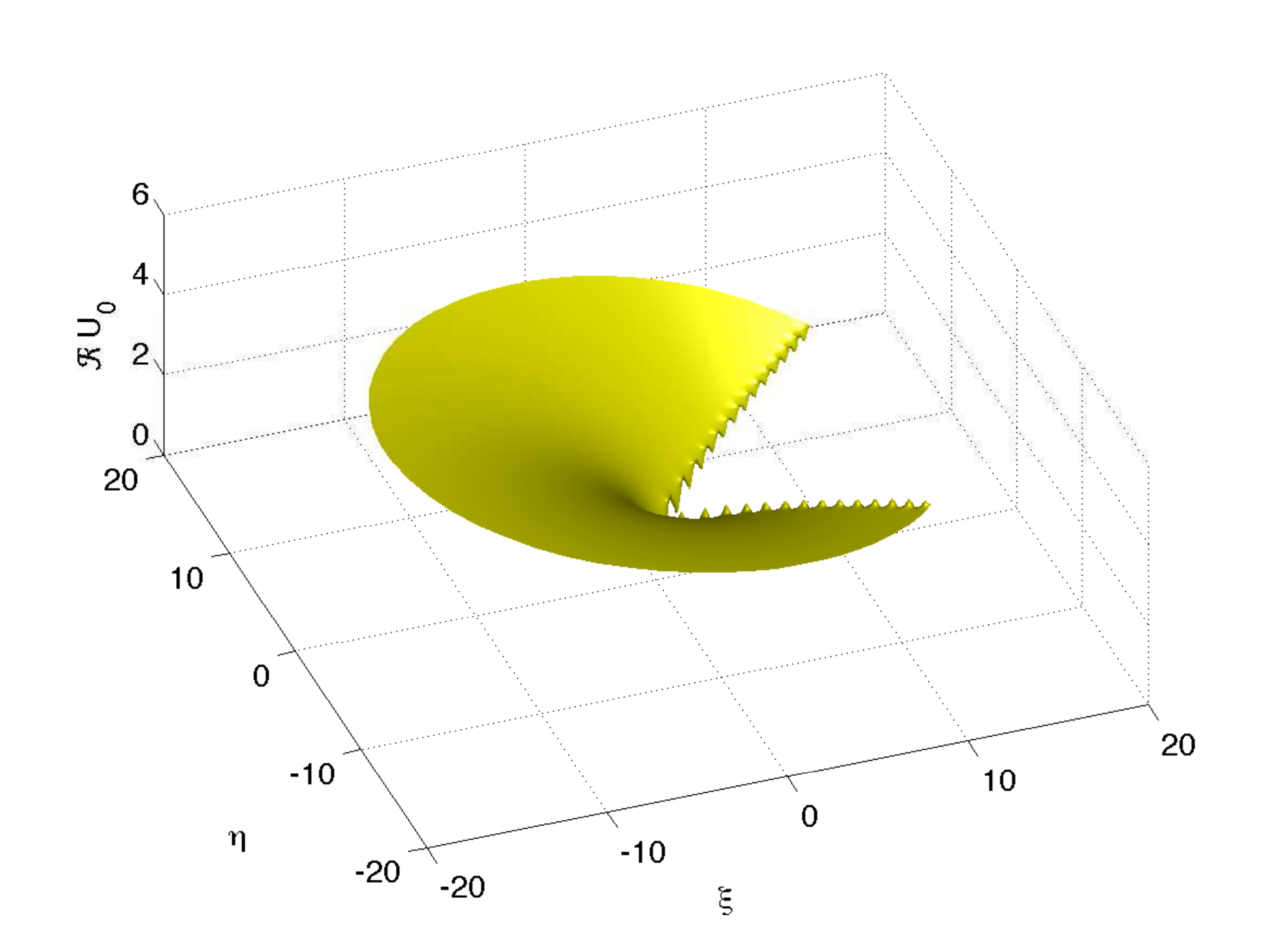}
\end{center}
 \caption{Real part of the type I tritronqu\'ee solution to the \PItwo\ equation (\ref{p12}) with 
 asymptotic condition (\ref{asym}) in the complex plane for $t=0$.}
   \label{PI2trireal_typeI}
\end{figure}
\begin{figure}[htb!]
\begin{center}
    \includegraphics[width=0.7\textwidth]{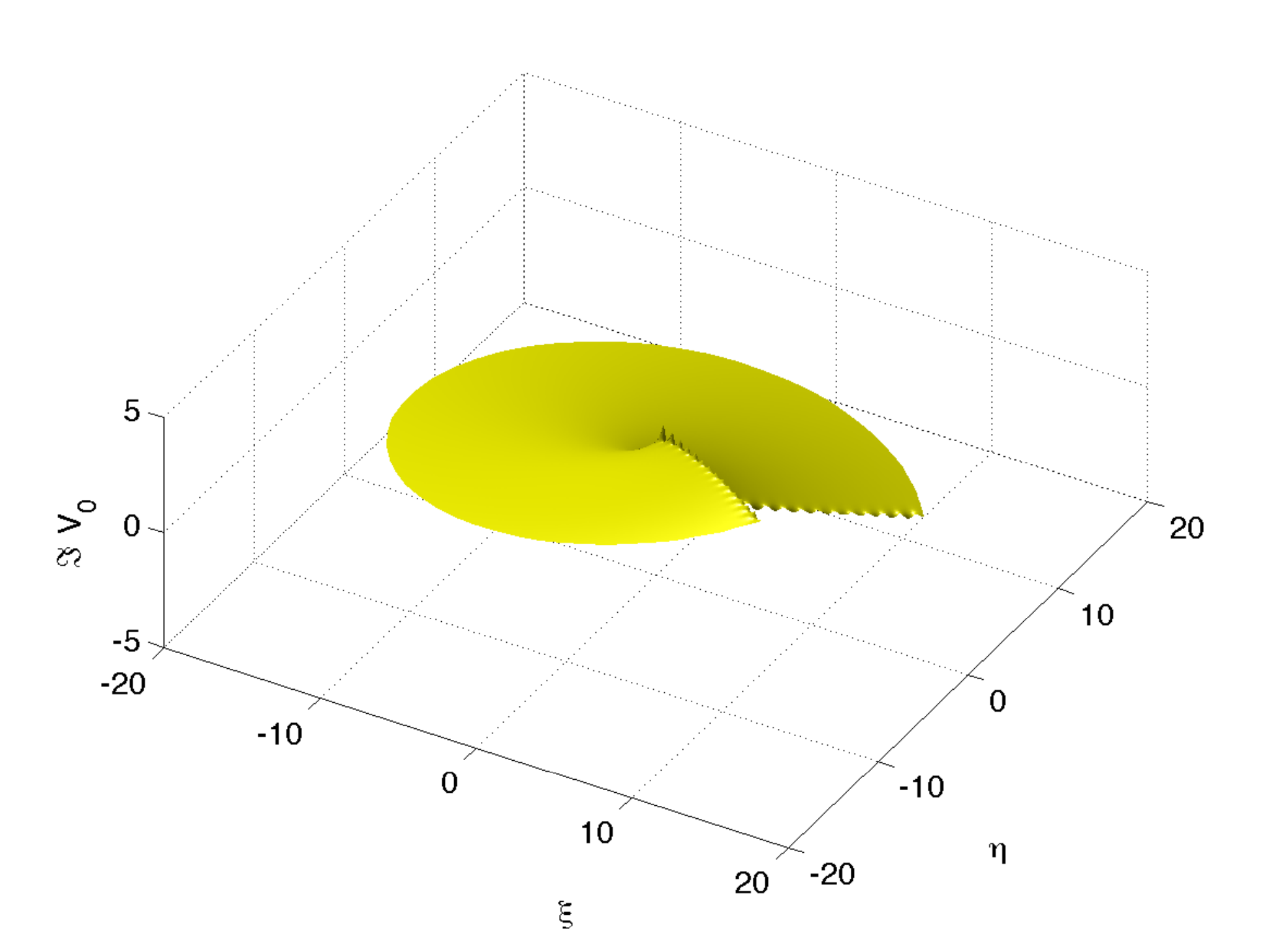}
\end{center}
 \caption{Imaginary part of the type I tritronqu\'ee solution to the \PItwo\ equation (\ref{p12}) with 
 asymptotic condition (\ref{asym}) in the complex plane for $t=0$.}
   \label{PI2triimag_typeI}
\end{figure}

\subsection{Type II tritronqu\'ee solutions}
The solution $U_0(x,t)$ characterized by the asymptotic 
behavior (\ref{asym}) is real and pole-free on the real line 
for all $t\in{\mathbb R}$ \cite{CV}. For large negative $t$, the solution
has no oscillations, but they appear for $t\sim0$, 
see Figure~\ref{figPI2t}. 
\begin{figure}[htb!]
\begin{center}
\includegraphics[width=0.8\textwidth]{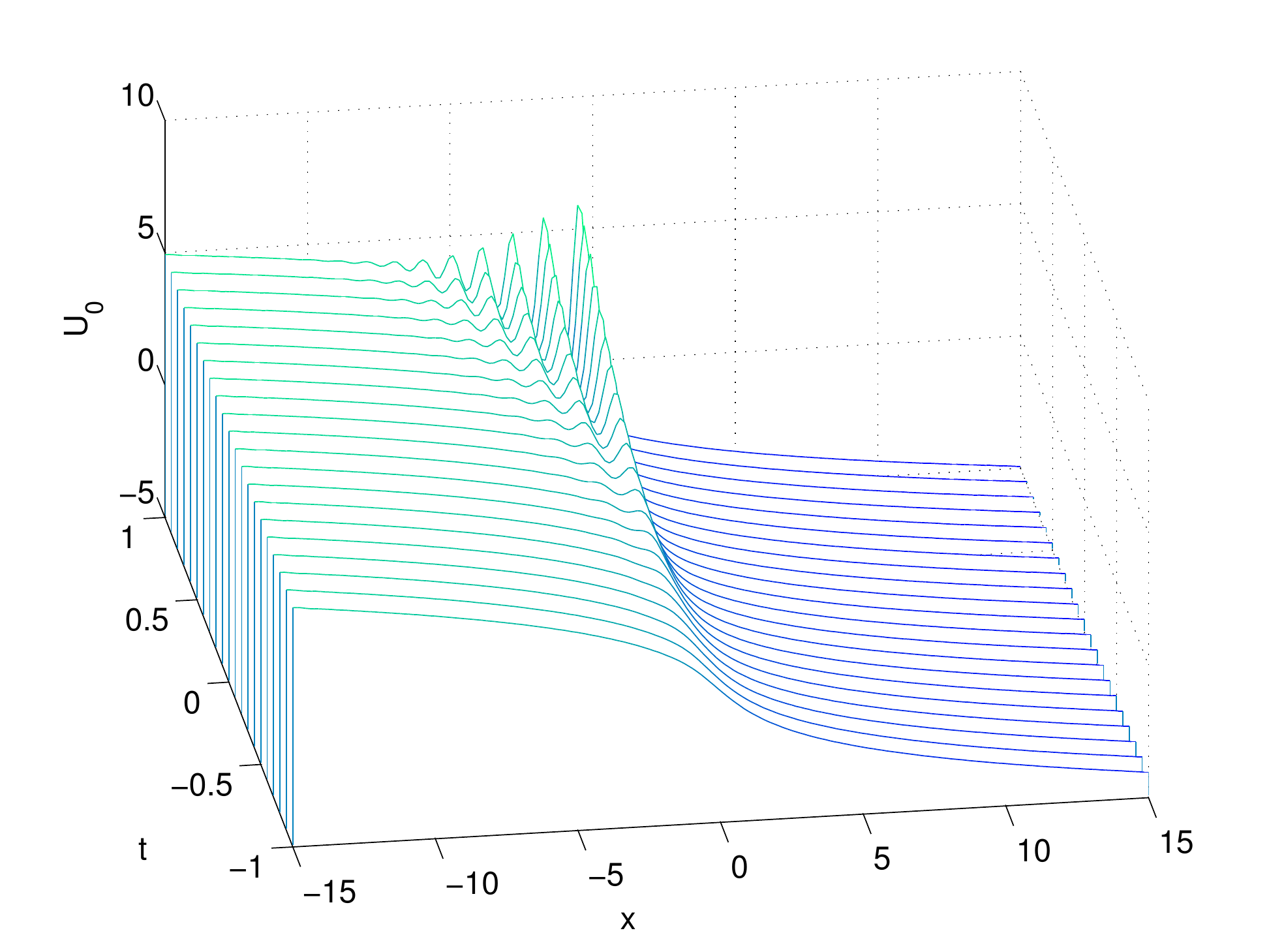}
\end{center}
 \caption{Solution $U_0(x,t)$ for $x\in\mathbb{R}$ to \PItwo\ satisfying the asymptotic 
condition (\ref{asym}) for several values of $t$.}
\label{figPI2t}
\end{figure}
For positive $t$, the oscillations rapidly develop, see
Figure~\ref{figPI2t2}.
\begin{figure}[htb!]
\begin{center}
\includegraphics[width=0.8\textwidth]{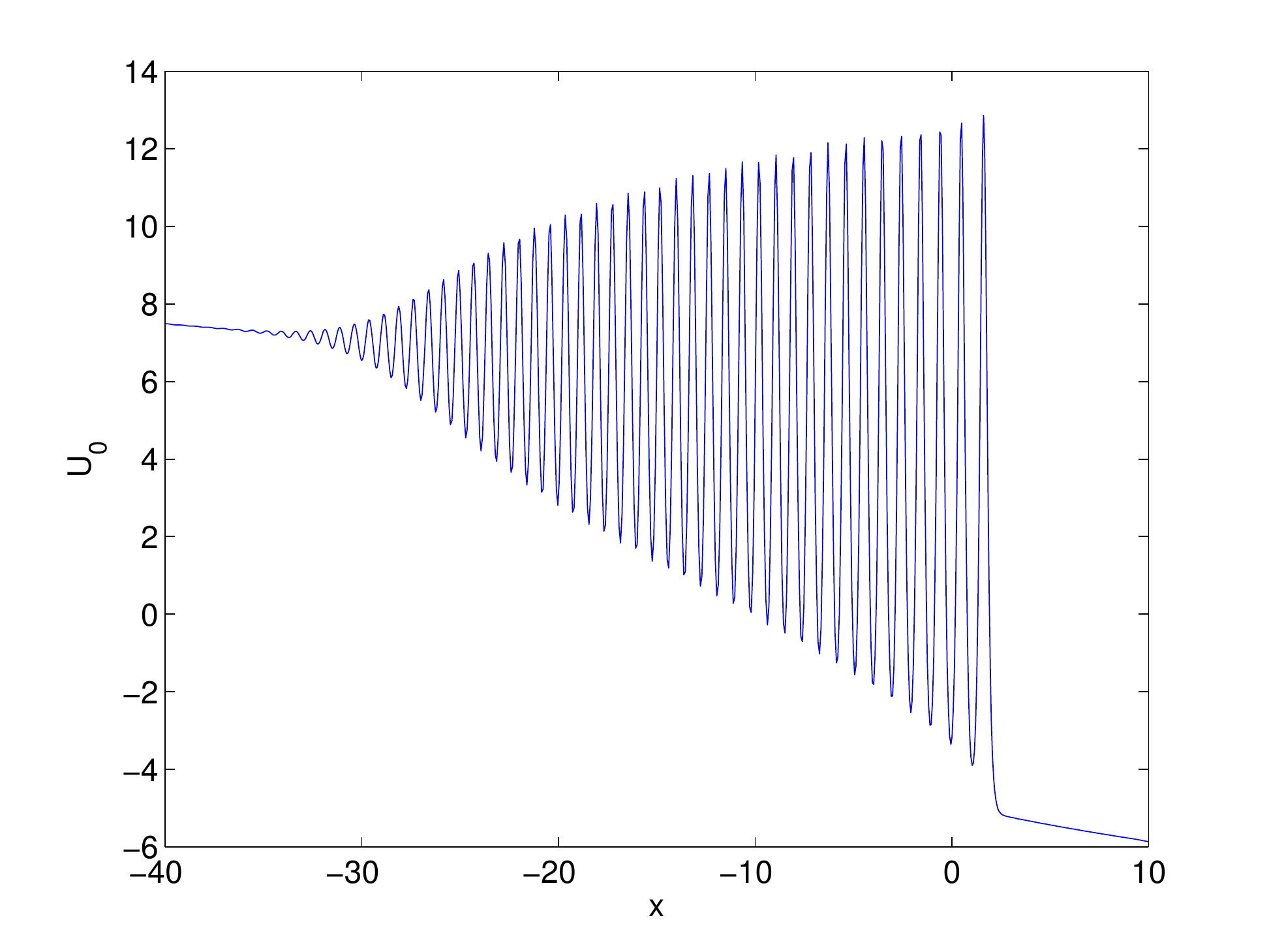}
\end{center}
\caption{Solution $U_0(x,t)$ to \PItwo \ for $t=4$.}
\label{figPI2t2}
\end{figure}

The location of the oscillations for $t>0$ in Figure~\ref{figPI2t} 
agrees with the theoretical prediction of the interval
$\bigl(-2\sqrt3\,t^{3/2},\tfrac{2\sqrt5}{9\sqrt3}t^{3/2}\bigr)$, 
see Figure~\ref{fig14} which belongs to the sector of the elliptic asymptotic 
behavior of $U_0(x,t)$. This indicates the presence of poles in 
the complex plane approaching the real axis \cite{Cl},\cite{GaSul}. 

To test this we solve 
\PItwo\ with the asymptotic 
condition (\ref{asym}) on a line parallel to the real axis parameterized by 
$x=\xi+b$ with $\xi\in\mathbb{R}$ and $b\in i\mathbb{R}$. We vary $b$ 
gradually from zero and stop when observing a strong increase in 
the absolute value of the solution. This can be seen in Figure~\ref{figPI2bt} 
for two values of $t$. For positive $t$, the presumed singularity moves 
very close to the real axis. 
\begin{figure}[htb!]
\begin{center}
    \includegraphics[width=0.45\textwidth]{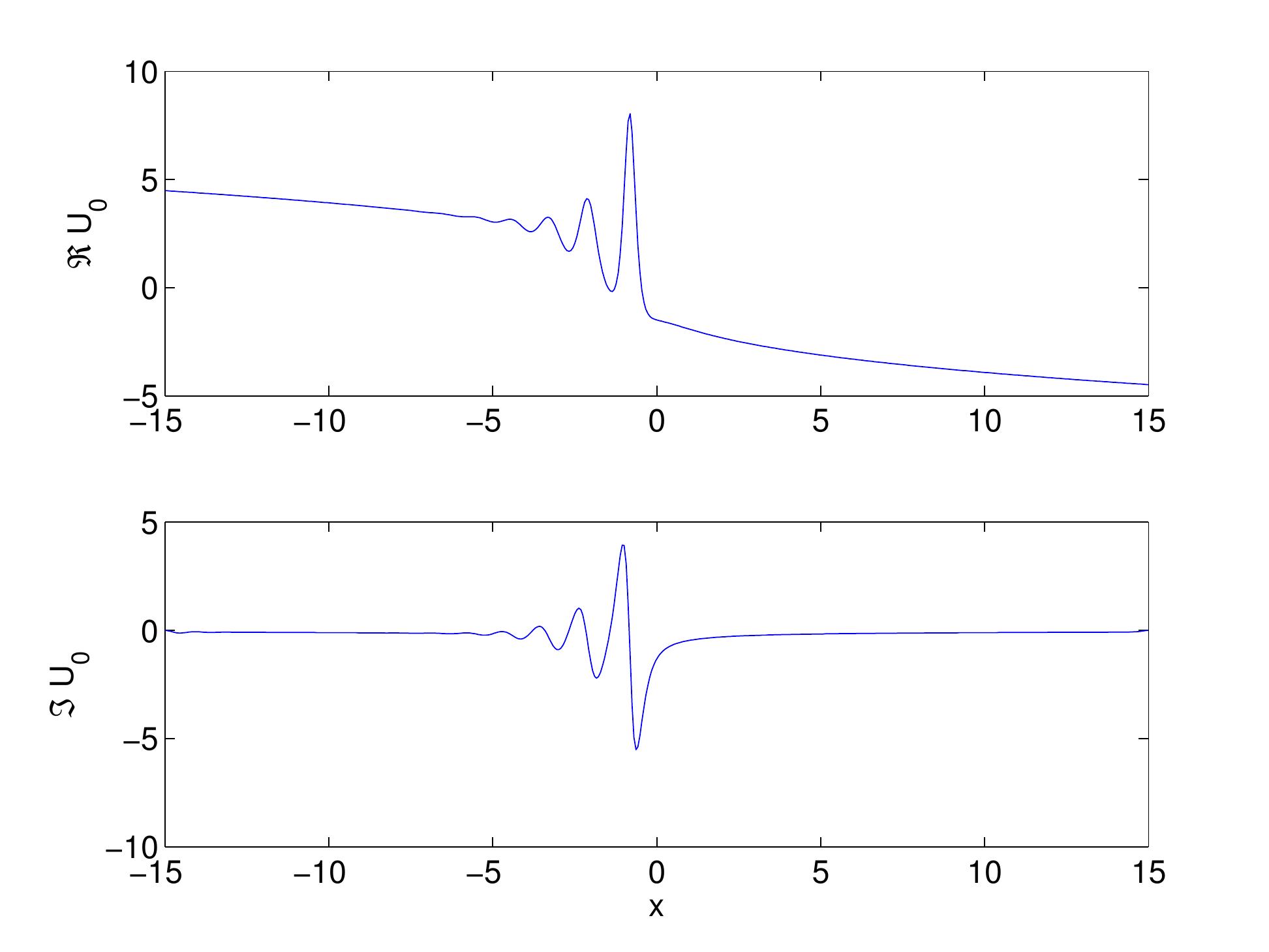}
    \includegraphics[width=0.45\textwidth]{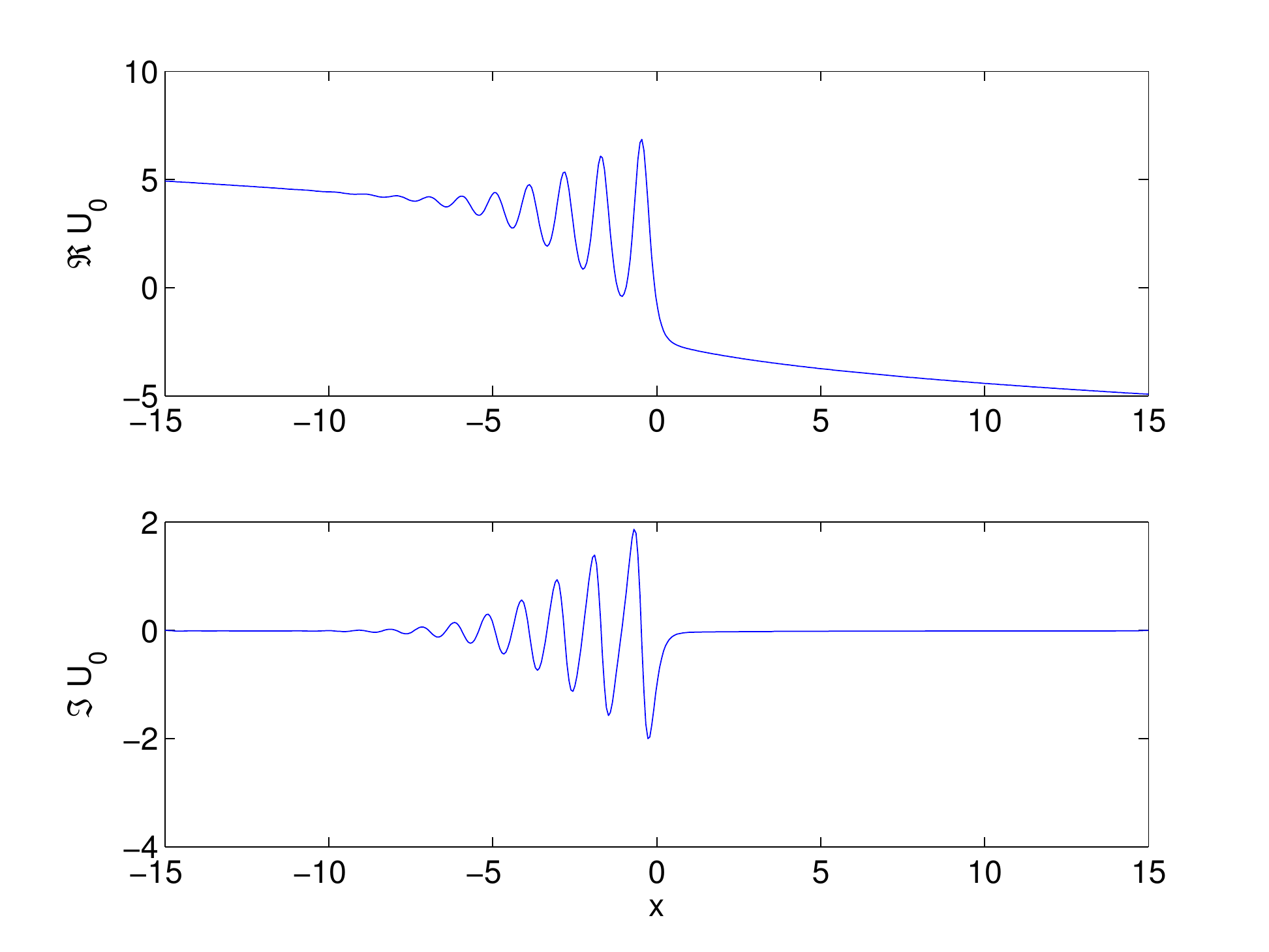}
\end{center}
 \caption{Solution to the \PItwo\ equation (\ref{p12}) with 
 asymptotic condition (\ref{asym}) on a line parallel to the real 
 axis, $x=\xi+b$, $\xi\in\mathbb{R}$; on the left the case $t=0$ and $b=0.8i$, on the 
 right $t=1$ and $b=0.1i$. We are plotting the maximum values of $|b|$ for which the numerical code does not break.
   In the second case $|b|$ is smaller because the poles of the solution are approaching the real axis as $t$ is getting bigger.}
   \label{figPI2bt}
\end{figure}

In the same way it is possible to study \PItwo\ (\ref{p12}) with the 
condition (\ref{asym}) on lines in the complex plane given by
$x=e^{i\phi}\xi$ with $\xi,\phi\in\mathbb{R}$ with the goal to 
identify numerically the sectors in the complex plane, 
where the solution has no poles.
 
In Figure~\ref{figPI2phit} we show the solution for $t=0$ close to the 
real axis for $\phi=0.15$. It can be seen that the amplitude of the 
oscillations decreases slower than on the real axis. We expect to observe
the trigonometric behavior of the solution on the line with
$\phi=3/7\arctan(1/\sqrt{5})\sim0.1802$ for $\xi<0$. 
The closer one comes to this line, the less reliable the numerical 
solution is since the asymptotic series yielding after truncation 
the boundary data 
converges more and more slowly for a finite value of $|x|$. 
\begin{figure}[htb!]
\begin{center}
    \includegraphics[width=0.7\textwidth]{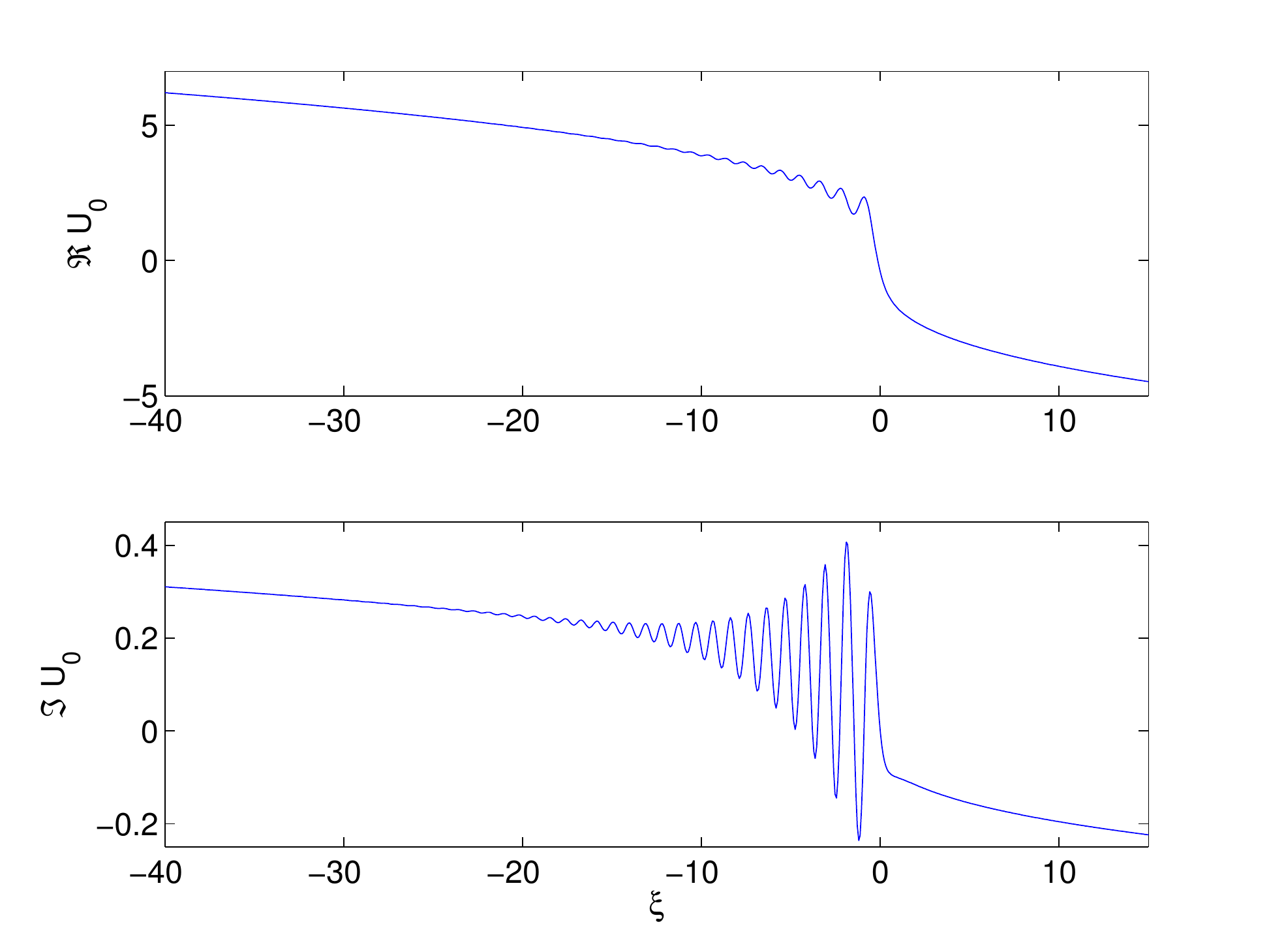}
\end{center}
\caption{Solution to the \PItwo\ equation (\ref{p12}) with 
asymptotic condition (\ref{asym}) on a line in the complex plane 
given by $x=e^{i\phi}\xi$ with $\xi,\phi\in\mathbb{R}$ for  
$t=0$ and $\phi=0.15$.}
\label{figPI2phit}
\end{figure}

As discussed in the previous sections, the regular sector is 
considerably larger in the vicinity of the positive real axis. 
In Figure~\ref{figPI2phitp} we show the solution on the half line 
$x=\xi e^{i\phi}$ for $\phi=1.5$ and $\xi>0$. As predicted the 
solution shows oscillations with asymptotically decreasing amplitude 
close to the line with  $\phi=3\pi/7+3/7\arctan(1/\sqrt{5})\sim1.5266$,
the dependence of $U_0(x,t)$ on $|x|$ becomes trigonometric. 
We are able to reach the value $\phi=1.5$. The regular sectors (angles between $\pi-0.18$ and $\pi+.18$ near the negative real axis and between $-1.5$ and $1.5$ near the positive real axis) are illustrated in 
Figures~\ref{figPI2phi} and~\ref{figPI2phiim} below. The oscillations near the boundaries of these sectors can be clearly recognized and are even more pronounced for the imaginary part.
\begin{figure}[htb!]
\begin{center}
    \includegraphics[width=0.6\textwidth]{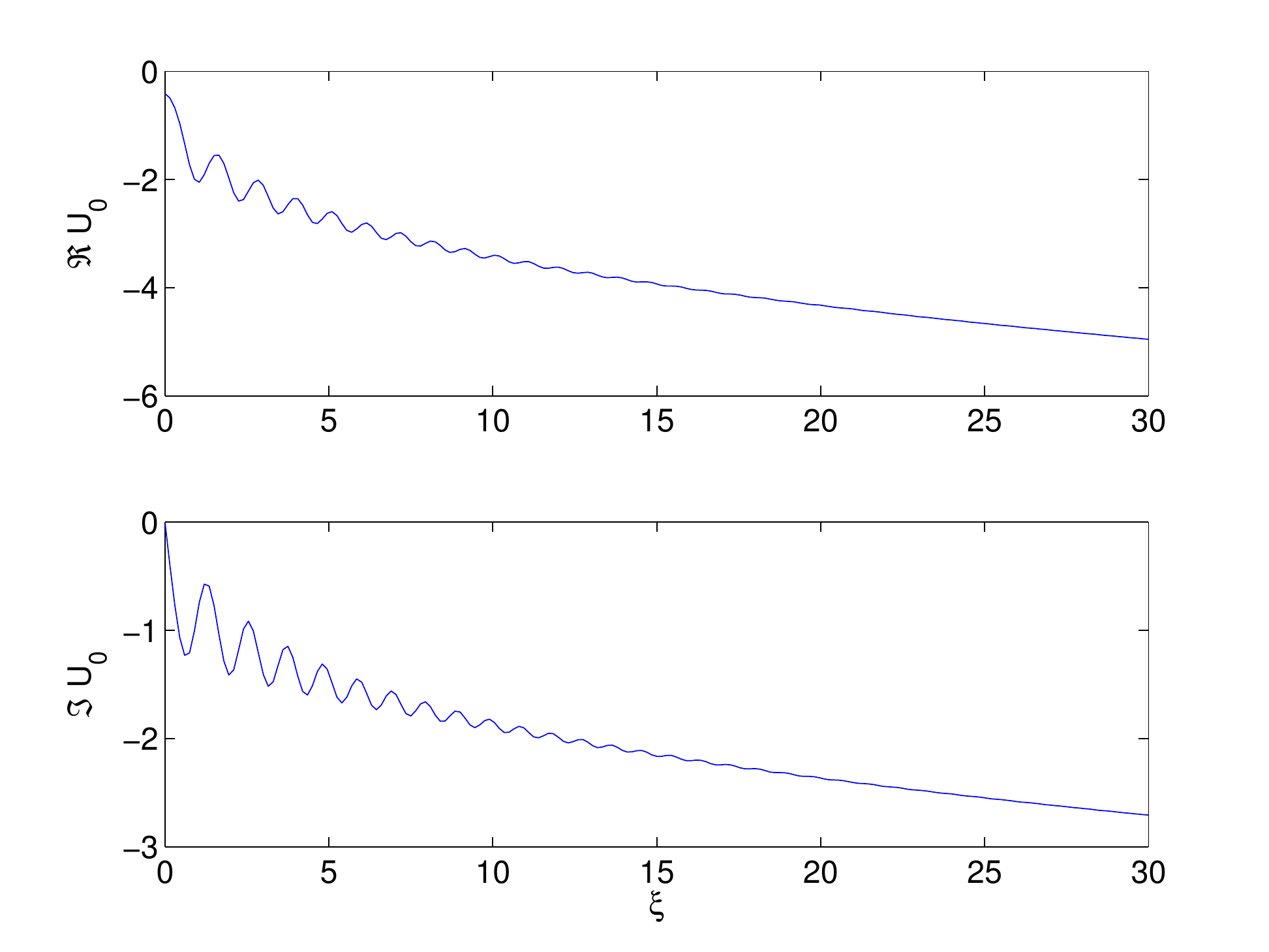}
\end{center}
 \caption{Solution to the PI2 equation (\ref{p12}) with 
 asymptotic condition (\ref{asym}) on a halfline in the complex plane 
 given by $x=e^{i\phi}\xi$ with $\xi>0$ for $t=0$ 
 and $\phi=1.5$.}
   \label{figPI2phitp}
\end{figure}

We again study the \PItwo\ solution with the condition (\ref{asym}) in 
the sectors of the complex plane between $-\phi$ and $\phi$ for the cases 
shown in Figure~\ref{figPI2phit}. We analytically continue  the cubic 
root not to be branched in the shown sectors close to the real axis. 
For $t=0$, one obtains Figure~\ref{figPI2phi} for the real part of the 
solution and  Figure~\ref{figPI2phiim} for the imaginary part. It 
appears that there are no poles in the shown sectors even for 
finite $|x|$. 
\begin{figure}[htb!]
\begin{center}
    \includegraphics[width=0.7\textwidth]{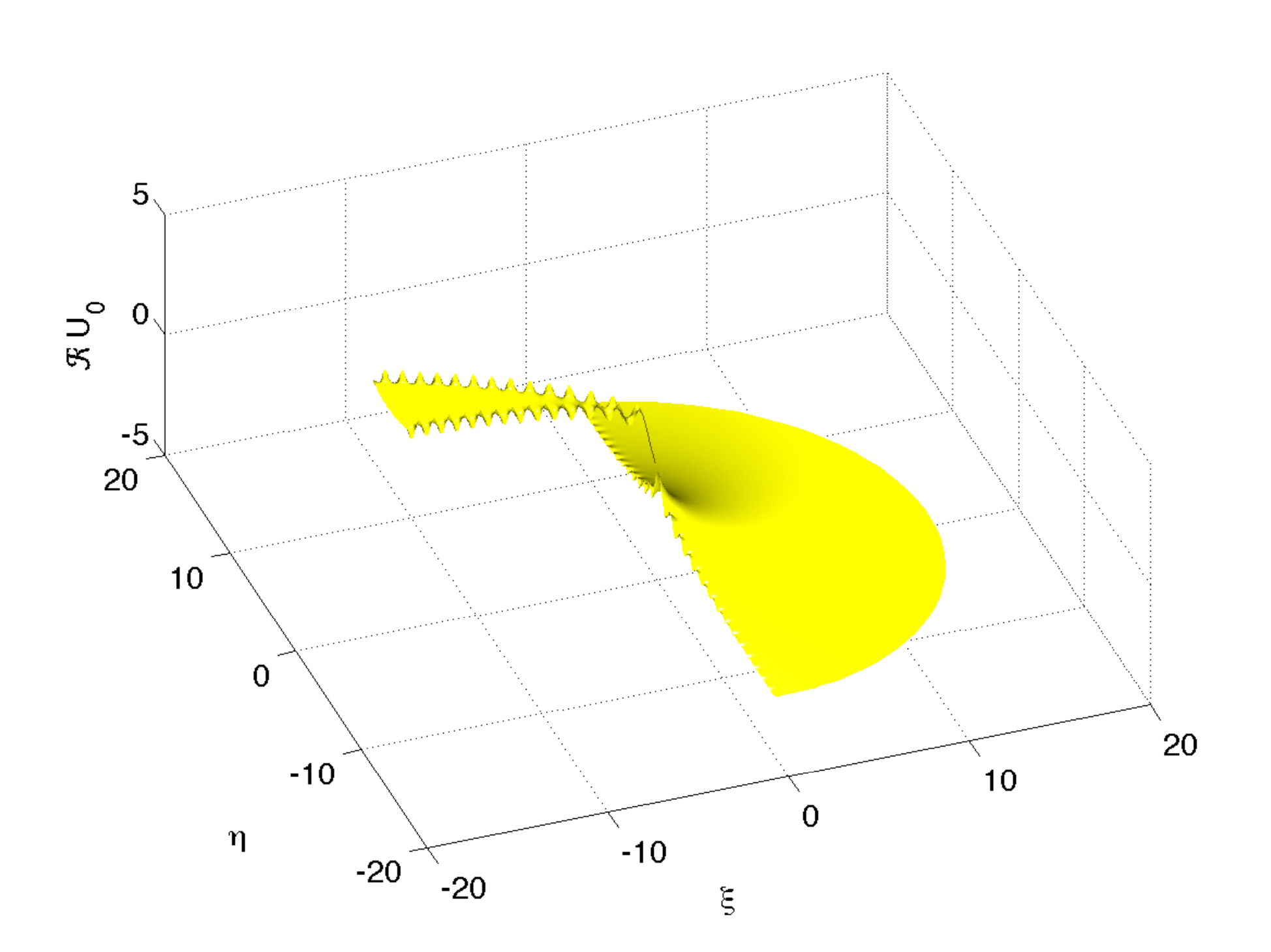}
\end{center}
 \caption{Real part of the solution to the \PItwo\ equation (\ref{p12}) with 
 asymptotic condition (\ref{asym}) in the complex plane for $t=0$.}
   \label{figPI2phi}
\end{figure}
\begin{figure}[htb!]
\begin{center}
    \includegraphics[width=0.7\textwidth]{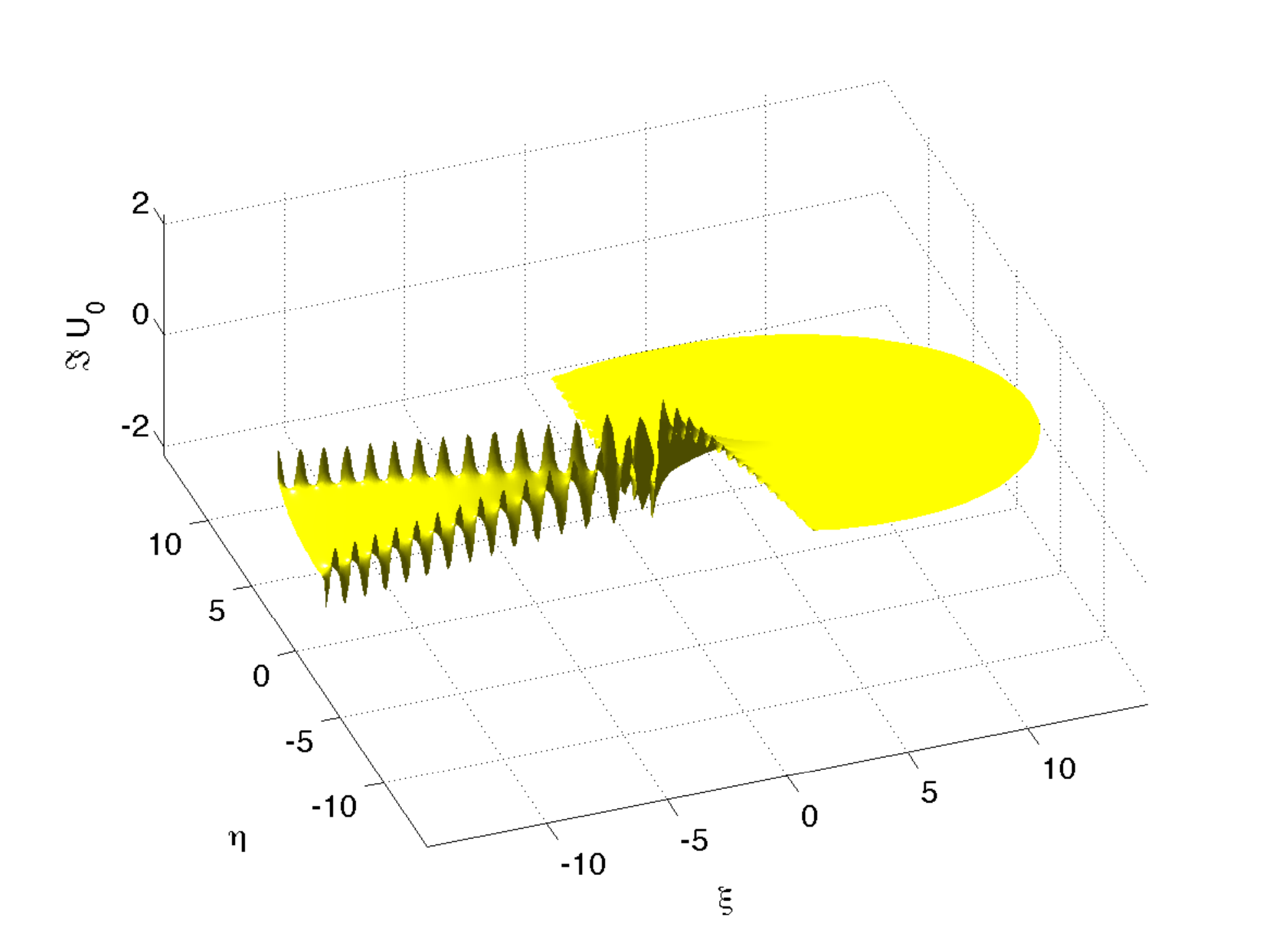}
\end{center}
 \caption{Imaginary part of the solution to the \PItwo\ equation (\ref{p12}) with 
 asymptotic condition (\ref{asym}) in the complex plane for $t=0$.}
   \label{figPI2phiim}
\end{figure}

It was shown in the previous sections, see e.g.\ (\ref{U0_sector}) 
and (\ref{V0_sector}) that the difference type I and type II 
solutions  is exponentially small  on the negative real axis. This 
can be seen in Fig.~\ref{figdiff}. The type II solution shows oscillations close to the origin which are not present for the type I solution, but then they agree in a way that no difference can be seen. In the left part of Fig.~\ref{figdiff} we therefore show the logarithmic plot of the absolute value of the difference between both solutions which decreases as expected exponentially with increasing $-x$, i.e., lineary in a logarithmic plot until the difference is of the order of the rounding error.
\begin{figure}[htb!]
\begin{center}
    \includegraphics[width=0.45\textwidth]{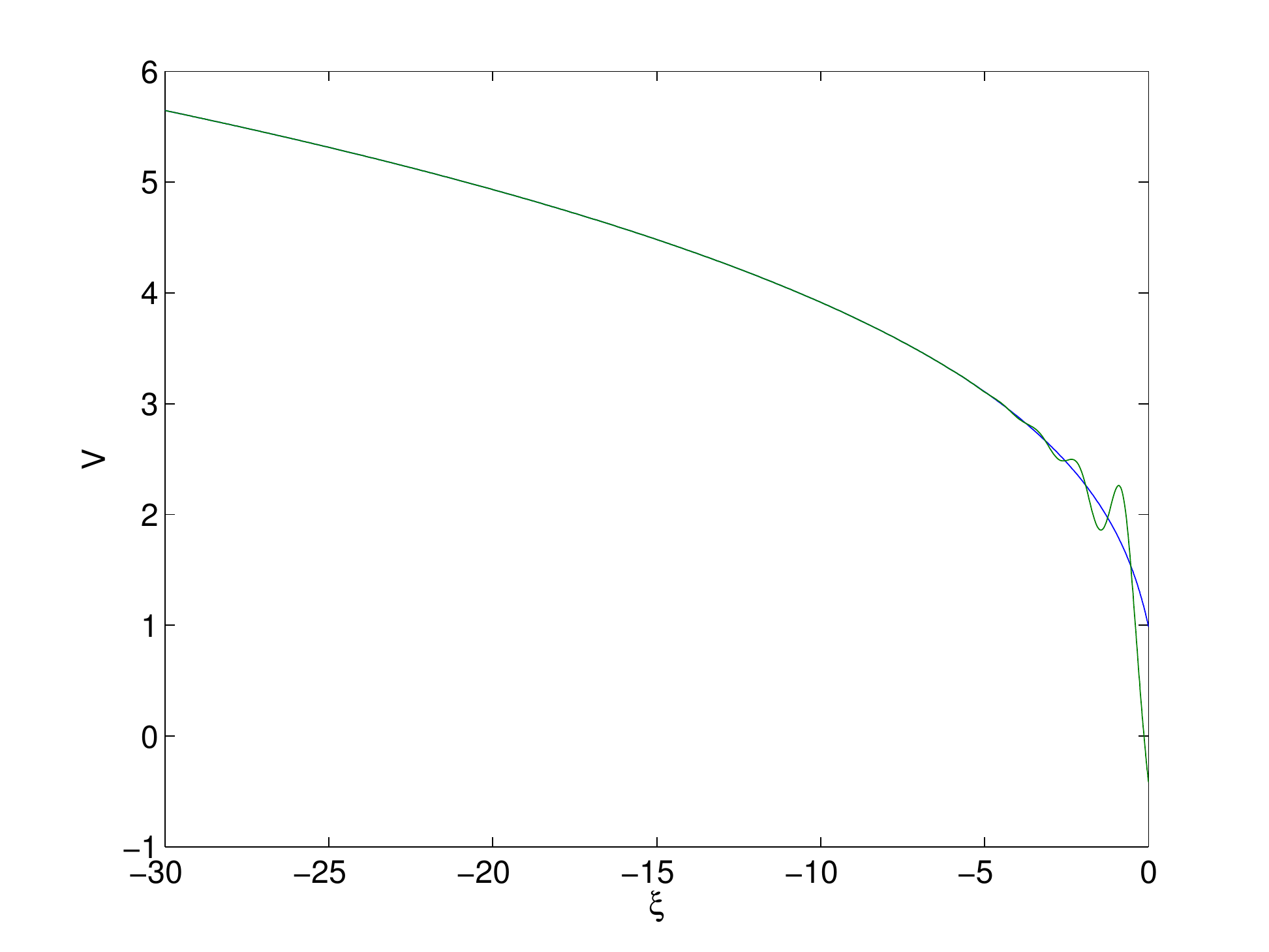}
    \includegraphics[width=0.45\textwidth]{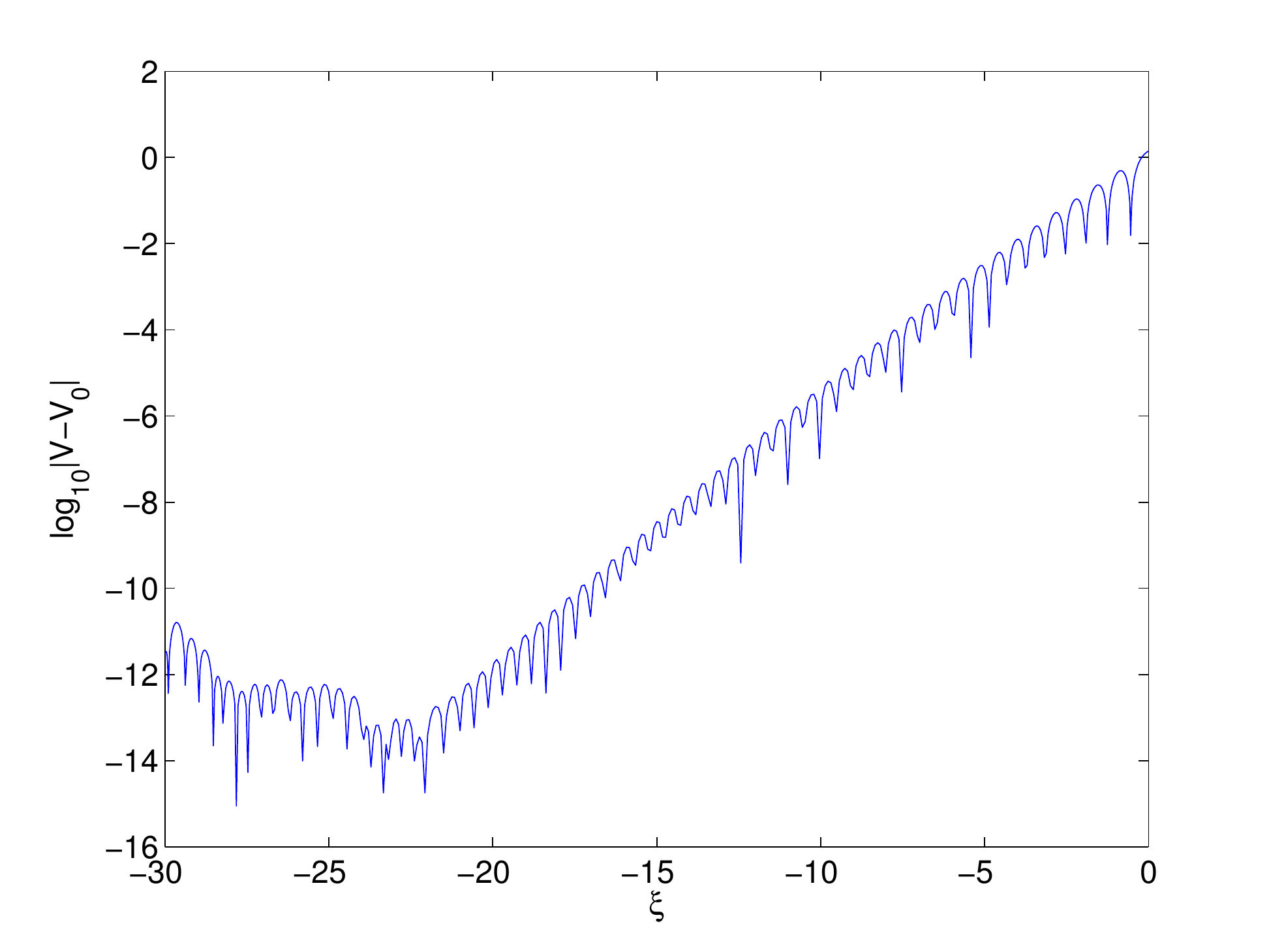}
\end{center}
 \caption{Type I and type II solution to the \PItwo\ equation (\ref{p12}) with 
 asymptotic condition (\ref{asym}) on the negative real axis for 
 $t=0$; on the left the type I solution in blue and the type II solution in green, on the 
 right the difference between both.}
   \label{figdiff}
\end{figure}

\section*{Acknowledgements} 
The work of A.K. was partially supported by 
the project SPbGU N 11.38.215.2014. He also thanks the ERC grant FroMPDEs for the support during his stay 
in  SISSA  when the part of the work  was done.\\

TG  was partially supported by  PRIN  Grant ÒGeometric and analytic theory of Hamiltonian systems in finite and infinite dimensionsÓ of Italian Ministry of Universities and Researches and by the FP7 IRSES grant RIMMP ÒRandom and Integrable Models in Mathematical PhysicsÓ. 

\bibliographystyle{plain}
\ifx\undefined\bysame
\newcommand{\bysame}{\leavevmode\hbox to3em{\hrulefill}\,}
\fi

\end{document}